\documentclass[twoside]{article}

\usepackage[preprint]{aistats2027}
\usepackage[round]{natbib}

\usepackage{amsthm, amssymb}
\usepackage{enumitem}
\usepackage{algorithm}
\usepackage{algpseudocode}
\usepackage{algorithmicx}
\usepackage{graphicx}
\usepackage{booktabs}
\usepackage{subcaption}
\usepackage{placeins}

\DeclareSymbolFont{bbold}{U}{bbold}{m}{n}
\DeclareSymbolFontAlphabet{\mathbbold}{bbold}

\renewcommand{\topfraction}{0.9}

\renewcommand{\textfraction}{0.1}
\renewcommand{\floatpagefraction}{0.8}

\newtheorem{theorem}{Theorem}[section]
\newtheorem{lemma}[theorem]{Lemma}
\newtheorem{prop}[theorem]{Proposition}

\newtheorem{remark}[theorem]{Remark}
\newtheorem{assumption}{Assumption}[section]

\newtheorem{definition}[theorem]{Definition}

\begin{document}

%

%

\twocolumn[

\aistatstitle{Collusion-Resistant Auction Design using Side Information}

\aistatsauthor{ Sukanya Kudva \And Edward Dowling \And Anil Aswani }

\aistatsaddress{ University of California, Berkeley } ]

\begin{abstract}
 Bidder collusion degrades welfare and fairness in standard auctions. To address this, we introduce the Conditional-Posted Price (C-PoP) mechanism, a collusion-resilient design that uses machine learning classifiers to predict, even imprecisely, which bidders are colluding. C-PoP computes Posted Prices conditioned on predicted non-colluders' bids, ensuring truthfulness under perfect classification, thereby improving fairness toward non-colluders. Under imperfect classification, we bound overall and non-colluder welfare to establish the mechanism's robustness and fairness, respectively. Crucially, these bounds provide insights for classifier design, revealing that false negatives (misclassifying colluders as non-colluders) are preferable to false positives (misclassifying non-colluders as colluders). Numerical experiments show that C-PoP achieves strong performance in welfare, revenue, and fairness under collusion, even when using simple, low-cost classifiers. 
\end{abstract}


\section{Introduction}

Standard auction mechanisms are vulnerable to bidder collusion. For instance, consider the Vickrey-Clarke-Groves (VCG) mechanism \citep{Vickrey_1961,Groves_1973}. Though it is dominant-strategy incentive-compatible (DSIC), meaning truthful reporting is optimal regardless of other bidders' strategies, and maximizes the auctioneer's revenue among all welfare-maximizing mechanisms \citep{Krishna_Perry_1997}, it remains vulnerable to collusion in both theory and practice \citep{friedman1964comment, Robinson_1985, hendricks1989collusion, goswami1996collusion}. Posted Price mechanisms offer bidders fixed, take-it-or-leave-it prices \citep{goldberg2004collusion-resistant}; while they are collusion-proof when selling infinitely many items, they are vulnerable to collusion when applied to finite numbers of items. Posted Price can also perform poorly with respect to social welfare \citep{ blumrosen2008posted}.

Despite the known vulnerability of standard auctions to collusion, designing mechanisms that preserve high welfare and revenue under collusive behavior remains an enduring challenge. To address this, we leverage machine learning classifiers to predict which bidders are colluding, even when such predictions are imperfect. Our goal is to design uniform-price mechanisms that guarantee high welfare, revenue, and fairness under imperfect classification. Specifically, we construct mechanisms that achieve dominant-strategy incentive compatibility (DSIC) for all bidders under perfect classification, hypothesizing that this structural property imparts robustness against prediction errors. We validate this hypothesis by establishing theoretical lower bounds on performance under imperfect classification. 

\subsection{Mechanism Design with Side Information} 


Recent literature explores mechanism design using side information to overcome the limitations of standard mechanisms. Prior work primarily focuses on side information regarding bidder types, such as preferences or item valuations \citep{DevanurSideInfo, Xu_Lu_2022, SPrasadVCG}, leveraging correlated types to improve VCG efficiency \citep{SPrasadVCG, Xu_Lu_2022} or distinguish bidders along a continuum in revenue-maximizing auctions \citep{DevanurSideInfo}. Beyond traditional settings, learning-augmented algorithms use side information to guarantee strategy-proofness in domains such as job scheduling \citep{Balkanski_Gkatzelis_Tan_2022} and facility location \citep{Facilitysideinfo, tang_characterization_2020}. We make two contributions to this literature. First, we provide the first framework incorporating side information on bidder collusion. Second, we expand learning-augmented mechanism design into the setting of auctions by accounting for imperfect predictions from machine learning classifiers.

\subsection{Collusion-Resistant Mechanisms}

Designing auction mechanisms robust to bidder collusion remains a major challenge. Under the strongest collusion model, colluding bidders act as a single utility-maximizing entity with full freedom to exchange items and side payments. Because mechanism design under this strict benchmark is notoriously difficult, prior literature frequently relaxes the model. Common relaxations include $t$-truthfulness, where coalitions of size at most $t$ cannot profit via internal transfers \citep{goldberg2004collusion-resistant}, and group strategyproofness, where no coalition can strictly benefit a member without making another worse off \citep{goldberg2004collusion-resistant,Juarez_2013,Basu_Mukherjee_2024}. These relaxed notions have been applied across auctions \citep{goldberg2004collusion-resistant}, network design \citep{Pal_Tardos_2003,Cheng_Yu_Zhang_2013,Tang_Yu_Zhao_2020}, and queueing models \citep{Kayı_Ramaekers_2010,Mitra_Mutuswami_2011}. In contrast to work relying on relaxed collusion models, our paper directly addresses the strongest collusion benchmark by incorporating side information into the mechanism design.

\subsection{Collusion Detection via Learning}

Extensive research has addressed collusion detection using statistical and machine learning techniques, particularly within public procurement auctions \citep{Conley_Decarolis_2016,Chassang_Kawai_Nakabayashi_Ortner_2022, Lyra_Damásio_Pinheiro_Bacao_2022}. For instance, \cite{Garcia_Rodríguez_Rodríguez-Montequín_Ballesteros-Pérez_Love_Signor_2022} provide a comparative analysis of machine learning approaches, while \cite{pmlr-v180-bonjour22a} propose a game-theoretic framework based on mutual information. More recently, the focus has expanded to tacit algorithmic collusion \citep{Calvano_Calzolari_Denicolo_Pastorello_2018, Assad_Clark_Ershov_Xu_2024}, prompting the development of novel detection methods based on statistical testing \citep{hespanhol2020hypothesis, Hartline_Long_Zhang_2024}. Building on this extensive and growing body of literature, it is now reasonable to assume that an auctioneer can gain access to a machine learning classifier capable of identifying colluders in an auction. Practical scenarios that support the use of such classifiers involve repeated auctions with available contextual bidder data. Our work is the first to leverage such classifiers to design inherently collusion-resistant auction mechanisms.

\subsection{Contributions and Outline}

Section \ref{sec:prelim} introduces the multi-unit auction setting, outlines the bidder model, and specifies a collusion model where colluders maximize joint utility, make side payments, exchange items, and possess full visibility into all submitted bids.

Section \ref{sec:vcgcoll} clarifies how collusion degrades the performance of standard auction mechanisms. We derive bounds showing that collusion impairs welfare in VCG auctions and reduces non-colluder welfare in Posted Price mechanisms. Furthermore, we characterize the optimal collusion strategies for both mechanism types.

Section \ref{sec:tapc} analyzes our Conditional-Posted Price (C-PoP) auction mechanism, which leverages a classifier to label bidders as colluders or non-colluders. Under perfect classification, C-PoP is DSIC for all participants. When classification is imperfect, C-PoP remains DSIC for non-colluders, and we derive bounds on welfare and non-colluder welfare that establish theoretical robustness. Crucially, these bounds yield a key classifier design principle: false negatives (misclassifying colluders as non-colluders) are less damaging than false positives (misclassifying non-colluders as colluders). Finally, we present simplified variants of C-PoP.

Section \ref{sec:ne} presents numerical experiments comparing C-PoP (driven by classifiers of varying accuracy) against VCG and Posted Price mechanisms under collusion. Across social welfare, non-colluder welfare, and revenue, C-PoP exhibits consistently robust performance compared to standard baselines. These results also corroborate our theoretical analysis on the asymmetric impact of false positives versus false negatives.

\section{Auction and Bidder Model}
\label{sec:prelim}

We consider a multi-unit auction where $r$ identical items are offered to $m$ unit-demand bidders. Each bidder $i \in \mathsf{M}$ (with $|\mathsf{M}| = m$) has a private valuation $v_i \ge 0$ for a single unit, submits a bid $b_i \ge 0$, and is charged payment $p_i$ upon winning. Let $\mathbf{b} = (b_1, \dots, b_m)$ and $\mathbf{v} = (v_1, \dots, v_m)$ denote the vectors of bids and true valuations, respectively. An auction mechanism consists of two components for each bidder $i$: an allocation rule $x_i(\mathbf{b}) \in \{0,1\}$ and a payment rule $p_i(\mathbf{b}) \in \mathbb{R}$. These rules are fixed prior to bid collection, and the design of an auction mechanism generally relates to the participants' utilities:

\begin{definition}[Bidder Utility]
The utility of bidder $i$ is $u_i(b_i, \mathbf{b}_{i-}) = v_i x_i(b_i, \mathbf{b}_{i-}) - p_i(b_i, \mathbf{b}_{i-})$, where $\mathbf{b}_{i-}$ are the submitted bids of all bidders except $i$.
\end{definition}
\begin{definition}[Revenue]
The auctioneer's utility is given by the revenue $\mathtt{Rev}(\mathbf{b}) = u_a(\mathbf{b}) = \sum_{i \in \mathsf{M}}p_i(\mathbf{b})$.
\end{definition}
Auctions that maximize welfare and revenue are referred to as \textit{efficient} and \textit{optimal} auctions, respectively. This paper focuses on efficiency while maximizing revenue to the extent possible.

\begin{definition}[Welfare]
Welfare $\mathtt{Wel}(\mathbf{b})= u_a(\mathbf{b}) + \sum_{i \in \mathsf{M}} u_i(\mathbf{b}) = \sum_{i \in \mathsf{M}} v_i x_i(\mathbf{b})$ is the total utility of all participants, meaning all bidders and the auctioneer.  \end{definition}

We assume bidders are rational utility-maximizers. The bidder set $\mathsf{M}$ is partitioned into disjoint subsets: non-colluding bidders $\mathsf{N}$ and colluding bidders $\mathsf{C}$, with $m = |\mathsf{M}| = |\mathsf{N}| + |\mathsf{C}| = n + c$. Let $\mathbf{b}^\mathsf{N}$ and $\mathbf{b}^\mathsf{C}$ be the bid vectors of non-colluding and colluding bidders, respectively. Valuations $\mathbf{v}^\mathsf{N}$ and $\mathbf{v}^\mathsf{C}$ are defined analagously.


\begin{definition}[Non-Colluding Bidders]
Each non-colluding bidder $i\in\mathsf{N}$ submits bid $b_i$ to maximize their utility $u_i(b_i, \mathbf{b}_{i-}) = v_i x_i(b_i, \mathbf{b}_{i-}) - p_i(b_i, \mathbf{b}_{i-})$.
\end{definition}
\begin{definition}[Non-Colluder Welfare]
Non-colluder Welfare $\mathtt{NWel}(\mathbf{b})= \sum_{i \in \mathsf{N}} v_i x_i(\mathbf{b})$ is defined to be the total valuation of all items received by non-colluders.  \end{definition}

\begin{definition}[Colluding Bidders]
Colluding bidders act as a single utility-maximizing entity that maximizes their joint utility $u_\mathsf{C}(\mathbf{b}^\mathsf{C}, \mathbf{b}^\mathsf{N}) = \sum_{i \in \mathsf{C}} (v_i x_i(\mathbf{b}) - p_i(\mathbf{b}))$, with full freedom to share information, transfer items, and make side payments among themselves.
\end{definition}

We adopt order-statistic notation to reference ranked bids and valuations. For any bidder subset $\mathsf{S} \subseteq \mathsf{M}$ with $|\mathsf{S}| = s$, let $b^\mathsf{S}_k$ and $v^\mathsf{S}_k$ denote the $k$-th largest bid and valuation in $\mathsf{S}$, respectively. Define $\mathsf{B}^\mathsf{S}_r = \{b^\mathsf{S}_1, \dots, b^\mathsf{S}_r\}$ and $\mathsf{V}^\mathsf{S}_r = \{v^\mathsf{S}_1, \dots, v^\mathsf{S}_r\}$ as the sets of the top $r$ bids and valuations in $S$. The complements are $\overline{\mathsf{B}^\mathsf{S}_r} = \mathsf{B}^\mathsf{S}_s \setminus \mathsf{B}^\mathsf{S}_r = \{b^\mathsf{S}_{r+1}, \dots, b^\mathsf{S}_s\}$ and $\overline{\mathsf{V}^\mathsf{S}_r} = \mathsf{V}^\mathsf{S}_s \setminus \mathsf{V}^\mathsf{S}_r = \{v^\mathsf{S}_{r+1}, \dots, v^\mathsf{S}_s\}$.



We distinguish realizations from random variables by using lowercase ($v, b$) for the former and uppercase ($V, B$) for the latter. Sets of realized values ($\mathsf{V}^\mathsf{S}_r, \mathsf{B}^\mathsf{S}_r$) correspond to sets of random variables ($\mathcal{V}^\mathsf{S}_r, \mathcal{B}^\mathsf{S}_r$).
\section{Weakness to Collusion}

\label{sec:vcgcoll}

In this section, we theoretically characterize the susceptibility of VCG and Posted Price to collusion.

\subsection{VCG with Collusion}
\label{ss:vcgcol}
Without collusion, VCG is efficient and DSIC. Under collusion, bid shading may occur \citep{sher_optimal_2012}. 

\begin{prop}[VCG under Collusion]
\label{prop:VCGeq}
Let $r_\mathsf{C}^{\scriptstyle col}$ $(r_\mathsf{C}^*)$, $r_{\mathsf{N}}^{\scriptstyle col}$ $(r_{\mathsf{N}}^*)$, and $\mathbf{b}^{\scriptstyle col}$ $(\mathbf{b}^*)$ be the number of items allocated to colluders, the number of items allocated to non-colluders, and the bid profiles in the presence (absence) of collusion, respectively, under VCG. Then
\begin{enumerate}[leftmargin=*]
    \item Colluders do not acquire more items through collusion ($r_\mathsf{C}^{\scriptstyle col} \leq r_\mathsf{C}^*$). They bid shade ($\mathbf{b}^{\mathsf{C}, \scriptstyle col}_i \leq \mathbf{v}^{\mathsf{C}\vphantom{,\scriptstyle col}}_i$ for $i \in \mathsf{C}$), bidding either $0$ or their true valuation $\mathbf{v}^\mathsf{C}_i$.
     \item The utility of each non-colluding bidder and the joint utility of the colluders do not decrease: $u_i(\mathbf{b}^*) \leq u_i(\mathbf{b}^{\scriptstyle col})$ for $i \in \mathsf{N}$ and $u_\mathsf{C}(\mathbf{b}^*) \leq u_\mathsf{C}(\mathbf{b}^{\scriptstyle col})$.
    \item The welfare and revenue are both non-increasing: $\mathtt{Wel}(\mathbf{b}^*) \geq \mathtt{Wel}(\mathbf{b}^{\scriptstyle col})$ and $\mathtt{Rev}(\mathbf{b}^*) \geq \mathtt{Rev}(\mathbf{b}^{\scriptstyle col})$.
\end{enumerate}
\end{prop}
Let $\mathrm{Wel}_{\textsc{VCG}}(\mathsf{N}, \mathsf{C})$ be the welfare of VCG with non-colluders $\mathsf{N}$ and colluders $\mathsf{C}$, and let $\mathrm{Wel}_{\textsc{OPT}}(\mathsf{M})$ be the welfare of an efficient mechanism with bidders $\mathsf{M}$. 



\begin{prop}\label{thm:BK} 
If bidder valuations $V_i$ are independent and identically distributed (i.i.d.), then
\begin{equation}
    \mathbb{E}\bigl[\mathtt{Wel}_{\textsc{VCG}}(\mathsf{N}\cup \mathsf{N}', \mathsf{C})\bigr] \geq \mathbb{E}\bigl[\mathtt{Wel}_{\textsc{OPT}}(\mathsf{N} \cup \mathsf{C})\bigr]   
\end{equation}
where $\mathsf{N}'$ are additional non-colluders with $|\mathsf{N}'| = |\mathsf{C}|$.
\end{prop}

Analogous to the Bulow-Klemperer theorem \citep{BKThm}, Proposition \ref{thm:BK} shows that recruiting $|\mathsf{C}|$ additional non-colluders recovers optimal welfare. This implies collusion substantially degrades VCG performance only when coalitions are large.

\subsection{Posted Price with Collusion}
\label{ss:pp}
Without collusion, Posted Price is DSIC though not efficient. Under collusion, overbidding occurs.
\begin{prop}[Posted Price under Collusion]
\label{thm:ppeq}
Let $r_\mathsf{C}^{\scriptstyle col}$ $(r_\mathsf{C}^*)$, $r_{\mathsf{N}}^{\scriptstyle col}$ $(r_{\mathsf{C}}^*)$, and $\mathbf{b}^{\scriptstyle col}$ $(\mathbf{b}^*)$ be the number of items allocated to colluders, the number of items allocated to non-colluders, and  bids in the presence (absence) of collusion, respectively, under Posted Price $p$. 
\begin{enumerate}[leftmargin=*]
    \item Colluders may get more items ($r_\mathsf{C}^{\scriptstyle col} \geq r_\mathsf{C}^*$) by overbidding ($\mathbf{b}^{\mathsf{C}, \scriptstyle col}_i \geq \mathbf{v}^{\mathsf{C}\vphantom{,\scriptstyle col}}_i$ for $i \in \mathsf{C}$), bidding either a value above the price $p$ or their true valuation $\mathbf{v}_i^\mathsf{C}$.
    \item The utility of each non-colluding bidder may decrease and the colluders' utility does not decrease: $u_i(\mathbf{b}^*) \geq u_i(\mathbf{b}^{\scriptstyle col})$ for $i \in \mathsf{N}$ and $u_\textsc{c}(\mathbf{b}^*) \leq u_\textsc{c}(\mathbf{b}^{\scriptstyle col})$.
    \item Welfare is non-increasing and revenue is constant: $\mathtt{Wel}(\mathbf{b}^*) \geq \mathtt{Wel}(\mathbf{b}^{\scriptstyle col})$ and $\mathtt{Rev}(\mathbf{b}^*) = \mathtt{Rev}(\mathbf{b}^{\scriptstyle col})$.
\end{enumerate}
\end{prop}
Let $\mathrm{NWel}_{p}(\mathsf{N}, \mathsf{C})$ be the non-colluder welfare of a Posted Price $p$ with non-colluders $\mathsf{N}$ and colluders $\mathsf{C}$. 

\begin{prop}\label{prop:ppbad} 
Let $q = \mathbb{P}(V_i > p)$, and suppose $|\mathsf{C}|=c> 2qr$. If the bidder valuations $V_i$ are i.i.d., then we have
\begin{multline}
\textstyle    \mathbb{E}\bigl[\mathtt{NWel}_p(\mathsf{N}, \mathsf{C})\bigr] \leq \bigl(\frac{1}{1-\epsilon}\bigr)\cdot\bigl(1-(1-q)\cdot\frac{c}{qn+c}+\\
\textstyle\epsilon_\mathsf{C}\cdot \frac{q(n+c)}{r}\bigr)\times
\mathbb{E}\bigl[\mathtt{NWel}_p(\mathsf{N} \cup \mathsf{C},\emptyset)\bigr],
\end{multline}
where $\epsilon = \exp(r-(n+c)q/2)$ and $\epsilon_\mathsf{C} = \exp(r-cq/2)$.
\end{prop}
This result shows that large collusive coalitions harm non-colluder welfare in Posted Price mechanisms.

\section{C-P\NoCaseChange{o}P Mechanism}
\label{sec:tapc}

This section presents and theoretically analzyes the C-PoP mechanism, which is  designed around the use of a classifier that identifies (even imprecisely) colluders.

A number of adaptive auction mechanisms compute a price by splitting bidders into two groups, applying a price or reserve price derived from the first group to the bids in the second. Because the second group's bids do not influence their own price, truthfulness is trivially achieved. In some of these mechanisms, bidders in the first group are ineligible to win an item \citep{Dhangwatnotai_Roughgarden_Yan_2015}, whereas in others, the second group is symmetrically used to compute a price for the first \citep{Goldberg_Hartline_Karlin_Saks_Wright_2006}. We disfavor the former because it sacrifices efficiency by rendering a subset of bidders ineligible by design; we disfavor the latter because it results in price discrimination.



Our C-PoP mechanism applies a single, uniform price to all winning bidders regardless of their subgroup and does not disqualify any bidder from receiving an item. Our design departs from existing Posted Price mechanisms in two key ways. First, we determine the Posted Price conditionally using the bids of non-colluding bidders (hence \emph{conditional}-Posted Price). Second, we charge a uniform price across all winners. Consequently, guaranteeing truthfulness is non-trivial, as the following two remarks illustrate:



\begin{remark}
\label{ex:1}
A mechanism that computes a Posted Price $p$ using all non-colluding bids $\mathbf{b}_{\mathsf{N}}$ is not truthful. A winning non-colluding bidder $i \in \mathsf{N}$ could strategically change $b_i$ to lower $p$, directly increasing their utility $u_i = v_i - p$. Thus, a truthful mechanism must ensure no winning bidder can influence the price applied to them, which is a requirement that is particularly challenging when enforcing a single uniform price across all winners without excluding bidders by design.
\end{remark}

\begin{remark}
\label{ex:2}
Uniformly allocating $r$ items at random among all bids exceeding a Posted Price $p$ fails truthfulness in the presence of collusion (see Section \ref{ss:pp}). Because colluders can overbid to manipulate allocation probabilities without altering $p$, a truthful Posted Price mechanism cannot rely on unconstrained uniform random allocation among qualifying bids.
\end{remark}


\subsection{Mechanism Design}


Algorithm \ref{alg:cpop} gives our C-PoP mechanism. It first runs a Posted Price auction at price $p^*$ for the non-colluding bidders $\mathsf{N}$, and then uses any remaining items to run a second Posted Price auction at the same price $p^*$ for the colluding bidders $\mathsf{C}$. To determine $p^*$, C-PoP employs a dynamic programming (DP) oracle (Algorithm \ref{alg:kporacle}) that computes an optimal price $p^*$ and a parameter $k^*$. The DP oracle ensures that $p^*$ is computed using only the values of the smallest $n - k^*$ non-colluding bids while maximizing the auctioneer's objective (e.g., welfare or revenue). Simpler alternatives to the DP oracle are provided in the Supplementary Materials, though they generally yield lower welfare or revenue.

\begin{algorithm}[t]
    \caption{C-PoP Mechanism}
    \label{alg:cpop}
    \begin{algorithmic}[1]
        \State \textbf{Input:} Number of items $r$; non-colluding $\mathsf{N}$ and colluding $\mathsf{C}$ bidders; bids $\mathbf{b}^{\mathsf{N}},\mathbf{b}^\mathsf{C}$; Oracle$(\mathbf{b}^{\mathsf{N}},n,c)$
        \State \textbf{Output:} 
        Allocation $x_i$ and payment $p_i$ for $i \in \mathsf{M}$
        \State Calculate  $(k^*,p^*) \leftarrow$Oracle$(\mathbf{b}^{\mathsf{N}},n,c)$
        \State \textbf{Phase 1: Run Posted Price on $\mathsf{N}$}
        \State Identify all of the qualifying non-colluding bidders: $\widehat{\mathsf{W}}_\mathsf{N} \gets \{i \in \mathsf{N} \mid \mathbf{b}^{\mathsf{N}}_i > p^*\}$
        \State Choose winning non-colluders: randomly select a subset $\mathsf{W}_{\mathsf{N}} \subseteq \widehat{\mathsf{W}}_{\mathsf{N}}$ of size $\min\{r,|\widehat{\mathsf{W}}_{\mathsf{N}}|\}$
        \State \textbf{Phase 2: Run Posted Price on $\mathsf{C}$ with remaining $\max\{r-|\mathsf{W}_{\textsc{n}}|, 0\}$ items}
        \State Identify all of the qualifying colluding bidders: $\widehat{\mathsf{W}}_{\mathsf{C}} \gets \{i \in \mathsf{C} \mid \mathbf{b}_i^{\mathsf{C}} > p^*\}$
        \State Choose winning colluders: randomly select a subset $\mathsf{W}_{\textsc{c}} \subseteq \hat{\mathsf{W}}_{\textsc{c}}$ of size $\min\{\max\{r-|\mathsf{W}_{\textsc{n}}|, 0\},|\widehat{\mathsf{W}}_\mathsf{C}|\}$
        \State Winners: $\mathsf{W} \gets \mathsf{W}_\textsc{n} \cup \mathsf{W}_\textsc{c}$
        \State Allocate items: $x_i \gets 1$ for $i \in \mathsf{W}$, $0$ otherwise
        \State Charge price: $p_i \gets p^*$ for $i \in \mathsf{W}$
    \end{algorithmic}
\end{algorithm}

\begin{algorithm}
    \caption{DP Oracle($\mathsf{b}^{\mathsf{N}}, n, c$)}
    \label{alg:kporacle}
    \begin{algorithmic}[1]
        \State \textbf{Input:} Count of items $r$, non-colluders $n$, and colluders $c$; non-colluder bids $\mathbf{b}^{\mathsf{N}}$; function $M(p, \overline{\mathsf{B}^\mathsf{N}_k})$
        \State \textbf{Output:} \(k^*, p^*\)
        \For{$k \in \{0,\dots, r\}$}
                \State $p_k \gets \arg \max_{p \ge \mathbf{b}^{\mathsf{N}}_{k+1}} M(p,\overline{\mathsf{B}^\mathsf{N}_k})$
        \EndFor
        \State \textbf{Phase 1: Compute value functions $V(p_k, \overline{\mathcal{B}^\mathsf{N}_{k}})$}
        \State Set $V(p_0, \overline{\mathcal{B}^\mathsf{N}_{0}}) = M(p_0, \overline{\mathcal{B}^\mathsf{N}_{0}}) = M(p_0, \mathcal{B}^\mathsf{N}_{n})$\;
    \For{$k = 1, \ldots, r$}
        \State Compute the backwards recursion \(V(p_k, \overline{\mathcal{B}^\mathsf{N}_k}) =
        \max \{
            M(p_k, \overline{\mathcal{B}^\mathsf{N}_k}),\;
            \mathbb{E}[V(p_{k-1}, \overline{\mathcal{B}^\mathsf{N}_{k-1}}) \mid \overline{\mathcal{B}^\mathsf{N}_{k}}]
        \}
        \)
    \EndFor
\State \textbf{Phase 2: Compute optimal split $k^*$}
        
        \State Set $k^* \leftarrow r$\;
        \While{$M(p_k, \overline{\mathsf{B}^\textsc{n}_k}) <$ 
        $\mathbb{E}[V( p_{k-1}, \overline{\mathcal{B}^\textsc{n}_{k-1}}) \mid \overline{\mathsf{B}^\textsc{n}_k}]$
    }
    \vspace{1pt}
            \State $k^* \leftarrow k^*-1$
        \EndWhile
        \State $p^* \leftarrow p_{k^*}$
    \end{algorithmic}
\end{algorithm}
\subsection{Truthfulness with Perfect Classification}

The theorem below establishes that C-PoP is truthful for all bidders under the DP oracle. We address the challenge outlined in Remark \ref{ex:1} by designing the DP oracle so that the Posted Price $p$, computed conditionally on a non-colluding bid $b$, always satisfies $p \geq b$. The oracle incrementally processes non-colluding bids ordered from lowest to highest, committing to a price at least as large as each evaluated bid. Specifically, at stage $k$, the dynamic program chooses between stopping to claim the current payoff $M(p_k, \overline{\mathsf{B}^{\mathsf{N}}_k})$, or advancing to stage $k-1$ to obtain the expected continuation value $\mathbb{E}[V(p_{k-1}, \overline{\mathsf{B}^{\mathsf{N}}_{k-1}})\mid\overline{\mathsf{B}^{\mathsf{N}}_k}]$. Choosing to continue commits to $k^* \le k-1$, which guarantees that the bidder with bid $\mathbf{b}^{\mathsf{N}}_k$ does not win. Consequently, any misreport by a non-colluding bidder either leaves them unallocated or yields a price $p \ge b$, ensuring that no profitable deviation exists. Finally, we resolve the issue in Remark \ref{ex:2} by allocating the items \emph{remaining} after the non-colluding auction to colluding bidders. Because both the allocation and pricing rules for colluders are completely independent of their bids, truthfulness holds for the coalition as well.

\begin{theorem}[C-PoP Truthfulness] \label{Lemma:dsic}
    C-PoP with the DP oracle is (ex-post) DSIC for all bidders, including colluders $\mathsf{C}$, when $\mathsf{N}$ and $\mathsf{C}$ are correctly specified. 
\end{theorem}


In the next subsection, we provide additional theoretical guarantees\,---\,namely, lower bounds on welfare, non-colluder utility, and price\,---\,for the C-PoP mechanism under bidder misclassification. We note that these bounds also hold under perfect classification.

\subsection{C-PoP with Misclassified Bidders}

A classifier may generate false positives (misclassifying non-colluders as colluders) and false negatives (misclassifying colluders as non-colluders). Let $\mathsf{N}_{\textsc{t}}$, $\mathsf{N}_{\textsc{f}}$, $\mathsf{C}_{\textsc{t}}$, and $\mathsf{C}_{\textsc{f}}$ be the correctly and incorrectly classified non-colluding bidders and colluding bidders, respectively.

This section presents theoretical bounds on the performance of C-PoP using the DP oracle when applied to the predicted sets of non-colluders $\widehat{\mathsf{N}} = \mathsf{N}_{\textsc{t}} \cup \mathsf{C}_{\textsc{f}}$ and colluders $\widehat{\mathsf{C}} = \mathsf{N}_{\textsc{f}} \cup \mathsf{C}_{\textsc{t}}$. These bounds assume a worst-case scenario where colluders have perfect knowledge of all bid values and bidder classifications; in practice, performance is likely to be significantly better given that colluders rarely possess such information.
 
\begin{algorithm}[t]
    \caption{Bid Manipulation with Misclassification}
    \label{alg:cpop-manipulation}
    \begin{algorithmic}[1]
        \State \textbf{Input:} Number of items $r$; Classified bidders $\mathsf{N}_\textsc{f}$, $\mathsf{N}_\textsc{t}$, $\mathsf{C}_\textsc{f}$, $\mathsf{C}_\textsc{t}$ with valuations $\mathbf{v}^{\mathsf{N}_\textsc{f}}$, $\mathbf{v}^{\mathsf{N}_\textsc{t}}$, $\mathbf{v}^{\mathsf{C}_\textsc{f}}$, $\mathbf{v}^{\mathsf{C}_\textsc{t}}$
        \State \textbf{Output:} Manipulated bids $b^{\mathsf{C}_\textsc{f}}$, $b^{\mathsf{C}_\textsc{t}}$.
        \State $p^*(|\mathsf{C}_\textsc f|) \gets$ Run C-PoP with bids $\mathbf{b}^{\widehat{\mathsf{N}}} = (\mathbf{v}^{\mathsf{N}_\textsc t}, \mathbf{v}^{\mathsf{C}_\textsc f})$ and  $\mathbf{b}^{\widehat{\mathsf{C}}} = (\mathbf{v}^{\mathsf{N}_\textsc f}, \mathbf{v}^{\mathsf{C}_\textsc t})$
        \State $\mathbf{b}^{\mathsf{C}_\textsc{f}}\gets \mathbf{v}^{\mathsf{C}_\textsc{f}}$
        \For{$k = |\mathsf{C}_\textsc{f}|,\dots,0$}
            \State $\overline{\mathsf{B}^\mathsf{S}_k} \gets 0$
            \State $p^*(k) \gets$ Run C-PoP with bids $\mathbf{b}^{\widehat{\mathsf{N}}} = (\mathbf{v}^{\mathsf{N}_\textsc t}, \mathbf{b}^{\mathsf{C}_\textsc f})$ and  $\mathbf{b}^{\widehat{\mathsf{C}}} = (\mathbf{v}^{\mathsf{N}_\textsc f}, \mathbf{v}^{\mathsf{C}_\textsc t})$
            \If{$p^*(k) > p^*(|\mathsf{C}_\textsc f|)$}
            \State{$u_\mathsf{C}(k)\gets-\infty$}
            \Else
            \For{$q = 0,\ldots,|\{i \in \mathsf{C}_\textsc{t} : v^{\mathsf{C}_\textsc{t}}_i < p^*(k)\}|$}
                \State $\mathbf{b}^{\mathsf{C}_\textsc{t}}\gets \mathbf{v}^{\mathsf{C}_\textsc{t}}$ 
                \State $\overline{\mathsf{B}^\mathsf{S}_k} \gets p^*(k) + \epsilon$ for small $\epsilon > 0$
                \State $w_\mathsf{C}(q) := \mathbb{E}[\sum_{i \in \mathsf{C}}x_i(v_i - p^*(k))]\gets$ under C-PoP with bids $\mathbf{b}^{\widehat{\mathsf{N}}} = (\mathbf{v}^{\mathsf{N}_\textsc t}, \mathbf{b}^{\mathsf{C}_\textsc f})$ and  $\mathbf{b}^{\widehat{\mathsf{C}}} = (\mathbf{v}^{\mathsf{N}_\textsc f}, \mathbf{b}^{\mathsf{C}_\textsc t})$
            \EndFor
            \State $q^*(k) \gets \arg\max_q w_\mathsf{C}(q)$
            \State $u_\mathsf{C}(k) \gets w_\mathsf{C}(q^*(k))$

            \EndIf
        \EndFor
        \State $k^* \gets \arg\max_ku_{\mathsf{C}}(k)$
        \State Choose $b^{\mathsf{C}_\textsc{f}}$, $b^{\mathsf{C}_\textsc{t}}$ to correspond to $k^*,p^*(k^*),q^*(k^*)$
    \end{algorithmic}
\end{algorithm}
Our next result shows that C-PoP provides additional robustness guarantees under misclassification, namely that it remains DSIC for all non-colluders and ensures a lower bound on the auction price. 
\begin{prop}\label{prop:pricebound}
Suppose C-PoP using the DP oracle is applied to the predicted set of non-colluders $\widehat{\mathsf{N}} = \mathsf{N}_{\textsc{t}} \cup \mathsf{C}_{\textsc{f}}$ and colluders $\widehat{\mathsf{C}} = \mathsf{N}_{\textsc{f}}\cup\mathsf{C}_{\textsc{t}}$. Then C-PoP is (ex-post) DSIC for all non-colluders $\mathsf{N}$. And if $|\mathsf{N}_{\textsc{t}}| \geq r+1$, then $p^* \geq \mathbf{b}^{\mathsf{N}_{\textsc{t}}}_{r+1}$ for the Posted Price.
\end{prop}

Let $\mathtt{Wel}_{\textsc{C-PoP}}(\widehat{\mathsf{N}}, \widehat{\mathsf{C}}, \mathbf{b}^{\widehat{\mathsf{N}}}, \mathbf{b}^{\widehat{\mathsf{C}}})$ be the welfare of the C-PoP mechanism that optimizes price $p^*$ as per Algorithm \ref{alg:cpop}, and let $\mathtt{CPoP}(p, \widehat{\mathsf{N}}, \widehat{\mathsf{C}}, \mathbf{b}^{\widehat{\mathsf{N}}}, \mathbf{b}^{\widehat{\mathsf{C}}})$ be the welfare realized by running Phases 1 and 2 from Algorithm \ref{alg:cpop} using the price $p$. Our mechanism sets $(k^*,p^*)$ using (conditional) expected welfare as $M(p, \overline{\mathsf{B}^{\widehat{\textsc{n}}}_{\scriptstyle k}})$ in the DP oracle, given in Algorithm \ref{alg:kporacle}, to achieve welfare guarantees. 

First, we show that when the predicted set of colluders is empty  ($\widehat{\mathsf C} = \emptyset$), C-PoP is equivalent to the VCG mechanism. Note that no prior knowledge about the distribution of bids is required in this case. 
\begin{prop}\label{lemma:VCG-PoPeq}
Suppose we have that $\widehat{\mathsf C} = \emptyset$. If $M(p, \overline{\mathsf{B}^{\widehat{\mathsf{N}}}_k}) = \mathbb{E}[\mathtt{CPoP}(p, \widehat{\mathsf{N}}, \emptyset, \mathbf{B}^{\widehat{\mathsf{N}}}, \emptyset) \mid  \overline{\mathsf{B}^{\widehat{\mathsf{N}}}_{\scriptstyle k}}]$, then C-PoP sets $p^* = \mathbf{b}^{\widehat{\mathsf{N}}}_{\scriptstyle r+1}$ and is equivalent to VCG on $\widehat{\mathsf{N}}$. 
\end{prop}
When the predicted set of colluders is non-empty, we compute the conditional expected welfare at price $p$ using prior information on valuation distributions and the observed bids $\mathbf{b}^{\widehat{\mathsf{N}}}$, as formalized in the next proposition. This calculation relies on the next assumption.
 

\begin{assumption}\label{assumption:iid}
Bidders' valuations are i.i.d. from a cumulative distribution function (CDF) $\mathsf{F}(\cdot)$, known to the auctioneer, with quantile $\mathsf{Q}(\cdot) := \mathsf{F}^{-1}(\cdot)$.
\end{assumption}
\begin{prop}[Conditional Expected Welfare] \label{lemma:conEW}
   If Assumption \ref{assumption:iid} holds, then for $p \geq \mathbf{b}^{\mathsf{\widehat{N}}}_{k+1}$ we have  \begin{multline}
   \label{eqn:cpopq}
\mathbb{E}[\mathtt{CPoP}(p, \widehat{\mathsf{N}}, \widehat{\mathsf{C}}, \mathbf{V}^{\widehat{\mathsf{N}}}, \mathbf{V}^{\widehat{\mathsf{C}}}) \mid   \overline{\mathsf{V}^{\widehat{\mathsf{N}}}_{\scriptstyle k}}]  = \\\mathbb{E}[\mathsf{Q}(U) \mid U \geq \mathsf{F}(p) ]\cdot \mathbb{E}(w),
\end{multline}
where \(\textstyle U \sim \mathrm{Uniform}[0,1]\),  \(w_{\widehat{\mathsf{N}}} \sim \mathrm{Binom}\Big(k, \tfrac{1-\mathsf{F}(p)}{1-\mathsf{F}(\mathbf{b}^{\widehat{\mathsf{N}}}_{k+1})}\Big)\), \(w_{\widehat{\mathsf{C}}} \sim \mathrm{Binom}\Big(\widehat{c}, 1-\mathsf{F}(p)\Big)\), and
   \begin{multline}
       \textstyle w = w_{\widehat{\mathsf{N}}}+ \sum_{\scriptstyle \nu=1}^{\scriptstyle r-w_{\widehat{\mathsf{N}}}-1}\nu\cdot \mathbb{P}(w_{\widehat{\mathsf{C}}}= \nu) +\\
       \textstyle \sum_{\scriptstyle \nu=r-w_{\widehat{\mathsf{N}}}}^{\scriptstyle \widehat{c}} \tfrac{\mathbb{E}[\sum_{i=1}^{\scriptstyle r-w_{\widehat{\mathsf{N}}}}\mathsf{Q}(U_i^{\nu}) \mid  U^{\nu}_{\nu} > \mathsf{F}(p)]}{\mathbb{E}[\mathsf{Q}(U) \mid U \ge \mathsf{F}(p)]} \cdot\mathbb{P}(w_{\widehat{\mathsf{C}}} = \nu).
   \end{multline}
\end{prop}


The next theorem provides welfare guarantees for C-PoP. Because C-PoP combines VCG and posted pricing, we consider a colluder strategy that blends the equilibrium strategies under VCG (see Proposition \ref{prop:VCGeq}) and posted pricing (see Proposition \ref{thm:ppeq}).
\begin{assumption}[Colluder Bidding Strategy]\label{assumption:col}
False negative colluders $\mathsf{C}_\textsc{f}$ shade their bids to lower the price set by C-PoP, while true positive colluders $\mathsf{C}_\textsc{t}$ overbid to get more allocation, with details as in Algorithm \ref{alg:cpop-manipulation}.
\end{assumption}


\begin{theorem}[Welfare Guarantee] \label{Thm:WelfMisclass}
 If Assumptions \ref{assumption:iid} and \ref{assumption:col} hold, $M\bigl(p_k, \overline{\mathsf{B}^{\widehat{\textsc{n}}}_k}\bigr)$ is chosen to be (\ref{eqn:cpopq}) in Algorithm \ref{alg:kporacle}, and  $|\mathsf{N}_{\textsc{t}}| \geq r+1$, then we have
\begin{multline}
\textstyle   \mathbb{E}\bigl[\mathtt{Wel}_{\textsc{C-PoP}}(\widehat{\mathsf{N}}, \widehat{\mathsf{C}}, \mathbf{B}^{\widehat{\mathsf N}}, \mathbf{B}^{\widehat{\mathsf C}}))\bigr]
 \textstyle \ge 
\\
\textstyle\mathbb{E}\bigl[\mathtt{Wel}_{\textsc{C-PoP}}(\mathsf{N}_{\textsc t}, \widehat{\mathsf{C}}, \mathbf{V}^{\mathsf N_{\textsc t}}, \mathbf{V}^{\widehat{\mathsf C}}))\bigr]  +\\
\textstyle -|\mathsf{N}_{\textsc F}|\cdot \mathbb{E}\bigl[\max \{ 0, V^{\mathsf N_{\textsc f}}_{1} - V^{\mathsf N_{\textsc T}}_{r+1}\}\bigr],  
\end{multline}
where $(\mathbf{B}^{\widehat{\mathsf N}}, \mathbf{B}^{\widehat{\mathsf C}})$ are the equilibrium bids under collusion and \((\mathbf{V}^{\mathsf N_{\textsc t}}, \mathbf{V}^{\widehat{\mathsf C}})\) are truthful bids.
\end{theorem}
Benchmarked against truthful bidding by $\mathsf{N}_{\mathrm{t}}$ and $\widehat{\mathsf{C}}$, this bound shows that false positives (misclassifying non-colluders as colluders) harm worst-case welfare more than false negatives (misclassifying colluders as non-colluders). It also quantifies C-PoP's welfare robustness against collusion, contrasting it with the vulnerability of VCG characterized in Section \ref{ss:vcgcol}.

\begin{remark}
    The additive gap vanishes as the ratio of false positives to true positives $|\mathsf{N}_{\textsc{f}}|/|\mathsf{N}_{\textsc{t}}|$ goes to zero, since \(\mathbb{E}[V^{\mathsf{N}_{\textsc{f}}}_{1} - V^{\mathsf{N}_{\textsc{t}}}_{r+1}] \le V_{\max} - \mathbb{E}[B^{\mathsf{N}_{\textsc{t}}}_{r+1}] = O(1/|\mathsf{N}_{\textsc{t}}|)\) for bounded distributions where the CDF is linear in the tail. (While the asymptotic rate differs for polynomial tails, a similar result still holds.)
\end{remark}

Finally, unlike the Posted Price mechanism that causes significant harm to non-colluders' welfare (Proposition \ref{prop:ppbad}), C-PoP guarantees a minimum non-colluders' welfare even when under large colluding coalitions.
\begin{theorem}[Non-colluder Welfare Guarantee]\label{thm:cpop-noncol}
  If Assumptions \ref{assumption:iid} and \ref{assumption:col} hold, $M\bigl(p_k, \overline{\mathsf{B}^{\widehat{\textsc{n}}}_k}\bigr)$ is chosen to be (\ref{eqn:cpopq}) in Algorithm \ref{alg:kporacle}, and  $|\mathsf{N}_{\textsc{t}}| \geq r+1$, then we have
  \begin{multline}
        \textstyle \mathbb{E}[\mathtt{NWel}_{\textsc{C-PoP}}(\widehat{\mathsf N}, \widehat{\mathsf C}, \mathbf{B}^{\widehat{\mathsf N}}, \mathbf{B}^{\widehat{\mathsf C}})] \ge \\
        \textstyle \mathbb{E}[\mathtt{NWel}_{\textsc{C-PoP}}(\mathsf N_{\textsc{t}}, \widehat{\mathsf{C}}, \mathbf{V}^{\mathsf N_{\textsc{T}}}, \mathbf{V}^{\widehat{\mathsf C}}))] +\\
        - |\mathsf{N}_{\textsc F}|\cdot\mathbb{E}\bigl[\max \{ 0, V^{\mathsf N_{\textsc f}}_{1} - V^{\mathsf N_{\textsc T}}_{r+1}\}\bigr],
  \end{multline}

where $(\mathbf{B}^{\widehat{\mathsf N}}, \mathbf{B}^{\widehat{\mathsf C}})$ are the equilibrium bids under collusion and \((\mathbf{V}^{\mathsf N_{\textsc t}}, \mathbf{V}^{\widehat{\mathsf C}})\) are truthful bids.
\end{theorem}
\subsection{Simplified Variants of C-PoP}
One simplified C-PoP mechanism replaces step 4 in Algorithm \ref{alg:kporacle} with $p_k \gets \mathbf{b}^\mathsf{N}_{k+1}$. We call this the VCG-Posted Price (V-PoP) mechanism since it uses the VCG price for giving $k$ items to the non-colluders.

In some instances, $\mathsf{F}(\cdot)$ and $\mathsf{Q}(\cdot)$ may only be partially known. Another simplified C-PoP mechanism uses:

\begin{assumption}\label{Assump:Qfunc}
Suppose the auctioneer knows value $L > 0$ such that $Lx \leq \mathsf{Q}(x)$ for all $x \in [0,1]$.
\end{assumption}
Using the above assumption, we can calculate minorants of (\ref{eqn:cpopq}) and use that as $M\bigl(p_k, \overline{\mathsf{B}^{\widehat{\textsc{n}}}_k}\bigr)$ within C-PoP. Details are provided in the Supplemental Materials.

\section{Numerical Experiments}
\label{sec:ne}

\begin{figure*}[!t]
    \centering
    \begin{subfigure}{0.32\textwidth}
        \centering
        \includegraphics[width=\linewidth]{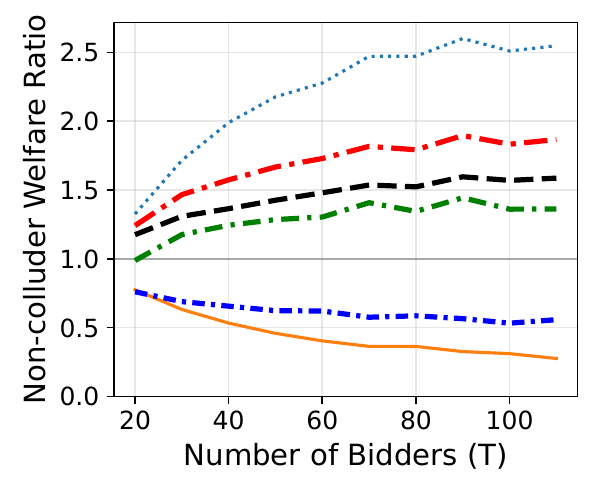}
        \caption{Non-colluder Welfare Ratio}
        \label{fig:tsweep_nc_U}
    \end{subfigure}
    \hfill
    \begin{subfigure}{0.32\textwidth}
        \centering
        \includegraphics[width=\linewidth]{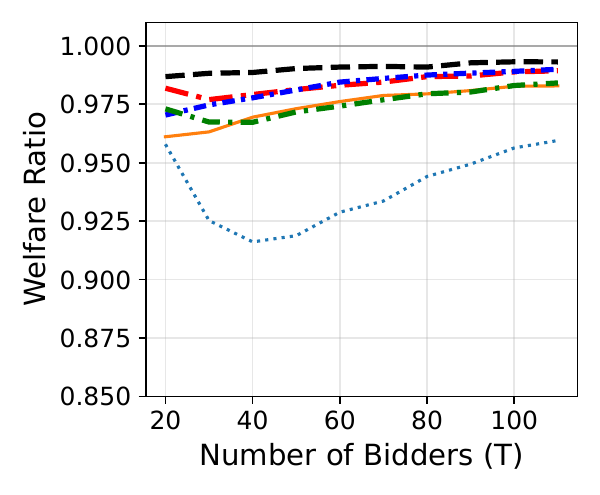}
        \caption{Welfare Ratio}
        \label{fig:welfare_U}
    \end{subfigure}
    \hfill
    \begin{subfigure}{0.32\textwidth}
        \centering
        \includegraphics[width=\linewidth]{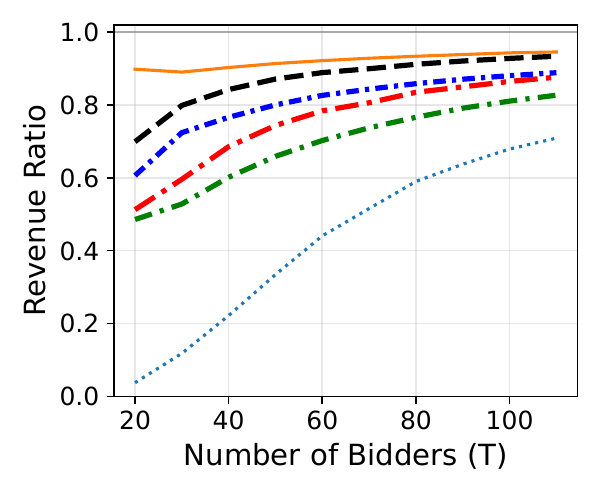}
        \caption{Revenue Ratio}
        \label{fig:revenue_U}
    \end{subfigure}

    \begin{subfigure}{0.32\textwidth}
        \centering
        \includegraphics[width=\linewidth]{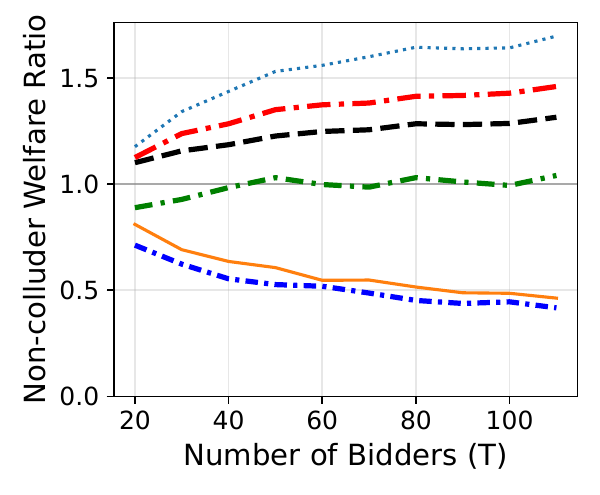}
        \caption{Non-colluder Welfare Ratio}
        \label{fig:tsweep_nc_E}
    \end{subfigure}
    \hfill
    \begin{subfigure}{0.32\textwidth}
        \centering
        \includegraphics[width=\linewidth]{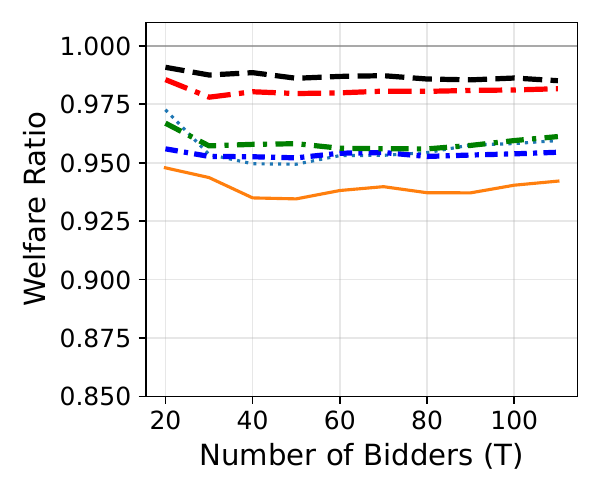}
        \caption{Welfare Ratio}
        \label{fig:welfare_E}
    \end{subfigure}
    \hfill
    \begin{subfigure}{0.32\textwidth}
        \centering
        \includegraphics[width=\linewidth]{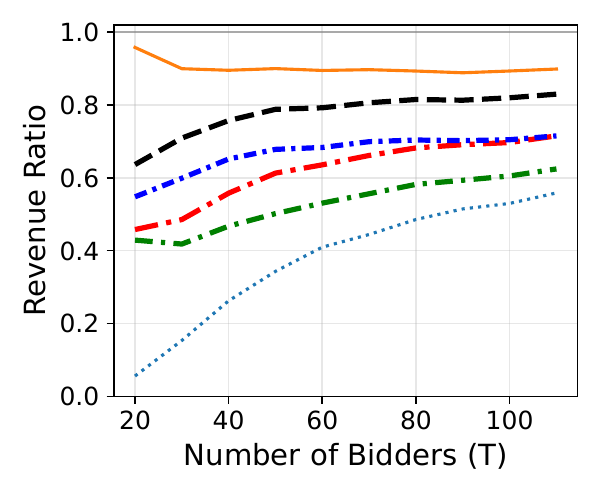}
        \caption{Revenue Ratio}
        \label{fig:revenue_E}
    \end{subfigure}
    \caption{Non-colluder welfare, welfare, and revenue ratios (relative to truthful VCG without collusion) versus the number of bidders $T$, for $r = 10$ and about 84\% colluders. The top row uses Uniform(0,1) valuations and the bottom row uses Exponential ($\lambda = 2$) valuations. The solid orange curve is Posted Price with collusion, the dotted light blue curve is VCG with collusion, the long dash-dotted red curve is C-PoP with the firm-size classifier, the medium dash-dotted green curve is C-PoP with the 50/50 classifier, the short dash-dotted blue curve is C-PoP with the random classifier, and the dashed black curve is C-PoP with perfect classification. The thin gray line marks a ratio of one.}
    \label{fig:exp-unknown-2}
\end{figure*}

\begin{figure*}[!t]
    \centering
    \begin{subfigure}{0.32\textwidth}
        \centering
        \includegraphics[width=\linewidth]{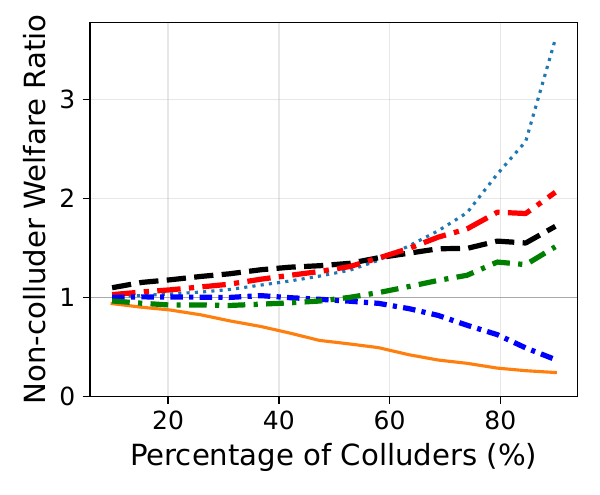}
        \caption{Non-colluder Welfare Ratio}
        \label{fig:share_nc_U}
    \end{subfigure}
    \hfill
    \begin{subfigure}{0.32\textwidth}
        \centering
        \includegraphics[width=\linewidth]{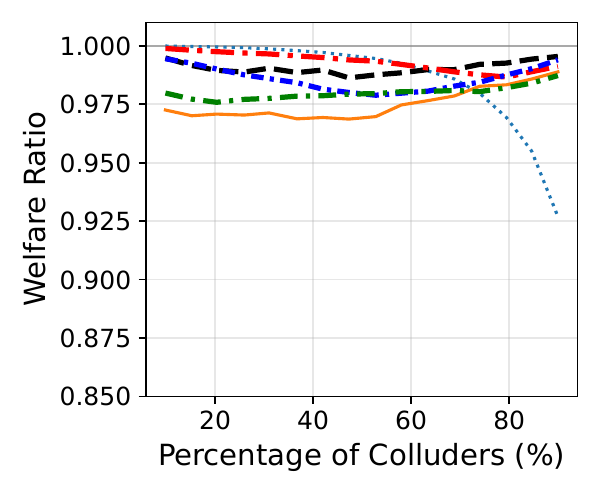}
        \caption{Welfare Ratio}
        \label{fig:share_wel_U}
    \end{subfigure}
    \hfill
    \begin{subfigure}{0.32\textwidth}
        \centering
        \includegraphics[width=\linewidth]{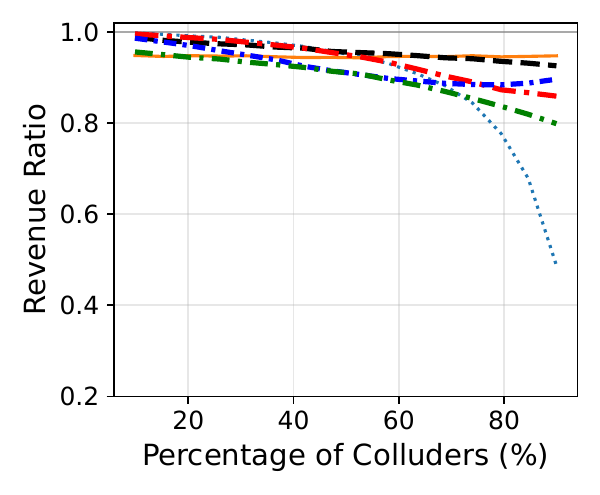}
        \caption{Revenue Ratio}
        \label{fig:share_rev_U}
    \end{subfigure}

    \begin{subfigure}{0.32\textwidth}
        \centering
        \includegraphics[width=\linewidth]{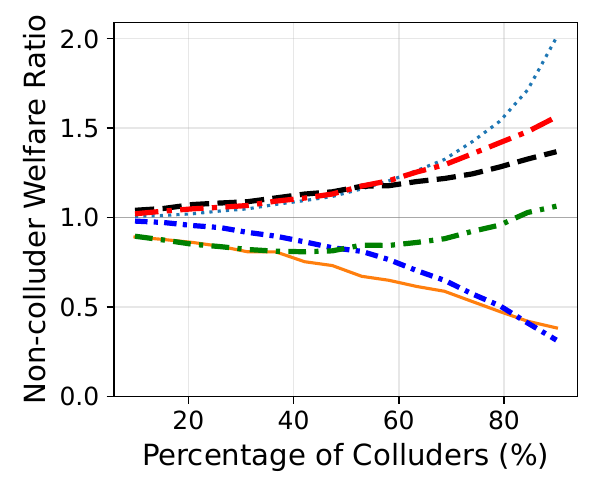}
        \caption{Non-colluder Welfare Ratio}
        \label{fig:share_nc_E}
    \end{subfigure}
    \hfill
    \begin{subfigure}{0.32\textwidth}
        \centering
        \includegraphics[width=\linewidth]{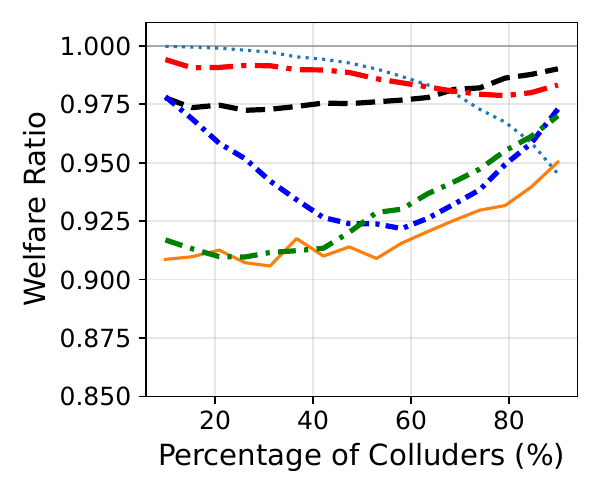}
        \caption{Welfare Ratio}
        \label{fig:share_wel_E}
    \end{subfigure}
    \hfill
    \begin{subfigure}{0.32\textwidth}
        \centering
        \includegraphics[width=\linewidth]{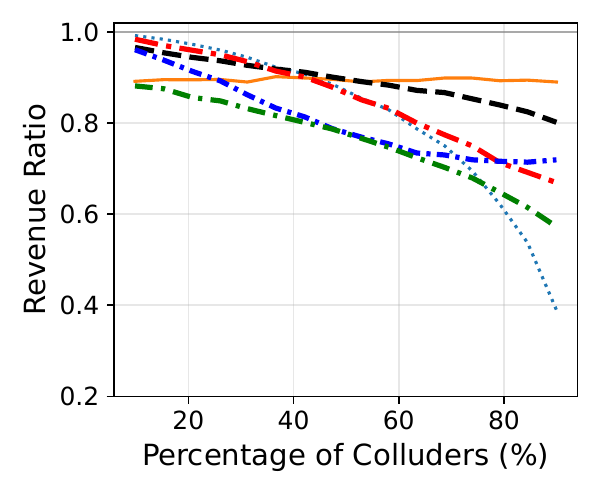}
        \caption{Revenue Ratio}
        \label{fig:share_rev_E}
    \end{subfigure}
    \caption{Non-colluder welfare, welfare, and revenue ratios (relative to truthful VCG without collusion) versus the percentage of colluders, for $r = 10$ and $T = 110$. The top row uses Uniform(0,1) valuations and the bottom row uses Exponential ($\lambda = 2$) valuations. The thin gray line marks a ratio of one. Line styles are as in Figure~\ref{fig:exp-unknown-2}.}
    \label{fig:share}
\end{figure*}

In this section, we conduct numerical experiments comparing C-PoP (paired with a DP oracle under classifiers of varying sophistication) against two benchmark auction mechanisms. Throughout, colluders internally reallocate won items to their highest-value members. Our benchmark mechanisms are:
\begin{enumerate}[leftmargin=*]
    \item \textit{VCG with Collusion}: A VCG auction where non-colluders bid truthfully, while colluders strategically shade their bids as described in Section \ref{ss:vcgcol}.
    \item \textit{Posted Price with Collusion}: A Posted Price auction where the price is optimized assuming truthful bidding, and items are randomly allocated to qualifying bidders. Colluders overbid as per Section \ref{ss:pp}.
\end{enumerate}

We endogenize bidder classification using a firm-size model, where each bidder $i \in \{1, \dots, T\}$ is characterized by a size attribute $S_i$ drawn independently of its valuation $V_i$. To model larger firms' greater capacity to collude, bidder $i$ is designated as a colluder with probability $\sigma(S_i)$ (via a sigmoid function $\sigma$) and a non-colluder otherwise. We also define a lower-bounding function $\ell(S_i) \leq \sigma(S_i)$ utilized by one of the classifiers we consider. The functional forms of $\sigma$ and $\ell$ are fully specified in the Supplementary Material.

We evaluate the misclassification robustness of C-PoP with the DP oracle by using classifiers of varying sophistication. We consider the following classifiers:
\begin{enumerate}[leftmargin=*]
    \item \textit{C-PoP with Perfect Classification}: All bidders are truthful by Theorem \ref{Lemma:dsic} with this perfect classifier.
    \item  \textit{C-PoP with Firm-Specific Classification}: Bidder $i$ is classified as a colluder with probability $\ell(S_i)$.
    \item \textit{C-PoP with Random Classification}: Bidder $i$ is classified as a colluder with probability equal to the percent of colluders in the auction.
    \item \textit{C-PoP with 50/50 Classification}: Bidder $i$ is  classified as a colluder with a 50\% chance.
\end{enumerate}

To enable direct comparisons across mechanisms, we normalize each metric by its value under truthful, collusion-free VCG (the efficient baseline). The non-colluder welfare ratio is defined as the average welfare received by non-colluders under a given mechanism divided by their average welfare under collusion-free VCG. This metric captures fairness: values below one indicate that non-colluders fare worse than in an uncolluded auction. Total welfare and revenue ratios are defined analogously, and all metrics are computed over identical simulated auction instances.

\subsection{Varying Number of Bidders}

We benchmarked C-PoP across $r \in \{10, 15\}$ items and bidder counts $T \in \{20, 30, \dots, 110\}$, using a firm-size distribution $\mathcal{S} \sim \mathrm{Beta}(5,3)$. Results were averaged over $1{,}000$ independent trials; complete tables and standard errors appear in the Supplementary Material. Figure \ref{fig:exp-unknown-2} shows that C-PoP with perfect classification  achieves the highest welfare across all $T$. Although posted pricing yields slightly higher revenue, it degrades non-colluder welfare relative to C-PoP. Finally, C-PoP exhibits strong robustness to classifier performance, maintaining stable welfare and revenue across classification models of varying sophistication.


\subsection{Varying Percentage of Colluders}

We next varied the percentage of colluders from 10\% to 90\% by setting $\mathcal{S} \sim \mathrm{Beta}(\alpha, 4)$ with $\alpha \in [1.01, 7.36]$, for fixed $T = 110$ and $r = 10$. Figures \ref{fig:share_rev_U} and \ref{fig:share_rev_E} show that VCG with collusion suffers severe revenue degradation when collusion is widespread, falling to 49\% (Uniform) and 39\% (Exponential) of the collusion-free baseline at 90\% colluders. In contrast, C-PoP with perfect classification retains at least 93\% and 80\% of collusion-free revenue, respectively. Figures \ref{fig:share_wel_U} and \ref{fig:share_wel_E} show that C-PoP under either perfect or firm-size classification achieves welfare at least as high as the lower of the two baselines across all colluder percentages.

Figures \ref{fig:share_nc_U} and \ref{fig:share_nc_E} show that C-PoP with an uninformative 50/50 classifier is at least as fair to non-colluders as Posted Price with collusion. While VCG with collusion yields higher non-colluder welfare above roughly 55--60\% collusion as colluders shade bids to cede items, its revenue collapses in this regime. C-PoP falls behind both baselines only under Exponential valuations with a random classifier at severe collusion, where high false-positive rates allow overbidding colluders to crowd out non-colluders in Phase~2. Because lower false-positive rates improve fairness, these results align with our theoretical finding that false negatives are preferable to false positives (Theorems \ref{Thm:WelfMisclass} and \ref{thm:cpop-noncol}). Additional experiments in the Supplementary Materials confirm this by varying classifier error rates independently.

Suppose the auctioneer knows the number of bidders $T$ and items $r$, but has no information about collusion\,---\,knowing neither individual colluder identities nor the total colluder fraction. The 50/50 classifier requires no prior collusion data (unlike the random classifier, which depends on the colluder fraction). Lacking explicit auctioneer preferences over fairness, revenue, and welfare, a natural decision criterion is to maximize the worst-case outcome across these three metrics over all collusion levels. For $T = 110$ and $r = 10$, C-PoP with the 50/50 classifier maintains worst-case retention of at least 80\% (Uniform) and 57\% (Exponential) of the collusion-free VCG outcome, compared to worst-case retention of 24\%/38\% under Posted Price and 49\%/39\% under VCG. Consequently, Figure \ref{fig:share} demonstrates that an auctioneer seeking to avoid severe performance drop-offs across revenue, welfare, and fairness would prefer C-PoP with the 50/50 classifier.

\vspace{0.03cm}
\section{Conclusion}
To mitigate bidder collusion, we proposed the C-PoP mechanism that uses machine learning predictions of collusion. C-PoP uses a DP oracle to sequentially set prices, prioritizing predicted non-colluders before offering remaining items to predicted colluders. With perfect classification, C-PoP is truthful even under collusion. For imperfect classification, we provided theoretical bounds on expected welfare and non-colluder welfare. Experiments showed that C-PoP outperforms VCG and Posted Price, maintaining strong performance even under imperfect classification. Future work includes extensions to heterogeneous items.\vfill\newpage

 \bibliographystyle{plainnat}
 \bibliography{references}

\clearpage
\appendix
\thispagestyle{empty}

\onecolumn
\aistatstitle{Supplementary Materials}

\section{Additional Results}

This section of the supplementary materials contains additional results.

\subsection{Bidding in VCG with Collusion}

\begin{lemma}
\label{Lemma:colluderVCG}
    Consider the VCG mechanism. Let $r_{\mathsf{C}}$ be the number of items colluders win, and let $r_{\mathsf{N}} = r-r_{\mathsf{C}}$ be the number of items non-colluders win. It is optimal for non-colluders to bid truthfully $\mathbf{b}^{\mathsf{N}} = \mathbf{v}^{\mathsf{N}}$. It is optimal for colluders to bid $\mathbf{b}^{\mathsf{C}}_{i} = \max\{\mathbf{v}^\mathsf{C}_i, \mathbf{b}^{\mathsf{N}}_{r_{\mathsf{N}}+1} + \epsilon\}$ for $i \in \{1, \dots, r_\mathsf{C} \}$ and $\mathbf{b}^\mathsf{C}_{i} = 0$ for $i \in \{r_\mathsf{C}+1, \dots, c \}$, for any small $\epsilon > 0$, where $b^{\mathsf{N}}_{r_{\mathsf{N}} +1}$ is the largest losing bid among non-colluding bidders. With this optimal bidding strategy, the colluders achieve a utility of
            \begin{equation}
            \textstyle u_\mathsf{C}(\mathbf{v}^{\mathsf{N}}; r_\mathsf{C}) := \max_{\mathbf{b}^\mathsf{C}} \left\{u_\mathsf{C}\left(\mathbf{b}^\mathsf{C}, \mathbf{v}^{\mathsf{N}}\right) \mid \sum_{i \in \mathsf{C}} x_i(\mathbf{b}^\mathsf{C}, \mathbf{v}^{\mathsf{N}}) = r_\mathsf{C} \right\} = \sum_{i = 1}^{r_\mathsf{C}} \mathbf{v}^\mathsf{C}_{i} - r_\mathsf{C} \cdot \mathbf{b}^{\mathsf{N}}_{r_{\mathsf{N}}+1}.
        \end{equation}
Let $r_{\mathsf{C}}^{*}$ and $r_{\mathsf{N}}^{*}$ be the VCG allocation to colluding and non-colluding bidders, respectively, in the absence of collusion. Even when manipulating their bids, the colluding bidders do not desire more than $r_\mathsf{C}^{*}$ items,that meaning that $u_\mathsf{C}(\mathbf{v}^{\mathsf{N}}; r_\mathsf{C}^*) \geq u_\mathsf{C}(\mathbf{v}^{\mathsf{N}}; r_\mathsf{C}^* + \Delta r)$ for all $\Delta r>0$.
\end{lemma}

\begin{proof}
Since the VCG mechanism is DSIC and since non-colluding bidders maximize only their own utility, their bids $\mathbf{b}^{\mathsf{N}}$ must be their true valuations $\mathbf{v}^{\mathsf{N}}$. The colluders maximize $u_\mathsf{C}(\mathbf{b}^\mathsf{C},\mathbf{v}^{\mathsf{N}}) = \sum_{i \in \mathsf{C}} v_i x_i - p_i$. For $\sum_{i \in \mathsf{C}} x_i = r_\mathsf{C}$, the $r_\mathsf{C}$ largest colluder bids win. The VCG payment $p_i$ is $\max\{\mathbf{b}^\mathsf{C}_{r_\mathsf{C}+1}, \mathbf{b}^{\mathsf{N}}_{r_{\mathsf{N}}+1}\}$. The choice $\mathbf{b}^\mathsf{C}_{i} = 0$ for $i \in \{r_\mathsf{C}+1, \dots, C \}$ minimizes the payment since $\mathbf{b}^\mathsf{C}_{r_\mathsf{C}+1}=0$. The winning bids $\mathbf{b}^\mathsf{C}_{i} = \max\{\mathbf{v}^\mathsf{C}_i, \mathbf{b}^{\mathsf{N}}_{r_{\mathsf{N}}+1} + \epsilon\}$ for $i \in \{1, \dots, r_\mathsf{C} \}$ ensure allocation and $\mathbf{b}^\mathsf{C}_{i} > \mathbf{b}^{\mathsf{N}}_{r_{\mathsf{N}}+1}$ for small $\epsilon >0$. Since payment $p_i = \mathbf{b}^{\mathsf{N}}_{r_{\mathsf{N}}+1}$, the maximum utility for fixed $r_\mathsf{C}$ is $ u_\mathsf{C}(\mathbf{b}^\mathsf{C},\mathbf{v}^{\mathsf{N}}) = \sum_{i = 1}^{r_\mathsf{C}} \mathbf{v}^\mathsf{C}_{i} - r_\mathsf{C} \cdot \mathbf{b}^{\mathsf{N}}_{r_{\mathsf{N}}+1}$. The welfare-maximizing allocation is $(r_\mathsf{C}^*, r_\mathsf{N}^*)$. For any $\Delta r > 0$, we have that
\begin{equation}
    \begin{aligned}
        \textstyle u_\mathsf{C}(\mathbf{v}^{\mathsf{N}}; r_\mathsf{C}^* + \Delta r) - u_\mathsf{C}(\mathbf{v}^{\mathsf{N}}; r_\mathsf{C}^*) &= \textstyle \sum_{i = r_\mathsf{C}^*+1}^{r_\mathsf{C}^*+\Delta r} \mathbf{v}^\mathsf{C}_{i} - \left(r_\mathsf{C}^* +\Delta r\right)\mathbf{b}^{\mathsf{N}}_{r_\mathsf{N}^*-\Delta r+1}+ r_\mathsf{C}^* \mathbf{b}^{\mathsf{N}}_{r_\mathsf{N}^*+1} \\
        & \textstyle= \sum_{i=1}^{\Delta r} \bigl(\mathbf{v}^\mathsf{C}_{r_\mathsf{C}^*+i} - \mathbf{b}^{\mathsf{N}}_{r_\mathsf{N}^*-\Delta r+1}\bigr) + r_\mathsf{C}^*\bigl( \mathbf{b}^{\mathsf{N}}_{r_\mathsf{N}^*+1} - \mathbf{b}^{\mathsf{N}}_{r_\mathsf{N}^*-\Delta r+1}\bigr) \\
        & \textstyle \leq 0.
    \end{aligned}
    \end{equation}
    This holds because: (1) $\mathbf{v}^\mathsf{C}_{r_\mathsf{C}^*+i} \leq 
    \mathbf{b}^{\mathsf{N}}_{r_\mathsf{N}^*}$ for $i \geq 1$, as these are items the colluders did not win in the truthful VCG allocation, and $\mathbf{b}^{\mathsf{N}}_{r_\mathsf{N}^*-\Delta r+1} \geq \mathbf{b}^{\mathsf{N}}_{r_\mathsf{N}^*}$; thus $\mathbf{v}^\mathsf{C}_{r_\mathsf{C}^*+i} \leq \mathbf{b}^{\mathsf{N}}_{r_\mathsf{N}^*-\Delta r+1}$. (2) $\mathbf{b}^{\mathsf{N}}$ is sorted in descending order, and $\mathbf{b}^{\mathsf{N}}_{r_\mathsf{N}^*+1} \leq \mathbf{b}^{\mathsf{N}}_{r_\mathsf{N}^*-\Delta r+1}$ since $\Delta r>0$. Both parts of the sum are non-positive.
\end{proof}

\subsection{Alternative Oracles}\label{app:oracle}
We present three other approaches for designing the oracle in Algorithm \ref{alg:kporacle}: unconditional and conditional maximization-based approaches, and a greedy approach. In general, approaches that incorporate more observed non-colluding bids $b^{\textsc{n}}$ when computing the optimal $(k^*,p^*)$ perform better, albeit at increased computational cost. In particular, dynamic programming performs best, followed by greedy, conditional, and unconditional maximization, in that order. 

\begin{algorithm}[htbp]
    \caption{Oracle($b^{\textsc{n}}, n, c$)}
    \begin{algorithmic}[1]
        \State \textbf{Input:} Total number of items $r$; number of non-colluding bidders and colluding bidders $n$, $c$ respectively; non-colluding bids $b^{\textsc{n}}$; function $M(p, \overline{\mathsf{B}^\textsc{n}_k})$. 
        \State \textbf{Parameter:} Approach $\in$ \{Unconditional Maximization, Conditional Maximization, Greedy\}
        \State \textbf{Output:} \(k^*, p^*\)
        
        \For{$k \in \{0,\cdots, r\}$}
            \State $p_k = \arg \max_{p \ge b^{\textsc{n}}_{k+1}} M(p,\overline{\mathsf{B}^\textsc{n}_k})$
        \EndFor
        
        \If {Approach $==$ Unconditional Maximization} \Comment{\textbf{Approach 1}}
        \State $k^* \leftarrow \arg\max_{k \in \{0, \cdots,r\}} \mathbb{E}\left[M(p_k, \overline{\mathcal{B}^\textsc{n}_k})\right]$
        \State $p^* \leftarrow p_{k^*}$
        \EndIf
        
        \If {Approach $==$ Conditional Maximization} \Comment{\textbf{Approach 2}}
        \State $k^* \leftarrow \arg\max_{k \in \{0, \cdots,r\}} \mathbb{E}\left[M(p_k, \overline{\mathcal{B}^\textsc{n}_k}) \mid \overline{\mathsf{B}^\textsc{n}_r}\right]$
        \State $p^* \leftarrow p_{k^*}$
        \EndIf
        
        \If {Approach $==$ Greedy} \Comment{\textbf{Approach 3}}
        \State Set $k^* \leftarrow r$
        \While{$M(p_{k^*}, \overline{\mathsf{B}^\textsc{n}_{k^*}}) < \mathbb{E}\left[M(p_{k^*-1}, \overline{\mathcal{B}^\textsc{n}_{k^*-1}}) \mid \overline{\mathsf{B}^\textsc{n}_{k^*}}\right]$}
            \State $k^* \leftarrow k^*-1$
        \EndWhile
        \State $p^* \leftarrow p_{k^*}$
        \EndIf
    \end{algorithmic}
\end{algorithm}

\subsection{C-PoP with Partially Known Bid Distribution}\label{app:minorant}
The next lemma links order statistics of bidders' valuations to those of samples from a uniform distribution, establishing a simple way to bound valuations using the quantile function’s upper and lower slopes.

\begin{lemma}\label{Lemma:LUbound}
    Given that Assumptions \ref{assumption:iid} and \ref{Assump:Qfunc} hold, an order statistic $V^\textsc{s}_k$ satisfies $L U^\textsc{s}_k \leq V^\textsc{s}_k$, where $U^\textsc{s}_k$ is the k-th largest order statistic among $S$ i.i.d. standard uniform samples. Furthermore, $\mathsf{F}(V^\textsc{n}_{k}) = U^\textsc{n}_{k}$.
\end{lemma}
\begin{proof}
The result is a consequence of the probability integral transform of continuous random variables. By equation $(2.3.7)$ of \cite{David_Nagaraja_2003}, $V^\textsc{s}_k = \mathsf{Q}(U^\textsc{s}_k)$.  $L U^\textsc{s}_k \leq V^\textsc{s}_k$ follows from Assumption \ref{Assump:Qfunc}. Using $\mathsf{Q}(\cdot) = \mathsf{F}^{-1}(\cdot)$ from Assumption \ref{assumption:iid}, we get $\mathsf{F}(V^\textsc{n}_{k}) = U^\textsc{n}_{k}$.
\end{proof}

Next, we construct two different minorants for expected welfare.
\begin{prop}[Minorant of Welfare] \label{Lemma:WelfareMinorant}
  Let Assumptions \ref{assumption:iid} and \ref{Assump:Qfunc} hold. Assuming truthful bidding, i.e., $\mathbf{B}^{\widehat{\mathsf{N}}} = \mathbf{V}^{\widehat{\mathsf{N}}}$ and $\mathbf{B}^{\widehat{\mathsf{C}}} = \mathbf{V}^{\widehat{\mathsf{C}}}$, we have:
\(\textstyle \mathbb{E}\left[\mathtt{CPoP}(p, \widehat{\mathsf{N}}, \widehat{\mathsf{C}}, \mathbf{V}^{\widehat{\mathsf{N}}}, \mathbf{V}^{\widehat{\mathsf{C}}}) \right] \geq \mathbb{E}\big[\frac{L(1+U^{\widehat{\mathsf{N}}}_{k+1})}{2}  \mathbb{E}\big[w_{\widehat{\mathsf{N}}} + \sum_{j=1}^{r-w_{\widehat{\mathsf{N}}}-1}j \mathbb{P}\big(w_{\widehat{\mathsf{C}}}=j\big) \textstyle + \sum_{j=r-w_{\widehat{\mathsf{N}}}}^{|\widehat{\mathsf{C}}|}\mathbb{P}\big(w_{\widehat{\mathsf{C}}}=j\big)\frac{2(r-w_{\widehat{\mathsf{N}}})}{(1+U^{\widehat{\mathsf{N}}}_{k+1})} \big(1-\frac{(1-U^{\widehat{\mathsf{N}}}_{k+1})(r-w_{\widehat{\mathsf{N}}}+1)}{(j+1)}\big) \big| U^{\widehat{\mathsf{N}}}_{k+1}\big] \big],
\)
  where $w_{\widehat{\mathsf{N}}} \sim \mathrm{Binom}\big(k,\frac{1-\mathsf{F}(p)}{1-U^{\widehat{\mathsf{N}}}_{k+1}}\big)$, $w_{\widehat{\mathsf{C}}} \sim \mathrm{Binom}\big(|\widehat{\mathsf{C}}|,1-\mathsf{F}(p)\big)$ and \(U^{\widehat{\mathsf{N}}}_{k+1} \sim \mathrm{Beta}(|\widehat{\mathsf{N}}|-k,k+1)\).
\end{prop}

\begin{proof} 
By Proposition \ref{lemma:conEW}, we have:
\begin{equation*}
\begin{split}
    &\textstyle \mathbb{E}\big[\mathtt{CPoP}(p, \widehat{\mathsf{N}}, \widehat{\mathsf{C}}, \mathbf{V}^{\widehat{\mathsf{N}}}, \mathbf{V}^{\widehat{\mathsf{C}}}) \mid p \ge V^{\widehat{\mathsf{N}}}_{k+1}, \overline{\mathcal{V}^{\widehat{\mathsf{N}}}_{k}}\big] = \mathbb{E}\big[w_{\widehat{\mathsf{N}}}\big]\mathbb{E}\big[ \mathsf{Q}(U)\mid U > \mathsf{F}(p) \big] \\
    &\textstyle + \mathbb{E}\big[\sum_{j=1}^{r-w_{\widehat{\mathsf{N}}}-1}j \mathbb{P}\big(w_{\widehat{\mathsf{C}}}=j\big)\mathbb{E}\big[ \mathsf{Q}(U) \mid U > \mathsf{F}(p)\big] \big] \textstyle + \mathbb{E}\big[\sum_{j=r-w_{\widehat{\mathsf{N}}}}^{|\widehat{\mathsf{C}}|}\mathbb{P}\big(w_{\widehat{\mathsf{C}}}=j\big)\mathbb{E}\big[\sum_{i=1}^{r-w_{\widehat{\mathsf{N}}}}  \mathsf{Q}(U_i^j) \mid  U_j^j > \mathsf{F}(p)\big] \big],
\end{split}
\end{equation*}
where \(U \sim \mathrm{Uniform}[0,1]\), $w_{\widehat{\mathsf{N}}} \sim \mathrm{Binom}\big(k,\frac{1-\mathsf{F}(p)}{1-\mathsf{F}(V^{\widehat{\mathsf{N}}}_{k+1})}\big)$, and $w_{\widehat{\mathsf{C}}} \sim \mathrm{Binom}\big(|\widehat{\mathsf{C}}|,1-\mathsf{F}(p)\big)$. Using Assumption \ref{Assump:Qfunc} and Lemma \ref{Lemma:LUbound}, we bound the above terms as follows:
\begin{align*}
    &\textstyle \mathbb{E}\big[w_{\widehat{\mathsf{N}}}\big]\mathbb{E}\big[ \mathsf{Q}(U)\mid U > \mathsf{F}(p) \big]  + \mathbb{E}\big[\sum_{j=1}^{r-w_{\widehat{\mathsf{N}}}-1}j \mathbb{P}\big(w_{\widehat{\mathsf{C}}}=j\big)\mathbb{E}\big[ \mathsf{Q}(U) \mid U > \mathsf{F}(p)\big] \big] \\
    &\textstyle + \mathbb{E}\big[\sum_{j=r-w_{\widehat{\mathsf{N}}}}^{|\widehat{\mathsf{C}}|}\mathbb{P}\big(w_{\widehat{\mathsf{C}}}=j\big)\mathbb{E}\big[\sum_{i=1}^{r-w_{\widehat{\mathsf{N}}}}  \mathsf{Q}(U_i^j) \mid  U_j^j > \mathsf{F}(p)\big] \big]\\
    &\textstyle \ge \mathbb{E}\big[w_{\widehat{\mathsf{N}}}\big]\mathbb{E}\big[ LU\mid U > U^{\widehat{\mathsf{N}}}_{k+1} \big] + \mathbb{E}\big[\sum_{j=1}^{r-w_{\widehat{\mathsf{N}}}-1}j \mathbb{P}\big(w_{\widehat{\mathsf{C}}}=j\big)\mathbb{E}\big[ LU \mid U > U^{\widehat{\mathsf{N}}}_{k+1}\big] \big] \\
    & \textstyle + \mathbb{E}\big[\sum_{j=r-w_{\widehat{\mathsf{N}}}}^{|\widehat{\mathsf{C}}|}\mathbb{P}\big(w_{\widehat{\mathsf{C}}}=j\big)\mathbb{E}\big[\sum_{i=1}^{r-w_{\widehat{\mathsf{N}}}}  LU_i^j \mid  U_j^j > U^{\widehat{\mathsf{N}}}_{k+1}\big] \big]
\end{align*}
\begin{align*}
    &\textstyle \ge  \frac{L(1+U^{\widehat{\mathsf{N}}}_{k+1})}{2}\mathbb{E}\big[w_{\widehat{\mathsf{N}}} + \sum_{j=1}^{r-w_{\widehat{\mathsf{N}}}-1}j \mathbb{P}\big(w_{\widehat{\mathsf{C}}}=j\big)\big] + L\mathbb{E}\big[\sum_{j=r-w_{\widehat{\mathsf{N}}}}^{|\widehat{\mathsf{C}}|}\mathbb{P}\big(w_{\widehat{\mathsf{C}}}=j\big)(r-w_{\widehat{\mathsf{N}}}) \big(1-\frac{(1-U^{\widehat{\mathsf{N}}}_{k+1})(r-w_{\widehat{\mathsf{N}}}+1)}{2(j+1)}\big) \big] \\
    &= \textstyle \frac{L(1+U^{\widehat{\mathsf{N}}}_{k+1})}{2} \mathbb{E}\big[w_{\widehat{\mathsf{N}}} + \sum_{j=1}^{r-w_{\widehat{\mathsf{N}}}-1}j \mathbb{P}\big(w_{\widehat{\mathsf{C}}}=j\big) + \sum_{j=r-w_{\widehat{\mathsf{N}}}}^{|\widehat{\mathsf{C}}|}\mathbb{P}\big(w_{\widehat{\mathsf{C}}}=j\big)\frac{2(r-w_{\widehat{\mathsf{N}}})}{(1+U^{\widehat{\mathsf{N}}}_{k+1})} \big(1-\frac{(1-U^{\widehat{\mathsf{N}}}_{k+1})(r-w_{\widehat{\mathsf{N}}}+1)}{(j+1)}\big) \big]
\end{align*}
The second inequality is because $Q(U) = V \geq LU$ by Lemma \ref{Lemma:LUbound} and $F(p) \ge F(V^{\widehat{\mathsf{N}}}_{k+1}) = U^{\widehat{\mathsf{N}}}_{k+1}$. The third equality uses $\mathbb{E}\left[U \mid U \geq U^{\widehat{\mathsf{N}}}_{k+1}\right] = (1+U^{\widehat{\mathsf{N}}}_{k+1})/2$ for $U \sim \text{Uniform}[0,1]$, and $\mathbb{E}\left[U^q_i \mid U^q_q \geq U^{\widehat{\mathsf{N}}}_{k+1}\right] = \left(1 -U^{\widehat{\mathsf{N}}}_{k+1}\right) \mathbb{E}\left[U^q_i\right] + U^{\widehat{\mathsf{N}}}_{k+1}$. The last equality holds because the summation of expected values of uniform order statistics is $\sum_{i=1}^{r-k}\mathbb{E}\left[U^q_i\right] = {\scriptstyle (r-k) \frac{(2q -r+k+1)}{2(q+1)}}$ when $r-k \leq q$. Finally, note that $U^{\widehat{\mathsf{N}}}_{k+1} \sim \text{Beta}\left(|\widehat{\mathsf{N}}|-k, k+1\right)$ and is the $(k+1)$-th highest order statistic from $|\widehat{\mathsf{N}}|$ i.i.d. standard uniform random variables.
\end{proof}
Let $\mathtt{RevCPoP}(p, \widehat{\mathsf{N}}, \widehat{\mathsf{C}}, \mathbf{B}^{\widehat{\mathsf{N}}}, \mathbf{B}^{\widehat{\mathsf{C}}})$ denote the revenue realized by running Phases 1 and 2 from Algorithm \ref{alg:cpop} using the price $p$, analogous to the welfare $\mathtt{CPoP}(p, \widehat{\mathsf{N}}, \widehat{\mathsf{C}}, \mathbf{B}^{\widehat{\mathsf{N}}}, \mathbf{B}^{\widehat{\mathsf{C}}})$ for the same price. Assuming truthful bidding ($\mathbf{B} = \mathbf{V}$), since the utility of all bidders and the auctioneer is non-negative, $\mathtt{CPoP}(p, \widehat{\mathsf{N}}, \widehat{\mathsf{C}}, \mathbf{V}^{\widehat{\mathsf{N}}}, \mathbf{V}^{\widehat{\mathsf{C}}}) \geq \mathtt{RevCPoP}(p, \widehat{\mathsf{N}}, \widehat{\mathsf{C}}, \mathbf{V}^{\widehat{\mathsf{N}}}, \mathbf{V}^{\widehat{\mathsf{C}}})$. Hence, a minorant of the revenue is also a valid minorant of the welfare.

\begin{prop}[Minorant of Revenue] \label{Lemma:RevenueMinorant}
  Let Assumptions \ref{assumption:iid} and \ref{Assump:Qfunc} hold. Then, 
\(    \textstyle \mathbb{E}\left[\mathtt{RevCPoP}(p, \widehat{\mathsf{N}}, \widehat{\mathsf{C}}, \mathbf{V}^{\widehat{\mathsf{N}}}, \mathbf{V}^{\widehat{\mathsf{C}}})\right] \geq \mathbb{E}\big[L \cdot U^{\widehat{\mathsf{N}}}_{k+1} \mathbb{E}\big[w_{\widehat{\mathsf{N}}} + \sum_{j=1}^{|\widehat{\mathsf{C}}|}\min \{j, r-w_{\widehat{\mathsf{N}}}\} \mathbb{P}\big(w_{\widehat{\mathsf{C}}}=j\big) \big| U^{\widehat{\mathsf{N}}}_{k+1}\big] \big], \)
  where $w_{\widehat{\mathsf{N}}} \sim \mathrm{Binom}\big(k,\frac{1-\mathsf{F}(p)}{1-U^{\widehat{\mathsf{N}}}_{k+1}}\big)$, $w_{\widehat{\mathsf{C}}} \sim \mathrm{Binom}\big(|\widehat{\mathsf{C}}|,1-\mathsf{F}(p)\big)$ and \(U^{\widehat{\mathsf{N}}}_{k+1} \sim \mathrm{Beta}(|\widehat{\mathsf{N}}|-k,k+1)\).
\end{prop}
\begin{proof}
The calculations are very similar to Proposition \ref{Lemma:WelfareMinorant}.
\end{proof}
\section{Missing Proofs}

Proofs of results from the main text are given here.

\subsection{Proof of Proposition \ref{prop:VCGeq}}

By Lemma \ref{Lemma:colluderVCG}, $r_\mathsf{C}^{\scriptstyle col} \leq r_\mathsf{C}^*$. Hence, $r_\mathsf{N}^{\scriptstyle col} \geq r_\mathsf{N}^*$ and $v^\mathsf{C}_i > b^{\mathsf{N}}_{r_\mathsf{N}^{\scriptstyle col}+1}$ (the payment price for winning items). The optimal winning bids are $b^{\mathsf{C}, \scriptstyle col}_{i} = \max\{v^\mathsf{C}_i, b^{\mathsf{N}}_{r_\mathsf{N}^{\scriptstyle col}+1} + \epsilon\} = v^\mathsf{C}_i$ for $i \in \{1, \dots, r_\mathsf{C}^{\scriptstyle col} \}$ and losing bids $b^{\mathsf{C}, \scriptstyle col}_{i} = 0$ for $i \in \{r_\mathsf{C}^{\scriptstyle col}+1, \dots, |\mathsf{C}| \}$. That is, colluding bidders are incentivized to bid as low as possible to win, and always bid either $v^\mathsf{C}_i$ or $0$ ensuring $b^{\mathsf{C},\scriptstyle col}_i \leq v^\mathsf{C}_i$ (bid shading).

Since $r_\mathsf{N}^ {\scriptstyle col} \geq r_\mathsf{N}^*$, the non-colluders' allocation is non-decreasing: $x_i(\mathbf{b}^{\scriptstyle col}) \geq x_i(\mathbf{b}^*)$ for $i\in \mathsf{N}$. The collusive payment is $p(\mathbf{b}^{\scriptstyle col}) = b^{\mathsf{N}}_{r_{\mathsf{N}}^{\scriptstyle col}+1}$, and truthful payment is $p(\mathbf{b}^{*}) = \max\{v^\mathsf{C}_{r_\mathsf{C}^*+1}, b^{\mathsf{N}}_{r_{\mathsf{N}}^*+1}\}$. Since $r_{\mathsf{N}}^{\scriptstyle col} \geq r_{\mathsf{N}}^*$, $b^{\mathsf{N}}_{r_{\mathsf{N}}^{\scriptstyle col}+1} \leq b^{\mathsf{N}}_{r_{\mathsf{N}}^*+1}$, so $p(\mathbf{b}^{\scriptstyle col}) \leq p(\mathbf{b}^*)$. We thus have $u_i(\mathbf{b}^{\scriptstyle col}) = x_i(\mathbf{b}^{\scriptstyle col}) v_i - p(\mathbf{b}^{\scriptstyle col}) \geq x_i(\mathbf{b}^*) v_i - p(\mathbf{b}^*) = u_i(\mathbf{b}^{*})$ for the non-colluder utility, $u_\mathsf{C}(\mathbf{b}^{\scriptstyle col}) = u_\mathsf{C}(\mathbf{v}^{\mathsf{N}}; r^{\scriptstyle col}_\mathsf{C}) \geq u_\mathsf{C}(\mathbf{v}^{\mathsf{N}}; r^{*}_\mathsf{C}) = u_\mathsf{C}(\mathbf{b}^{*})$ by definition of $r^{\scriptstyle col}_\mathsf{C}$ for the colluder utility, and $u_a(\mathbf{b}^{\scriptstyle col}) = r \cdot p(\mathbf{b}^{\scriptstyle col}) \leq r \cdot p(\mathbf{b}^*) = u_a(\mathbf{b}^*)$ for the auctioneer's revenue.

VCG without collusion maximizes social welfare, $\mathtt{Wel}(\mathbf{b}^*)$. Any deviation $\mathbf{b}^{\scriptstyle col}$ that results in a different allocation must yield $\mathtt{Wel}(\mathbf{b}^{\scriptstyle col}) \leq \mathtt{Wel}(\mathbf{b}^*)$. This, combined with $\mathtt{Rev}(\mathbf{b}^{\scriptstyle col}) \leq \mathtt{Rev}(\mathbf{b}^*)$ (from the auctioneer's utility), confirms that both welfare and revenue are non-increasing.

\subsection{Proof of Proposition \ref{thm:BK}}
    First, we prove $\mathtt{Wel}_{\textsc{VCG}}(\mathsf{N},\mathsf{C}) \geq \mathtt{Wel}_{\textsc{VCG}}(\mathsf{N},\emptyset)$ and $\mathtt{Rev}_{\textsc{VCG}}(\mathsf{N},\mathsf{C}) \geq \mathtt{Rev}_{\textsc{VCG}}(\mathsf{N},\emptyset)$. From Proposition \ref{prop:VCGeq}, when all $\mathsf{N} \cup \mathsf{C}$ bidders participate in a VCG mechanism, the non-colluding bidders bid $\mathbf{v}^\mathsf{N}$ (i.e., truthfully) and the colluding bidders bid $\mathbf{b}^{\mathsf{C},\scriptstyle col} \leq \mathbf{v}^\mathsf{C}$ (i.e., either bid shade to 0 or report their true valuation). The welfare achieved is $\mathtt{Wel}(\mathbf{v}^\mathsf{N}, \mathbf{b}^{\mathsf{C},\scriptstyle col})$. Since the VCG mechanism allocates items to the top bidders, when a colluding bidder $i$ bid shades to $0$, an item could be allocated to a bidder $j$ who values it less than $i$ leading to a welfare loss. In the worst case, all the colluders bid shade to $0$ achieving a welfare $\mathtt{Wel}(\mathbf{v}^\mathsf{N}, \mathbf{0})$. Hence,  $\mathtt{Wel}(\mathbf{v}^\mathsf{N}, \mathbf{b}^{\mathsf{C},\scriptstyle col}) \geq \mathtt{Wel}(\mathbf{v}^\mathsf{N}, \mathbf{0})$. Now, consider a VCG auction with only the non-colluding bidders $\mathsf{N}$ achieving a welfare $\mathtt{Wel}(\mathbf{v}^\mathsf{N})$. Since the bidders have non-negative valuations, adding $|\mathsf{C}|$ more bids of $0$ cannot hurt the welfare, that is $\mathtt{Wel}(\mathbf{v}^\mathsf{N}) \leq \mathtt{Wel}(\mathbf{v}^\mathsf{N}, \mathbf{0})$. The result follows from $\mathtt{Wel}(\mathbf{v}^\mathsf{N}, \mathbf{b}^{\mathsf{C},\scriptstyle col}) \geq \mathtt{Wel}(\mathbf{v}^\mathsf{N}, \mathbf{0}) \geq \mathtt{Wel}(\mathbf{v}^\mathsf{N})$. The proof for revenue is similar. When $|\mathsf{N}| > r$, the price drops to $v^\mathsf{N}_{r+1}$ in the worst case when all colluding bidders bid shade to $0$, and the revenue $\mathtt{Rev}_{\textsc{VCG}}(\mathsf{N},\mathsf{C}) \geq v^\mathsf{N}_{r+1} r= \mathtt{Rev}_{\textsc{VCG}}(\mathsf{N},\emptyset)$. When $|\mathsf{N}| \leq r$, $\mathtt{Rev}_{\textsc{VCG}}(\mathsf{N},\emptyset) = 0$ and $\mathtt{Rev}_{\textsc{VCG}}(\mathsf{N}, \mathsf{C}) \geq \mathtt{Rev}_{\textsc{VCG}}(\mathsf{N},\emptyset)$ holds trivially.

From the claim proved above, $\mathtt{Wel}_{\textsc{VCG}}(\mathsf{N} \cup\mathsf{N}',\mathsf{C}) \geq \mathtt{Wel}_{\textsc{VCG}}(\mathsf{N}\cup\mathsf{N}',\emptyset)$. Now suppose $\mathsf{N}'_{\scriptstyle col}$ are the bidders in $\mathsf{N}$ who have decided to collude. Since VCG by construction achieves the highest possible welfare when no bidders collude, this means we must have $\mathtt{Wel}_{\textsc{VCG}}(\mathsf{N}\cup\mathsf{N}',\emptyset) \geq \mathtt{Wel}_{\textsc{OPT}}(\mathsf{N}\cup\mathsf{N}'_{\scriptstyle col})$. Since the bidders are assumed to have true valuations from an i.i.d. distribution, $\mathbb{E}[\mathtt{Wel}_{\textsc{OPT}}(\mathsf{N}\cup\mathsf{N}'_{\scriptstyle col})] = \mathbb{E}[\mathtt{Wel}_{\textsc{OPT}}(\mathsf{N}\cup\mathsf{C})]$. Putting this all together, we get \(\mathbb{E}[\mathtt{Wel}_{\textsc{VCG}}(\mathsf{N} \cup \mathsf{N}', \mathsf{C})] \geq \mathbb{E}[\mathtt{Wel}_{\textsc{OPT}}(\mathsf{N}\cup\mathsf{C})]\), which is the desired result.

\subsection{Proof of Proposition \ref{thm:ppeq}}
For a given Posted Price \(p\), let \( x_{\mathsf C}(\mathbf{b}^\mathsf{C})= \sum_{i \in \mathsf C} \mathbf{1}(b^{\mathsf C}_{i} > p) \) and \(x_{\mathsf N} = \sum_{i \in \mathsf N} \mathbf{1}(v^{\mathsf N}_{i} > p) \) be the number of colluder and non-colluder bids above the price, respectively. Let \(\mathbb{P}(r_\mathsf{C} | x_{\mathsf C}(\mathbf{b}^\mathsf{C}), x_{\mathsf N})\) be the probability that colluders receive \(r_\mathsf{C}\) items. Then, the colluders choose their bids to maximize their expected utility \(\textstyle u_\mathsf{C}(\mathbf{v}^{\mathsf{N}}, \mathbf{b}^{\mathsf{C}}) = \sum_{r_{\mathsf{C}}} \mathbb{P}(r_\mathsf{C} | x_{\mathsf C}(\mathbf{b}^\mathsf{C}), x_{\mathsf N}) \sum_{i = 1}^{r_\mathsf{C}} v^\mathsf{C}_{i} - r_\mathsf{C} \cdot p\). First, notice that if colluders underbid, i.e. $b^{\mathsf{C}, \scriptstyle under}_i \leq v^{\mathsf{C}}_i$ for $i \in \mathsf{C}$, then \( u_\mathsf{C}(\mathbf{v}^{\mathsf{N}}, \mathbf{b}^{\mathsf{C}, {\scriptstyle under}}) \leq u_\mathsf{C}(\mathbf{v}^{\mathsf{N}}, \mathbf{v}^{\mathsf{C}}) \). This is because underbidding ($b^{\mathsf{C}, \scriptstyle under}_i \leq p$ for $v^{\mathsf{C}}_i > p$) is strictly suboptimal. $x_{\mathsf C} ( \mathbf{v}^{\mathsf{C}}) \geq x_{\mathsf C} ( \mathbf{b}^{\mathsf{C}, {\scriptstyle under}})$, and hence \(\mathbb{P}(r_\mathsf{C} | x_{\mathsf C}(\mathbf{b}^{\mathsf{C}, {\scriptstyle under}}), x_{\mathsf N})\) 
has a probability distribution of allocated items ($r_\mathsf{C}$) that is shifted downwards when compared to \(\mathbb{P}(r_\mathsf{C} | x_{\mathsf C}(\mathbf{v}^{\mathsf{C}}), x_{\mathsf N})\). Since the marginal utility of winning items for these bidders is strictly positive ($v^{\mathsf{C}}_i - p > 0$), a downward shift in $r_\mathsf{C}$ strictly lowers expected utility. Hence, colluders do not underbid. They, in fact, benefit from overbidding (for instance, when at least \(r\) of them value the item above price \(p\) then it is optimal for all of them to overbid above the price). Therefore, colluders may get more items ($r_\mathsf{C}^{\scriptstyle col} \geq r_\mathsf{C}^*$) by overbidding ($b^{\mathsf{C}, \scriptstyle col}_i \geq v^{\mathsf{C}\vphantom{,\scriptstyle col}}_i$ for $i \in \mathsf{C}$), bidding either a value above the price $p$ or their true valuation $v_i^\mathsf{C}$.

Since the optimal colluding bidding strategy maximizes their utility by overbidding to gain more items, it improves the colluders' utility at the cost of non-colluders' utility. That is, $u_i(\mathbf{b}^*) \geq u_i(\mathbf{b}^{\scriptstyle col})$ for $i \in \mathsf{N}$ and $u_\textsc{c}(\mathbf{b}^*) \leq u_\textsc{c}(\mathbf{b}^{\scriptstyle col})$. Also, overbidding ensures that \( x_{\mathsf C}(\mathbf{b}^{\mathsf{C}, {\scriptstyle col}}) + x_{\mathsf N} \geq x_{\mathsf C}(\mathbf{v}^{\mathsf{C}}) + x_{\mathsf N}\), and the number of items sold does not decrease. The number of items sold does not increase either because colluders only overbid when the total number of bidders valuing the item above the price are at least \(r\). If there are fewer bidders than \(r\), then overbidding gives them extra items at a negative marginal utility ($v^{\mathsf{C}}_i - p < 0$). Therefore, $\mathtt{Rev}(\mathbf{b}^*) = \mathtt{Rev}(\mathbf{b}^{\scriptstyle col})$. In contrast, the welfare could decrease as since the colluders have manipulated their bids to gain more items, and could be taking these items away from non-colluders who value it more. Hence, $\mathtt{Wel}(\mathbf{b}^*) \geq \mathtt{Wel}(\mathbf{b}^{\scriptstyle col})$.

\subsection{Proof of Proposition \ref{prop:ppbad}}

Consider the random variable $X_\mathsf{C} = |\{i \in \mathsf{C} : V_i > p\}|$, and note this is a binomial random variable with $c$ trials and success probability $q$. The Chernoff bound for a binomial random variable gives $\mathbb P(X_\mathsf{C}\leq(1-\delta)cq)\leq \exp(-\delta^2cq/2)$ with $\delta=1-r/cq$. This means $\mathbb{P}(X_\mathsf{C} \leq r) \leq \exp(-(cq-r)^2/2cq) \leq \exp(r-cq/2)$. Similar bounds hold for the random variable $X = |\{i \in \mathsf{M} : V_i > p \}|$ and $X_\mathsf{N} = |\{i \in \mathsf{N} : V_i > p\}|$ with $\mathbb{P}(X \leq r) \leq \exp(r-(n+c)q/2)$ and $\mathbb{P}(X_\mathsf{N} \leq r) \leq \exp(r-nq/2)$, respectively.

Now suppose $X_\mathsf{C} \geq r$, and note that in this event the optimal bidding strategy is for every colluder to overbid above the Posted Price $p$. On the other hand, Posted Price is DSIC for non-colluders. Since items are allocated randomly among all bids above the Posted Price, we have
\begin{equation}
\textstyle    \mathbb{E}\bigl[\mathtt{NWel}_p(\mathsf{N}, \mathsf{C}) \mid X_\mathsf{C} \geq r\bigr] = \mathbb{E}\bigl[V \mid V > p\bigr]\cdot r \cdot \mathbb{E}\bigl(X_\mathsf{N}/(X_\mathsf{N} + c)\bigr) \leq \mathbb{E}\bigl[V \mid V > p\bigr]\cdot r\cdot qn/(qn + c),
\end{equation}
where the last inequality follows from Jensen's inequality. In the case when $X_\mathsf{C} < r$, we have the bound
\begin{equation}
\textstyle    \mathbb{E}\bigl[\mathtt{NWel}_p(\mathsf{N}, \mathsf{C}) \mid X_\mathsf{C} < r\bigr] \leq \mathbb{E}\bigl[V \mid V > p\bigr]\cdot \mathbb{E}\bigl(X_\mathsf{N}\bigr) = \mathbb{E}\bigl[V \mid V > p\bigr]\cdot qn,
\end{equation}
Combining, we have that unconditionally:
\begin{equation}
\textstyle    \mathbb{E}\bigl[\mathtt{NWel}_p(\mathsf{N}, \mathsf{C}) \bigr] \leq \mathbb{E}\bigl[V \mid V > p\bigr]\cdot qn \cdot \exp(r-cq/2) + \mathbb{E}\bigl[V \mid V > p\bigr]\cdot r\cdot qn/(qn + c).
\end{equation}
But because all valuations are i.i.d., we have
\begin{equation}
    \mathbb{E}\bigl[\mathtt{NWel}_p(\mathsf{N}\cup \mathsf{C},\emptyset)\bigr] = \mathbb{E}\bigl[V \mid V > p\bigr]\cdot \mathbb{E}\bigl(\min\{X,r\}\bigr)\cdot n/(n + c) \geq \mathbb{E}\bigl[V \mid V > p\bigr]\cdot r \cdot \bigl(1-\exp(r-(n+c)q/2)\bigr)\cdot n/(n + c).
\end{equation}
Combining this with the above yields
\begin{multline}
\mathbb{E}\bigl[\mathtt{NWel}_p(\mathsf{N}, \mathsf{C})\bigr]/\mathbb{E}[\mathtt{NWel}_p(\mathsf{N}\cup \mathsf{C},\emptyset)] \leq 1/(1-\exp(r-(n+c)q/2))\cdot (qn\cdot(n+c)/(n\cdot(qn+c)) + \\
\exp(r-cq/2)\cdot(qn \cdot (n+c))/(n \cdot r)) = \\
1/(1-\exp(r-(n+c)q/2))\cdot (1-c(1-q)/(qn+c)+\exp(r-cq/2)\cdot (q\cdot(n+c)/r)).
\end{multline}

\subsection{Proof of Theorem \ref{Lemma:dsic}}
\paragraph{Non-colluders.}Consider bidder $i \in \mathsf{N}$ and fix the bids from other bidders $\mathbf{b}_{-i}$. First, we characterise how changes in the bidder $i$'s reported bid affects the outcome of the DP oracle~\ref{alg:kporacle} in the next lemma.

\begin{lemma}\label{lemma:k-charac}
    Fix $i \in \mathsf{N}$, $\mathbf{b}^{\mathsf{N}}_{-i}$. Consider two possible bid values \(b_i\) and \(b'_i\) s.t. \(b'_i > b_i\). Define \(\mathbf{b}^\mathsf{N} = (b_i, \mathbf{b}^{\mathsf{N}}_{-i})\), \(\mathbf{b}^{'\mathsf{N}} = (b'_i, \mathbf{b}^{\mathsf{N}}_{-i})\),  \(i_{b_i} = \sum_{j \in \mathsf{N}} \mathbbold{1}(b_j \ge b_i)\) and \(i_{b'_i} = \sum_{j \in \mathsf{N}} \mathbbold{1}(b_j \ge b'_i)\). Then the following holds: \\
    (1) If \(k^*(\mathbf{b}^{\mathsf{N}}, c) \ge i_{b_i} \), then \(k^*(\mathbf{b}^{'\mathsf{N}}, c) = k^*(\mathbf{b}^{\mathsf{N}}, c)\) and \(p^*(\mathbf{b}^{'\mathsf{N}}, c) = p^*(\mathbf{b}^{\mathsf{N}}, c)\). \\
    (2) If  \(k^*(\mathbf{b}^{\mathsf{N}}, c) < i_{b_i}\), then \(k^*(\mathbf{b}^{'\mathsf{N}}, c) <  \min\{i_{b_i}, r+1\}\) and \(p^*(\mathbf{b}^{'\mathsf{N}}, c) \ge  \max\{b_i, b^{\mathsf{N}}_{r+1}\}\).
\end{lemma}
\begin{proof}
     In Phase 2 of the DP procedure, the while loop condition \(M(p_k, \overline{\mathsf{B}^\mathsf{N}_k})< 
        \mathbb{E}\left[V( p_{k-1}, \overline{\mathcal{B}^\mathsf{N}_{k-1}}) \mid \overline{\mathsf{B}^\mathsf{N}_k}\right]\) is met at \(k^*(\mathbf{b}^{\mathsf{N}}, c) \) by definition. We consider the following two cases:\\
\noindent (1) \(k^*(\mathbf{b}^{\mathsf{N}}, c) \ge i_{b_i}\): The top bid values \(\{b^{\mathsf{N}}_1, \dots,  b^{\mathsf{N}}_{k^*(\mathbf{b}^{\mathsf{N}}, c)}\}\) are \textit{not} used in the phase 2 of DP procedure by construction. If \(k^*(\mathbf{b}^{\mathsf{N}}, c) \ge i_{b_i}\), then both \(b'_i, b_i\) remain in the top \(k^*(\mathbf{b}^{\mathsf{N}}, c)\) bids. Hence, the increase in bid value from $b_i$ to $b'_i$ does \textit{not} affect the DP outcome and \(k^*(\mathbf{b}^{'\mathsf{N}}, c) = k^*(\mathbf{b}^{\mathsf{N}}, c)\) and \(p^*(\mathbf{b}^{'\mathsf{N}}, c) = p^*(\mathbf{b}^{\mathsf{N}}, c)\).\\
\noindent (2) \(k^*(\mathbf{b}^{\mathsf{N}}, c) < i_{b_i}\): By definition, \(i_{b'_i} < i_{b_i}\) since \(b'_i > b_i\). The while loop in DP's phase 2 evaluates the same conditions \(M(p_k, \overline{\mathsf{B}^\mathsf{N}_k})< \mathbb{E}\left[V( p_{k-1}, \overline{\mathcal{B}^\mathsf{N}_{k-1}}) \mid \overline{\mathsf{B}^\mathsf{N}_k}\right]\) for all \(k > \max\{i_{b_i}, i_{b'_i}\} = i_{b_i}\) irrespective of whether bids \(\mathbf{b}^{\mathsf{N}}\) or \(\mathbf{b}^{'\mathsf{N}}\) are used. Since \(k^*(\mathbf{b}^{\mathsf{N}}, c) < i_{b_i}\), we also know these while loop conditions are all met for \(k > i_{b_i}\), and for both bid vectors \(\mathbf{b}^{\mathsf{N}}\) and \(\mathbf{b}^{'\mathsf{N}}\). Hence, \(k^*(\mathbf{b}^{'\mathsf{N}}, c) <  \min\{i_{b_i}, r+1\}\) and \(p^*(\mathbf{b}^{'\mathsf{N}}, c) \ge  \max\{b^{\mathsf{N}}_{i_{b_i}}, b^{\mathsf{N}}_{r+1}\}\) for bids \(\mathbf{b}^{'\mathsf{N}}\). Also note that \(k^*(\mathbf{b}^{\mathsf{N}}, c) <  \min\{i_{b_i}, r+1\}\) and \(p^*(\mathbf{b}^{\mathsf{N}}, c) \ge  \max\{b^{\mathsf{N}}_{i_{b_i}}, b^{\mathsf{N}}_{r+1}\}\) is always true when  \(k^*(\mathbf{b}^{\mathsf{N}}, c) < i_{b_i}\).
\end{proof}

For bid \(b_i\), the corresponding utility of bidder \(i\) is $u_i(b_i, \mathbf{b}_{-i}) = v_i x_i(b_i, \mathbf{b}_{-i}) - p_i(b_i, \mathbf{b}_{-i}) $. Let \(v_i\) be the true valuation of bidder $i$. We show that $u_i(v_i, \mathbf{b}_{-i}) \ge u_i(b_i, \mathbf{b}_{-i}) $ for all possible bid values $b_i$, and hence that the mechanism is DSIC for bidder $i$ by definition. The C-PoP mechanism uses the DP oracle~\ref{alg:kporacle} to calculate $k^*(\mathbf{b}^{\mathsf{N}}, c)$ and $p^*(\mathbf{b}^{\mathsf{N}}, c)$. The allocation rule for a non-colluding bidder is \(x_i(b_i, \mathbf{b}_{-i}) = \mathbbold{1}(b_i > p^*(\mathbf{b}^{\mathsf{N}}, c))\) and \(p_i(b_i, \mathbf{b}_{-i}) = p^*(\mathbf{b}^{\mathsf{N}}, c) \mathbbold{1}(b_i > p^*(\mathbf{b}^{\mathsf{N}}, c)) \). The DP oracle always sets \(p^*(\mathbf{b}^{\mathsf{N}}, c) \ge b^{\mathsf{N}}_{k^*(\mathbf{b}^{\mathsf{N}}, c)+1}\), where \(k^*(\mathbf{b}^{\mathsf{N}}, c)\) is set through phase 2 of its algorithm.

Define \(\mathbf{v}^\mathsf{N} = (v_i, \mathbf{b}^{\mathsf{N}}_{-i})\), \(\mathbf{b}^{\mathsf{N} } = (b_i, \mathbf{b}^{\mathsf{N}}_{-i})\), \(\mathbf{v}^{\mathsf{M}} = (\mathbf{v}^{\mathsf{N}}, \mathbf{b}^{\mathsf{C}})\), \(\mathbf{b}^{\mathsf{M}} = (\mathbf{b}^{\mathsf{N}}, \mathbf{b}^{\mathsf{C}})\), \(i_{v_i} = \sum_{j \in \mathsf{N}} \mathbbold{1}(b_j \ge v_i)\) and \(i_{b_i} = \sum_{j \in \mathsf{N}} \mathbbold{1}(b_j \ge b_i)\). We consider the following four cases.

\paragraph{Case 1: \(b_i > v_i\) and \(k^*(\mathbf{v}^{\mathsf{N}},c) \ge i_{v_i}\).} By Lemma \ref{lemma:k-charac} (1), \(k^*(\mathbf{b}^{\mathsf{N}}, c) = k^*(\mathbf{v}^{\mathsf{N}}, c)\) and \(p^*(\mathbf{b}^{\mathsf{N}}, c) = p^*(\mathbf{v}^{\mathsf{N}}, c)\). Since \(b_i > v_i\) and \(p^*(\mathbf{b}^{\mathsf{N}}, c) = p^*(\mathbf{v}^{\mathsf{N}}, c)\), it holds that \(x_i(\mathbf{b}^{\mathsf{N}}) \ge x_i(\mathbf{v}^{\mathsf{N}})\) and \(p_i(\mathbf{b}^{\mathsf{N}}) \ge p_i(\mathbf{v}^{\mathsf{N}})\). The utility with bids \(b_i\) and \(v_i\) are
\begin{align*}
    & u_i(\mathbf{b}^{\mathsf{M}}) = \big(v_i -p^*(\mathbf{v}^{\mathsf{N}}, c)\big) \cdot \mathbbold{1}\big(b_i > p^*(\mathbf{v}^{\mathsf{N}}, c)\big) \\
    & u_i(\mathbf{v}^{\mathsf{M}}) = \big(v_i -p^*(\mathbf{v}^{\mathsf{N}}, c)\big) \cdot \mathbbold{1}\big(v_i > p^*(\mathbf{v}^{\mathsf{N}}, c)\big)
\end{align*}
If \(b_i > v_i > p^*(\mathbf{v}^{\mathsf{N}}, c)\), then \(u_i(\mathbf{b}^{\mathsf{M}}) = u_i(\mathbf{v}^{\mathsf{M}}) = v_i - p^*(\mathbf{v}^{\mathsf{N}}, c)\). If \(p^*(\mathbf{v}^{\mathsf{N}}, c) > b_i >v_i\), then \(u_i(\mathbf{b}^{\mathsf{M}}) = u_i(\mathbf{v}^{\mathsf{M}}) = 0\). If \(b_i> p^*(\mathbf{v}^{\mathsf{N}}, c) >v_i\), then \(u_i(\mathbf{b}^\mathsf{M}) = v_i - p^*(\mathbf{v}^{\mathsf{N}}, c) < 0\) and \(u_i(\mathbf{v}^{\mathsf{M}}) = 0\). Hence, \(u_i(\mathbf{v}^{\mathsf{M}}) \ge u_i(\mathbf{b}^{\mathsf{M}}) \) holds in this case.

\paragraph{Case 2: \(b_i > v_i\) and \(k^*(\mathbf{v}^{\mathsf{N}},c) < i_{v_i}\).} By lemma~\eqref{lemma:k-charac} (2), \(k^*(\mathbf{b}^{\mathsf{N}}, c) <  \min\{i_{v_i}, r+1\}\) and \(p^*(\mathbf{b}^{\mathsf{N}}, c) \ge  \max\{v_i, v^{\mathsf{N}}_{r+1}\}\). Since, \(k^*(\mathbf{v}^{\mathsf{N}},c) < i_{v_i}\), no item is allocated to bidder $i$ when they bid \(v_i\), i.e., \(x_i(\mathbf{v}^{\mathsf{N}}) = 0\), \(p_i(\mathbf{v}^{\mathsf{N}}) = 0\) and \(u_i(\mathbf{v}^{\mathsf{M}}) = 0\). The utility with bid \(b_i\)
\begin{align*}
    u_i(\mathbf{b}^{\mathsf{M}}) & = \big(v_i - p^*(\mathbf{b}^{\mathsf{N}})\big) \cdot \big(\mathbbold{1}(b_i >  p^*(\mathbf{b}^{\mathsf{N}})\big) \\
    & \le v_i - p^*(\mathbf{b}^{\mathsf{N}}) \le v_i - \max\{v_i, v^{\mathsf{N}}_{r+1}\} \\
    & \le 0 = u_i(\mathbf{v}^{\mathsf{M}}).
\end{align*}

\paragraph{Case 3: \(b_i < v_i\) and \(k^*(\mathbf{b}^{\mathsf{N}},c) \ge i_{b_i}\).} From lemma~\eqref{lemma:k-charac} (1), it is true that \(k^*(\mathbf{b}^{\mathsf{N}}, c) = k^*(\mathbf{v}^{\mathsf{N}}, c)\) and \(p^*(\mathbf{b}^{\mathsf{N}}, c) = p^*(\mathbf{v}^{\mathsf{N}}, c)\). Then, comparing the utilities, 
\begin{align*}
    u_i(\mathbf{b}^{\mathsf{M}}) & = \big(v_i -p^*(\mathbf{v}^{\mathsf{N}}, c)\big) \cdot \mathbbold{1}\big(b_i > p^*(\mathbf{v}^{\mathsf{N}}, c)\big) \\
    & \le \big(v_i -p^*(\mathbf{v}^{\mathsf{N}}, c)\big) \cdot \mathbbold{1}\big(v_i > p^*(\mathbf{v}^{\mathsf{N}}, c)\big) \quad (\because v_i > b_i) \\
    & = u_i(\mathbf{v}^{\mathsf{M}}).
\end{align*}

\paragraph{Case 4: \(b_i < v_i\) and \(k^*(\mathbf{b}^{\mathsf{N}},c) < i_{b_i}\). } By lemma~\eqref{lemma:k-charac} (2), \(k^*(\mathbf{b}^{\mathsf{N}}, c) <  \min\{i_{b_i}, r+1\}\) and \(p^*(\mathbf{b}^{\mathsf{N}}, c) \ge  \max\{b_i, b^{\mathsf{N}}_{r+1}\}\). Since, \(k^*(\mathbf{b}^{\mathsf{N}},c) < i_{b_i}\), no item is allocated to bidder $i$ when they bid \(b_i\), i.e., \(x_i(\mathbf{b}^{\mathsf{N}}) = 0\), \(p_i(\mathbf{b}^{\mathsf{N}}) = 0\) and \(u_i(\mathbf{b}^{\mathsf{M}}) = 0\). Hence, \(u_i(\mathbf{v}^{\mathsf{M}}) \ge u_i(\mathbf{b}^{\mathsf{M}})\) holds trivially.

\paragraph{Colluders.} Consider bidder $i \in \mathsf{C}$. The remaining $r-w_{\mathsf{N}}$ items, where \(w_{\mathsf{N}}\) is independent of the colluding bids, are allocated to colluders via the Posted Price $p^*(\mathbf{b}^\mathsf{N},c)$. In our model for colluding bidders, we consider them to be maximizing their total utility $\sum_{i \in \mathsf{C}} u_i(\mathbf{b}^{\mathsf{M}})$. Since they exchange monetary payments and items amongst themselves, it suffices to consider the top bidding colluders while calculating their net utility. Hence, $\sum_{i \in \mathsf{C}} u_i(\mathbf{b}^{\mathsf{M}}) = \sum_{i = 1}^{r-w_{\mathsf{N}}} \left(v^\mathsf{C}_i - p^*(\mathbf{b}^\mathsf{N},c)\right) \mathbbold{1}\big(b^\mathsf{C}_i >p^*(\mathbf{b}^\mathsf{N},c)\big)$. For any $p^*(\mathbf{b}^\mathsf{N},c)$, this net utility is maximized when $v^\mathsf{C}_i > p^*(\mathbf{b}^\mathsf{N},c) \Rightarrow \mathbbold{1}\big(b^\mathsf{C}_i > p^*(\mathbf{b}^\mathsf{N},c)\big) = 1$. Hence, bidding truthfully, i.e., $v_i^\mathsf{C} = b_i^\mathsf{C} $ $\forall i \in \mathsf{C}$ is a utility-maximizing strategy.

Hence, for every bidder, truthful reporting is a weakly dominant strategy, and the mechanism is (ex-post) DSIC.

\subsection{Proof of Proposition~\ref{prop:pricebound}}
\paragraph{Correctly classified non-colluders.} Consider \(i \in \mathsf{N}_{\textsc{t}}.\) Fix the bids from other bidders \(\mathbf{b}_{-i}\). Lemma~\ref{lemma:k-charac} (using \(\widehat{\mathsf{N}}\) instead of \(\mathsf{N}\), and \(|\widehat{\mathsf{C}}|\) instead of \(|\mathsf{C}|\)) can be used to characterize how bidder $i$'s reported bids affect the outcomes of the DP oracle. Subsequently, we can consider four cases with respect to the true valuation of bidder \(v_i\) as in the proof of theorem~\ref {Lemma:dsic} to prove truthfulness.

\paragraph{Incorrectly classified non-colluders.} Consider \(i \in \mathsf{N}_{\textsc{f}}.\) Their utility is \(u_i(\mathbf{b}^{\mathsf{M}}) = v_i - p^*(\mathbf{b}^{\widehat{\mathsf{N}}},|\widehat{\mathsf{C}}|) \mathbbold{1}(b_i >p^*(\mathbf{b}^{\widehat{\mathsf{N}}},|\widehat{\mathsf{C}}|))\), where \(b_i\),  \(v_i\) are their reported and true valuations respectively, and the price \(p^*(\mathbf{b}^{\widehat{\mathsf{N}}},|\widehat{\mathsf{C}}|)\) is independent of their reported bid value. Hence, their utility is maximized when \(v_i > p^*(\mathbf{b}^{\widehat{\mathsf{N}}},|\widehat{\mathsf{C}}|) \Rightarrow \mathbbold{1}\big(b_i > p^*(\mathbf{b}^{\widehat{\mathsf{N}}},|\widehat{\mathsf{C}}|)\big) = 1\) and truthful reporting \(v_i = b_i\) is a utility-maximizing strategy.

\paragraph{Price lower bound.} When \(|\mathsf{N}_{\textsc{t}}| \ge r+1\), the DP oracle is assured to set a price \(p^*(\mathbf{b}^{\widehat{\mathsf{N}}},|\widehat{\mathsf{C}}|) \ge b^{\widehat{\mathsf{N}}}_{r+1} \ge b^{\mathsf{N}_{\textsc{t}}}_{r+1}\).

\subsection{Proof of Proposition~\ref{lemma:VCG-PoPeq}}
If the C-PoP mechanism was run only on the set of predicted non-colluders \(\widehat{\mathsf{N}}\), then the mechanism would allocate all the $r$ items to bidders in \(\widehat{\mathsf{N}}\). The DP oracle would set the price using \(p_r\) as follows.
\begin{align}
    & p^*(\mathbf{b}^{\widehat{\mathsf{N}}}, 0) = p_r = \arg \max_{p \ge \mathbf{b}^{\widehat{\mathsf{N}}}_{r+1}} \mathbb{E}\Big[\sum_{i=1}^{r} {b}^{\widehat{\mathsf{N}}}_{i} \mathbbold{1}({b}^{\widehat{\mathsf{N}}}_{i} > {b}^{\widehat{\mathsf{N}}}_{r+1})\Big] = {b}^{\widehat{\mathsf{N}}}_{r+1}
\end{align}
Since the price is set at \({b}^{\widehat{\mathsf{N}}}_{r+1}\), the C-PoP mechanism would allocate items to the top $r$ bids in \(\widehat{\mathsf{N}}\). Therefore, it also holds that \(
    x_i(\mathbf{b}^{\widehat{\mathsf{N}}}, \emptyset) = \mathbbold{1}({b}^{\widehat{\mathsf{N}}}_{i} > {b}^{\widehat{\mathsf{N}}}_{r+1}).\) Hence, the C-PoP mechanism and VCG are equivalent when run on a set of classified non-colluders and without any bidders classified as colluders.

\subsection{Proof of Proposition \ref{lemma:conEW}}
Assume truthful bidding, i.e. $\mathbf{B}^{\widehat{\mathsf{N}}} = \mathbf{V}^{\widehat{\mathsf{N}}}$. First note that the expected welfare conditioned on all the non-colluding $\overline{\mathcal{V}^{\widehat{\mathsf{N}}}_{k}}$ is the same as that conditioned on the highest non-colluding losing bid $V^{\widehat{\mathsf{N}}}_{k +1}$. Also, since the colluding bidders can exchange both items and monetary payments among themselves, the specific allocation of items to individual bidders is irrelevant. For welfare computation, it suffices to consider the highest colluding bids among the winners. Using Assumption \ref{assumption:iid} and the probability integral transform of continuous random variables, $V^{\widehat{\mathsf{N}}}_k = \mathsf{Q}(U^{\widehat{\mathsf{N}}}_k)$ (see Equation $(2.3.7)$ of \cite{David_Nagaraja_2003}), where $U^{\widehat{\mathsf{N}}}_k$ is the $k$-th largest order statistic of $|\widehat{\mathsf{N}}|$ samples of standard uniform random variables. Using $\mathsf{Q}(\cdot) = \mathsf{F}^{-1}(\cdot)$, we get $\mathsf{F}(V^{\widehat{\mathsf{N}}}_{k}) = U^{\widehat{\mathsf{N}}}_{k}$. We use these facts to simplify the expression for the expected welfare \(\textstyle M(p, \overline{\mathcal{V}^{\widehat{\mathsf{N}}}_{\scriptstyle k}}) = \mathbb{E}\left[\mathtt{CPoP}(p, \widehat{\mathsf{N}}, \widehat{\mathsf{C}}, \mathbf{V}^{\widehat{\mathsf{N}}}, \mathbf{V}^{\widehat{\mathsf{C}}}) \mid p\ge V^{\widehat{\mathsf{N}}}_{k +1}, \overline{\mathcal{V}^{\widehat{\mathsf{N}}}_{k}}\right] = \mathbb{E}\left[\sum_{i \in \mathsf{M}} V_i x_i(\mathbf{V}^{\mathsf{M}}) \mid p\ge V^{\widehat{\mathsf{N}}}_{k +1}, V^{\widehat{\mathsf{N}}}_{k +1}\right]\), where \(\mathsf{M} = \widehat{\mathsf{N}} \cup \widehat{\mathsf{C}}\).

Let $w_{\widehat{\mathsf{N}}}$ be the number of non-colluding bids above the price $p$. Then, $w_{\widehat{\mathsf{N}}} \sim \mathrm{Binom}\big(k,\frac{1-\mathsf{F}(p)}{1-\mathsf{F}(V^{\widehat{\mathsf{N}}}_{k+1})}\big)$. Similarly, let $w_{\widehat{\mathsf{C}}}$ be the number of colluding bids above the price $p$, and $w_{\widehat{\mathsf{C}}} \sim \mathrm{Binom}\big(|\widehat{\mathsf{C}}|,1-\mathsf{F}(p)\big)$.
\begin{equation*}
\begin{aligned}
& \textstyle\mathbb{E}\big[\sum_{i \in \mathsf{M}} V_i x_i(\mathbf{V}^{\mathsf{M}})\ |\ p\ge V^{\widehat{\mathsf{N}}}_{k +1}, V^{\widehat{\mathsf{N}}}_{k+1}\big] \\ 
&= \textstyle \mathbb{E}\big[\sum_{j=1}^{w_{\widehat{\mathsf{N}}}} V^{k}_j \mid V^{k}_{k} > p \big] + \mathbb{E}\big[\sum_{j=1}^{r-w_{\widehat{\mathsf{N}}}-1}\mathbb{P}\big(w_{\widehat{\mathsf{C}}}=j\big)\mathbb{E}\big[\sum_{i=1}^{j}  V_i^j \mid  V_j^j > p\big] \big] \\
& \textstyle \quad+ \mathbb{E}\big[\sum_{j=r-w_{\widehat{\mathsf{N}}}}^{|\widehat{\mathsf{C}}|}\mathbb{P}\big(w_{\widehat{\mathsf{C}}}=j\big)\mathbb{E}\big[\sum_{i=1}^{r-w_{\widehat{\mathsf{N}}}}  V_i^j \mid  V_j^j > p\big] \big]\\
&= \textstyle \mathbb{E}\big[w_{\widehat{\mathsf{N}}}\big]\mathbb{E}\big[ V\ |\ V > p  \big] +  \mathbb{E}\big[\sum_{j=1}^{r-w_{\widehat{\mathsf{N}}}-1}j \mathbb{P}\big(w_{\widehat{\mathsf{C}}}=j\big)\mathbb{E}\big[ V \mid  V > p\big] \big] \\
& \textstyle \quad+ \mathbb{E}\big[\sum_{j=r-w_{\widehat{\mathsf{N}}}}^{|\widehat{\mathsf{C}}|}\mathbb{P}\big(w_{\widehat{\mathsf{C}}}=j\big)\mathbb{E}\big[\sum_{i=1}^{r-w_{\widehat{\mathsf{N}}}}  V_i^j \mid  V_j^j > p\big] \big]\\
&= \textstyle \mathbb{E}\big[w_{\widehat{\mathsf{N}}}\big]\mathbb{E}\big[ \mathsf{Q}(U)\ |\ U > \mathsf{F}(p)  \big] +  \mathbb{E}\big[\sum_{j=1}^{r-w_{\widehat{\mathsf{N}}}-1}j \mathbb{P}\big(w_{\widehat{\mathsf{C}}}=j\big)\mathbb{E}\big[ \mathsf{Q}(U) \mid  U > \mathsf{F}(p)\big] \big] \\
& \textstyle \quad+ \mathbb{E}\big[\sum_{j=r-w_{\widehat{\mathsf{N}}}}^{|\widehat{\mathsf{C}}|}\mathbb{P}\big(w_{\widehat{\mathsf{C}}}=j\big)\mathbb{E}\big[\sum_{i=1}^{r-w_{\widehat{\mathsf{N}}}}  \mathsf{Q}(U_i^j) \mid  U_j^j > \mathsf{F}(p)\big] \big],
\end{aligned}
\end{equation*}
where $U \sim \mathrm{Uniform}[0,1]$.

\subsection{Proof of Theorem \ref{Thm:WelfMisclass}}
When misclassified, let the colluders use the bid manipulation strategy in Algorithm 3. Let $i_{\textsc{c}}^\ast$ be the optimal number of zeroed $\mathsf C_{\textsc{F}}$ bids, $p^*$ the resulting C-PoP price using $\widehat{\mathsf N}$'s reported bids, and $q_{\textsc{c}}^\ast $ the optimal number of overbidding colluders in $\mathsf C_{\textsc{T}}$. For the proof, we explicitly write the bid vectors used in the welfare calculations to make a distinction between the true valuations $(\mathbf{V}^{\widehat{\mathsf N}}, \mathbf{V}^{\widehat{\mathsf C}})$ and manipulated bids $(\mathbf{B}^{\widehat{\mathsf N}}, \mathbf{B}^{\widehat{\mathsf C}})$.

\paragraph{Distribution of manipulated bids.} The bottommost $i_{\textsc{c}}^\ast$ colluders in $\mathsf C_{\textsc{F}}$ underbid, a total of $q_{\textsc{c}}^\ast $ colluders in $\mathsf C_{\textsc{T}}$ bid above the price $p^*$ and the rest of the bids are all truthful. Let $\mathsf F_{\mathsf C_{\textsc{F}}}$ and $\mathsf F_{\mathsf C_{\textsc{T}}}$ be the CDF of the manipulated colluders bids $\mathbf{B}^{\mathsf C_{\textsc{F}}}$, $\mathbf{B}^{\mathsf C_{\textsc{T}}}$ respectively given their true valuations $\mathbf{V}^{\mathsf C_{\textsc{F}}}$, $\mathbf{V}^{\mathsf C_{\textsc{T}}}$ and $i_{\textsc{c}}^\ast$ and $q_{\textsc{c}}^\ast $. Then, 
    \[ \textstyle
\mathsf{F}_{\mathsf C_{\textsc{F}}}(x) \;=\; \tfrac{i_{\textsc{c}}^\ast}{|\mathsf C_{\textsc{F}}|}\,\mathbbold 1(x \ge 0) \;+\; \tfrac{1}{|\mathsf C_{\textsc{F}}|}\sum_{j=1}^{|\mathsf C_{\textsc{F}}|- i_{\textsc{c}}^\ast} \mathbbold 1(V^{\mathsf C_{\textsc{F}}}_{j} \le x),
\]
 \[\textstyle 
\mathsf{F}_{\mathsf C_{\textsc{T}}}(x) \;=\; \tfrac{q_{\textsc{c}}^\ast}{|\mathsf C_{\textsc{T}}|}\,\mathbbold 1(x > p^*) \;+\; \tfrac{1}{|\mathsf C_{\textsc{T}}|}\sum_{j=q_{\textsc{c}}^\ast + 1}^{|\mathsf C_{\textsc{T}}|} \mathbbold 1(V^{\mathsf C_{\textsc{T}}}_{j} \le x).
\]

\paragraph{Effect of false negatives $\mathsf{C}_{\textsc{F}}$.} 
Consider the bids $(\mathbf{B}^{\widehat{\mathsf{N}}}, \mathbf{V}^{\widehat{\mathsf{C}}})$, where only the misclassified colluders in $\mathsf{C}_{\textsc{F}}$ shade their bids. Because the C-PoP mechanism computes the optimal price \(p^*\) independently of $\widehat{\mathsf{C}}$'s bids, which remains unchanged even under the fully manipulated profile $(\mathbf{B}^{\widehat{\mathsf{N}}}, \mathbf{B}^{\widehat{\mathsf{C}}})$. Assume $|\mathsf{N}_{\textsc{T}}| \geq r+1$. If correctly classified colluders ($\mathsf{C}_{\textsc{T}}$) bid truthfully and misclassified colluders ($\mathsf{C}_{\textsc{F}}$) restrict their manipulation to underbidding to $0$, the worst-case welfare is realized when all misclassified colluders drop out of the auction (i.e., bid $0$). Because $\mathsf{C}_{\textsc{T}}$ is truthful, the mechanism's assumption that unseen bids follow the true distribution remains valid. The expected welfare calculation depends exclusively on the $(k+1)$-th bid statistic, whose distribution remains unskewed as $|\mathsf{N}_{\textsc{T}}| \geq r+1$. Thus, underbidding to $0$ is mathematically equivalent to the $\mathsf{C}_{\textsc{F}}$ bidders being entirely absent from the auction. Hence, the expected welfare achieved by C-PoP satisfies
\begin{align}
   \textstyle \mathbb{E}[\mathtt{Wel}_{\textsc{C-PoP}}(\widehat{\mathsf N}, \widehat{\mathsf C}, \mathbf{B}^{\widehat{\mathsf N}}, \mathbf{V}^{\widehat{\mathsf C}})] & \ge \mathbb{E}[\mathtt{Wel}_{\textsc{C-PoP}}(\mathsf N_{\textsc{T}}, \widehat{\mathsf C}, \mathbf{V}^{\mathsf N_{\textsc{T}}}, \mathbf{V}^{\widehat{\mathsf C}})],
\end{align} 
where the right-hand side corresponds to the C-PoP mechanism with optimal price for bidders $(\mathsf N_{\textsc{T}}, \widehat{\mathsf C})$ and assuming that their true valuations are known.

\paragraph{Effect of false positives $\mathsf{N}_{\textsc{F}}$ and overbidding among $\mathsf{C}_{\textsc{T}}$. \\}

\noindent $V_{1}^{\mathsf{N}_\textsc{f}}$ is the largest valuation among false positives $\mathsf{N}_{\textsc{f}}$. It holds that $\mathbb{E}[\mathtt{Wel}_{\textsc{C-PoP}}(\widehat{\mathsf N}, \widehat{\mathsf C}, \mathbf{B}^{\widehat{\mathsf N}}, \mathbf{V}^{\widehat{\mathsf C}})] - \mathbb{E}[\mathtt{Wel}_{\textsc{C-PoP}}(\widehat{\mathsf N}, \widehat{\mathsf C}, \mathbf{B}^{\widehat{\mathsf N}}, \mathbf{B}^{\widehat{\mathsf C}})] \le |\mathsf{N}_{\textsc F}|  \mathbb{E}[\max\{0 ,V_{1}^{\mathsf{N}_\textsc{f}} -B^{\mathsf N_{\textsc T}}_{r+1}\}]$, since $|\mathsf{N}_{\textsc F}|$ upper bounds the number of non-colluding bidders that overbidding colluders can replace in item allocation and $\mathbb{E}[\max\{0, V_{1}^{\mathsf{N}_\textsc{f}}- p^*\}] \le \mathbb{E}[\max\{0, V_{1}^{\mathsf{N}_\textsc{f}} - B^{\mathsf N_{\textsc T}}_{r+1}\}]$ (using Proposition \ref{prop:pricebound}) upper bounds the loss incurred from each such replacement since replacing colluders also require (expected) non-negative utility.

\subsection{Proof of Theorem \ref{thm:cpop-noncol}}
When misclassified and using the bid manipulation strategy in Algorithm 3, the distribution of bids are as described in the proof of Theorem \ref{Thm:WelfMisclass}.

\paragraph{Effect of false negatives $\mathsf{C}_{\textsc{F}}$.} 
Consider the bids $(\mathbf{B}^{\widehat{\mathsf{N}}}, \mathbf{V}^{\widehat{\mathsf{C}}})$, where only the misclassified colluders in $\mathsf{C}_{\textsc{F}}$ shade their bids. Because the C-PoP mechanism computes the optimal price $p^*$ independently of $\widehat{\mathsf{C}}$'s bids, $p^*$ remains unchanged even under the fully manipulated profile $(\mathbf{B}^{\widehat{\mathsf{N}}}, \mathbf{B}^{\widehat{\mathsf{C}}})$. When misclassified colluders underbid or equivalently drop out of the auction (as explained in the proof of Theorem \ref{Thm:WelfMisclass}) and correctly classified colluders are truthful, the welfare of the non-colluding bidders can only improve. This is because the misclassified colluding bidders are voluntarily giving up items to create a price drop and correctly classified colluders are assumed to not manipulate their bids to gain more items. Therefore,
\begin{align}
   \textstyle \mathbb{E}[\mathtt{NWel}_{\textsc{C-PoP}}(\widehat{\mathsf N}, \widehat{\mathsf C}, \mathbf{B}^{\widehat{\mathsf N}}, \mathbf{V}^{\widehat{\mathsf C}})] & \ge \mathbb{E}[\mathtt{NWel}_{\textsc{C-PoP}}(\mathsf N_{\textsc{T}}, \widehat{\mathsf C}, \mathbf{V}^{\mathsf N_{\textsc{T}}}, \mathbf{V}^{\widehat{\mathsf C}})], 
\end{align} 
where the right-hand side corresponds to the non-colluding welfare when the C-PoP mechanism optimizes a price $p'$ with bidders $(\mathsf N_{\textsc{T}}, \widehat{\mathsf C})$ and assuming that their true valuations are known.

\paragraph{Effect of false positives $\mathsf{N}_{\textsc{F}}$ and overbidding among $\mathsf{C}_{\textsc{T}}$. \\}

\noindent $V_{1}^{\mathsf{N}_\textsc{F}}$ is the largest valuation among false positives $\mathsf{N}_{\textsc{F}}$. Then, 
\begin{align}
\mathbb{E}[\mathtt{NWel}_{\textsc{C-PoP}}(\widehat{\mathsf N}, \widehat{\mathsf C}, \mathbf{B}^{\widehat{\mathsf N}}, \mathbf{V}^{\widehat{\mathsf C}})] - \mathbb{E}[\mathtt{NWel}_{\textsc{C-PoP}}(\widehat{\mathsf N}, \widehat{\mathsf C}, \mathbf{B}^{\widehat{\mathsf N}}, \mathbf{B}^{\widehat{\mathsf C}})] \le |\mathsf{N}_{\textsc F}| \mathbb{E}[\max\{0, V_{1}^{\mathsf{N}_\textsc{F}} - B^{\mathsf N_{\textsc T}}_{r+1}\}],
\end{align}
since $|\mathsf{N}_{\textsc F}|$ upper bounds the number of non-colluding bidders that overbidding colluders can replace in item allocation and $\mathbb{E}[\max\{ 0, V_{1}^{\mathsf{N}_\textsc{F}}- p^*\}] \le \mathbb{E}[\max\{ 0, V_{1}^{\mathsf{N}_\textsc{F}}-B^{\mathsf N_{\textsc T}}_{r+1}\}]$ upper bounds the loss incurred from each such replacement since replacing colluders also require (expected) non-negative utility.


\raggedbottom
\renewcommand{\topfraction}{0.9}\renewcommand{\bottomfraction}{0.8}\renewcommand{\textfraction}{0.07}
\renewcommand{\floatpagefraction}{0.8}\setcounter{topnumber}{4}\setcounter{bottomnumber}{4}\setcounter{totalnumber}{8}

\section{Additional Experiments}\label{app:experiments}

We present additional numerical experiments to reveal the robustness of our algorithm across other parameters such as the number of items $r$ and the results for the colluder percentage sweep for $T=20$ (Figures~\ref{fig:app-T20} and~\ref{fig:app-r15}).

To substantiate our claim that false negatives are preferable to false positives, we simulate C-PoP for $T=110$, $r=10$, the Uniform[0,1] and Exponential ($\lambda = 2$) distributions, colluder percentages of about 84\% and 50\%, and varying FP/FN rate combinations. Results are displayed in Figures~\ref{fig:errgrid84} and~\ref{fig:errgrid50}, respectively.

\begin{figure}[!htbp]
    \centering
    \begin{subfigure}{0.32\textwidth}
        \centering
        \includegraphics[width=\linewidth]{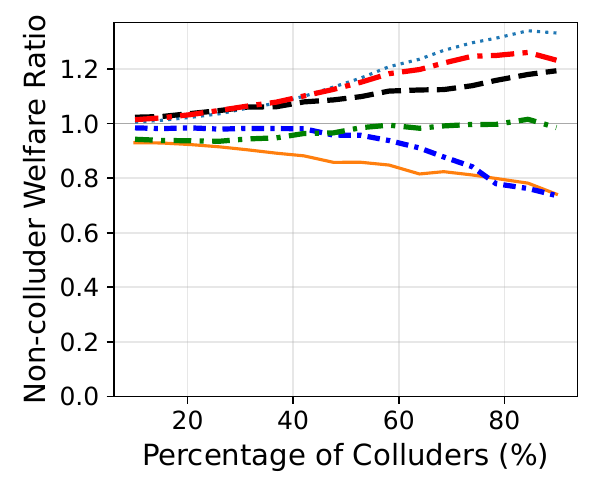}
        \caption{Non-colluder Welfare Ratio}
        \label{fig:appT20_nc_U}
    \end{subfigure}
    \hfill
    \begin{subfigure}{0.32\textwidth}
        \centering
        \includegraphics[width=\linewidth]{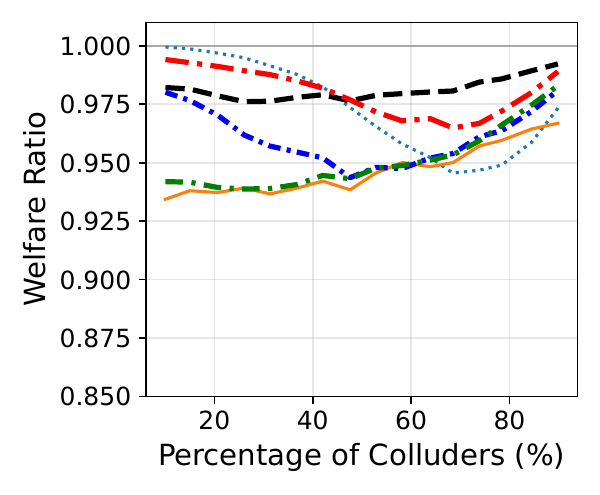}
        \caption{Welfare Ratio}
        \label{fig:appT20_wel_U}
    \end{subfigure}
    \hfill
    \begin{subfigure}{0.32\textwidth}
        \centering
        \includegraphics[width=\linewidth]{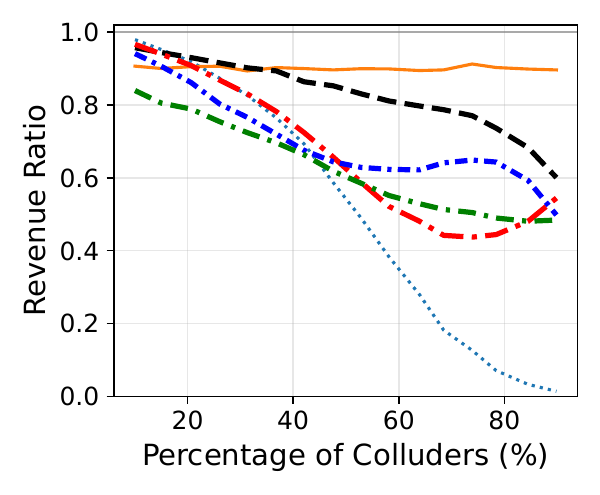}
        \caption{Revenue Ratio}
        \label{fig:appT20_rev_U}
    \end{subfigure}

    \begin{subfigure}{0.32\textwidth}
        \centering
        \includegraphics[width=\linewidth]{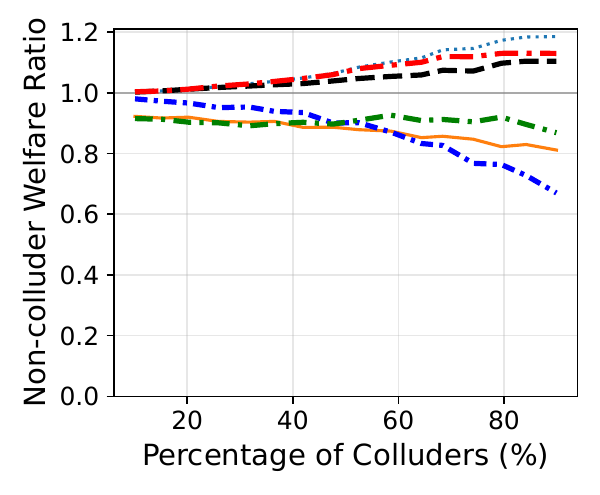}
        \caption{Non-colluder Welfare Ratio}
        \label{fig:appT20_nc_E}
    \end{subfigure}
    \hfill
    \begin{subfigure}{0.32\textwidth}
        \centering
        \includegraphics[width=\linewidth]{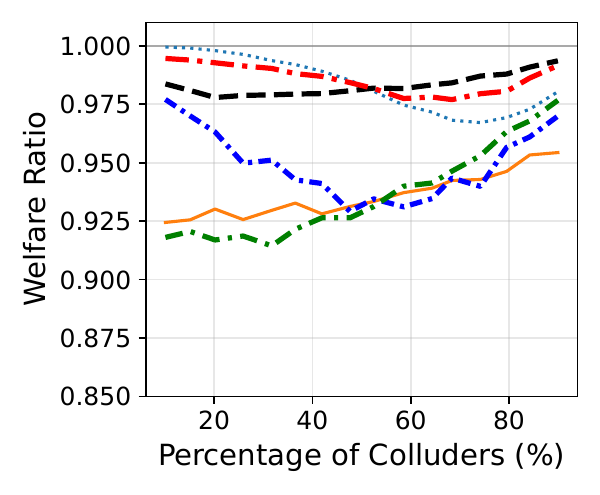}
        \caption{Welfare Ratio}
        \label{fig:appT20_wel_E}
    \end{subfigure}
    \hfill
    \begin{subfigure}{0.32\textwidth}
        \centering
        \includegraphics[width=\linewidth]{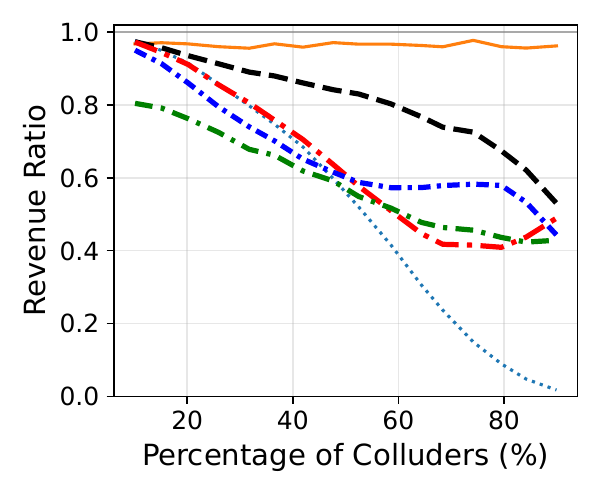}
        \caption{Revenue Ratio}
        \label{fig:appT20_rev_E}
    \end{subfigure}
    \caption{Non-colluder welfare, welfare, and revenue ratios (relative to truthful VCG without collusion) versus the percentage of colluders in a thin market with $T = 20$ bidders and $r = 10$ items. The top row uses Uniform(0,1) valuations and the bottom row uses Exponential ($\lambda = 2$) valuations. Line styles are as in Figure~\ref{fig:exp-unknown-2}; the thin gray line marks a ratio of one. Each point is the mean over 1{,}000 simulated auctions; standard errors are at most 0.02 (Tables~\ref{tab:T20-nc-U}--\ref{tab:T20-rev-E}).}
    \label{fig:app-T20}
\end{figure}

\begin{figure}[!htbp]
    \centering
    \begin{subfigure}{0.32\textwidth}
        \centering
        \includegraphics[width=\linewidth]{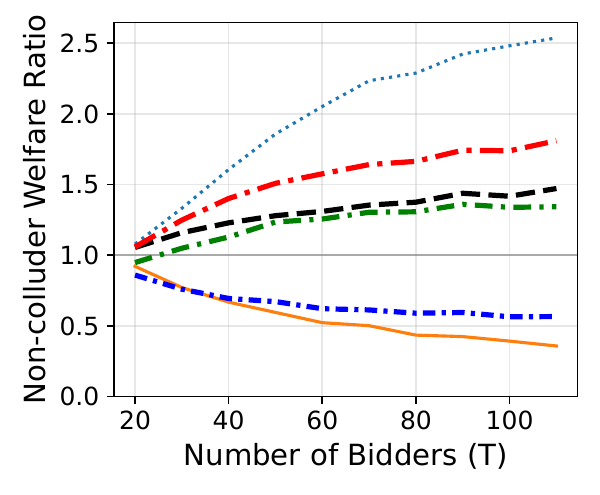}
        \caption{Non-colluder Welfare Ratio}
        \label{fig:appR15_nc_U}
    \end{subfigure}
    \hfill
    \begin{subfigure}{0.32\textwidth}
        \centering
        \includegraphics[width=\linewidth]{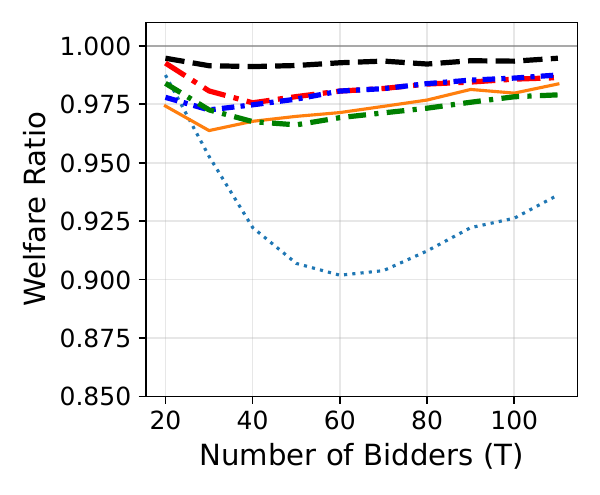}
        \caption{Welfare Ratio}
        \label{fig:appR15_wel_U}
    \end{subfigure}
    \hfill
    \begin{subfigure}{0.32\textwidth}
        \centering
        \includegraphics[width=\linewidth]{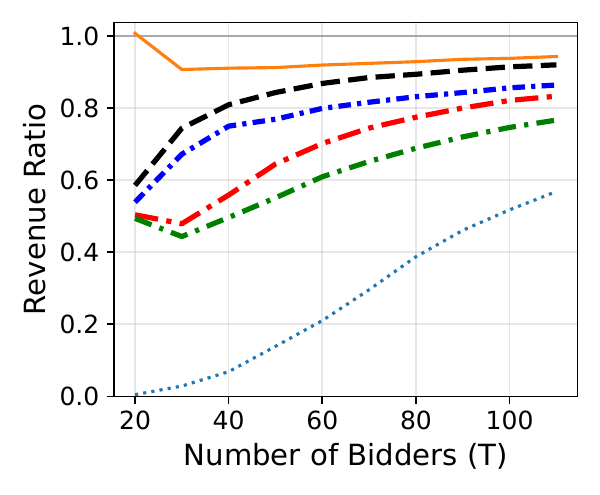}
        \caption{Revenue Ratio}
        \label{fig:appR15_rev_U}
    \end{subfigure}

    \begin{subfigure}{0.32\textwidth}
        \centering
        \includegraphics[width=\linewidth]{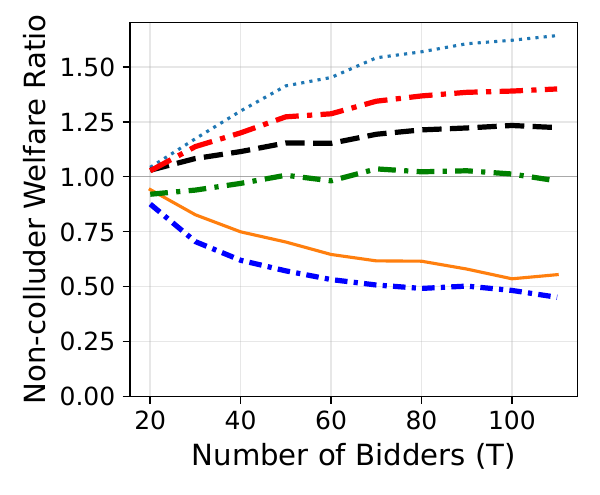}
        \caption{Non-colluder Welfare Ratio}
        \label{fig:appR15_nc_E}
    \end{subfigure}
    \hfill
    \begin{subfigure}{0.32\textwidth}
        \centering
        \includegraphics[width=\linewidth]{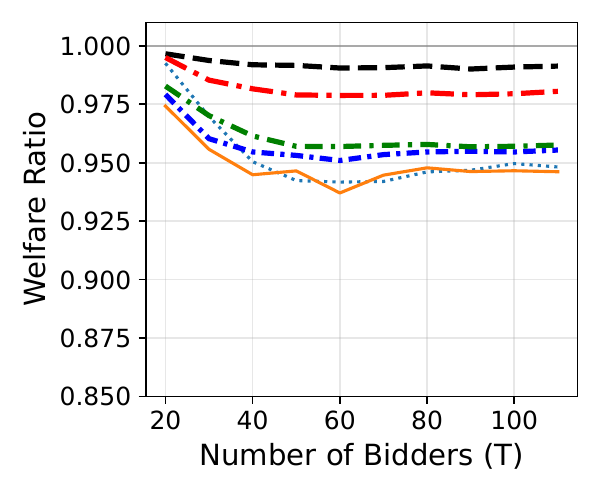}
        \caption{Welfare Ratio}
        \label{fig:appR15_wel_E}
    \end{subfigure}
    \hfill
    \begin{subfigure}{0.32\textwidth}
        \centering
        \includegraphics[width=\linewidth]{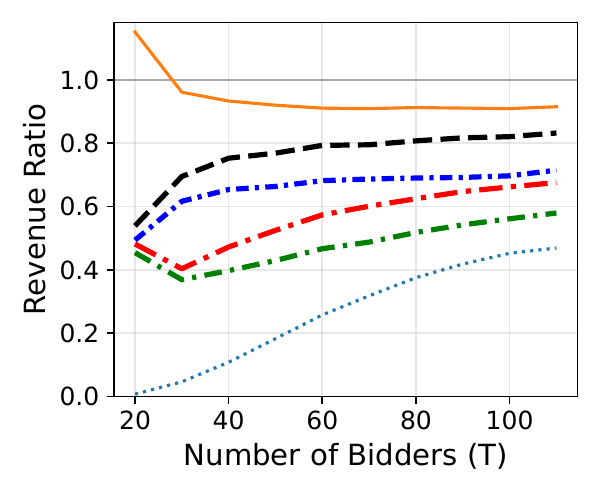}
        \caption{Revenue Ratio}
        \label{fig:appR15_rev_E}
    \end{subfigure}
    \caption{Non-colluder welfare, welfare, and revenue ratios (relative to truthful VCG without collusion) versus the number of bidders $T$, for $r = 15$ items and firm sizes $\mathcal{S} \sim \mathrm{Beta}(5,3)$ (about 84\% colluders). The top row uses Uniform(0,1) valuations and the bottom row uses Exponential ($\lambda = 2$) valuations. Line styles are as in Figure~\ref{fig:exp-unknown-2}; the thin gray line marks a ratio of one. With few bidders relative to items, Posted Price can earn more revenue than truthful VCG, whose price is the $(r{+}1)$-th highest bid. Each point is the mean over 1{,}000 simulated auctions; standard errors are at most 0.04 (Tables~\ref{tab:r15-nc-U}--\ref{tab:r15-rev-E}).}
    \label{fig:app-r15}
\end{figure}

\begin{figure}[!htbp]
    \centering
    \begin{subfigure}{0.32\textwidth}
        \centering
        \includegraphics[width=\linewidth]{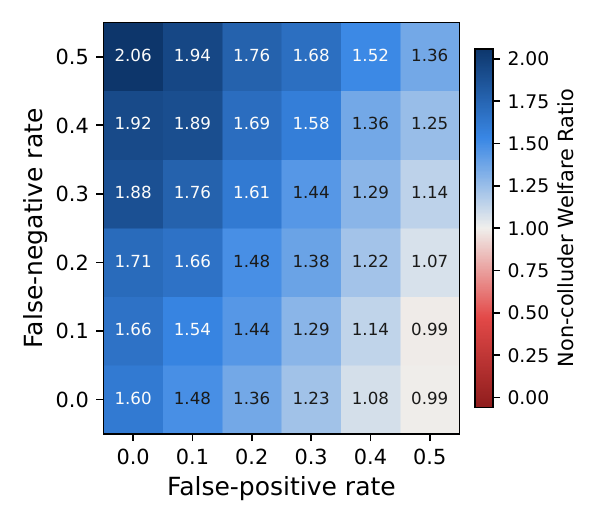}
        \caption{Non-colluder Welfare Ratio}
        \label{fig:err84_nc_U}
    \end{subfigure}
    \hfill
    \begin{subfigure}{0.32\textwidth}
        \centering
        \includegraphics[width=\linewidth]{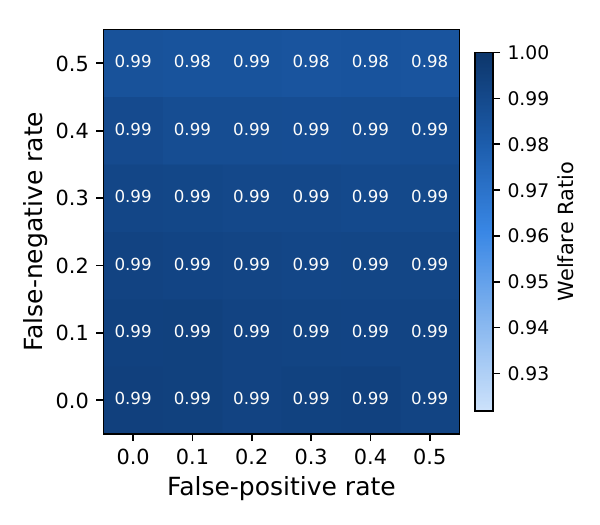}
        \caption{Welfare Ratio}
        \label{fig:err84_wel_U}
    \end{subfigure}
    \hfill
    \begin{subfigure}{0.32\textwidth}
        \centering
        \includegraphics[width=\linewidth]{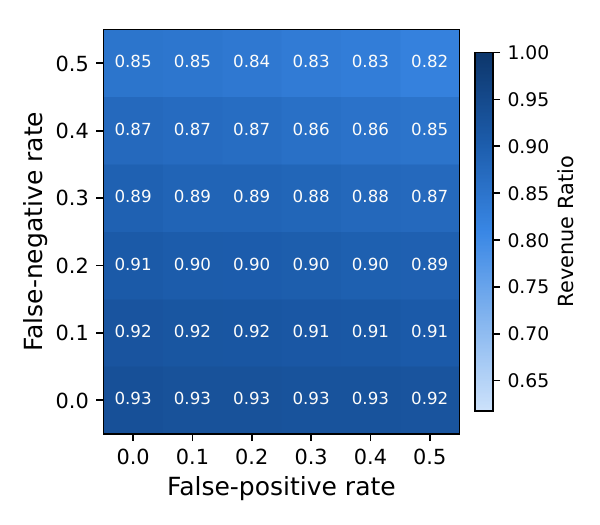}
        \caption{Revenue Ratio}
        \label{fig:err84_rev_U}
    \end{subfigure}

    \begin{subfigure}{0.32\textwidth}
        \centering
        \includegraphics[width=\linewidth]{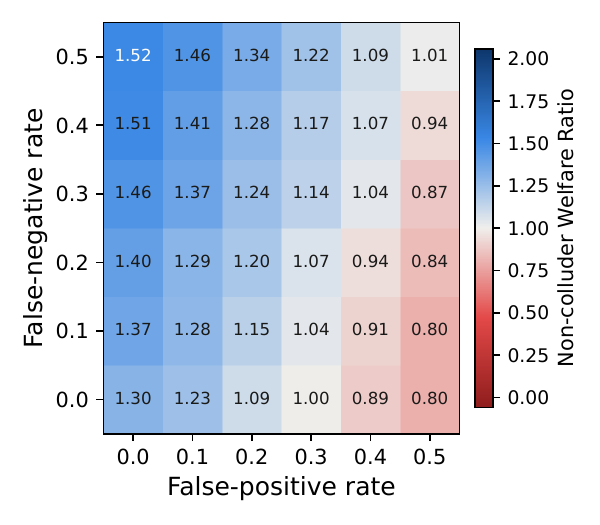}
        \caption{Non-colluder Welfare Ratio}
        \label{fig:err84_nc_E}
    \end{subfigure}
    \hfill
    \begin{subfigure}{0.32\textwidth}
        \centering
        \includegraphics[width=\linewidth]{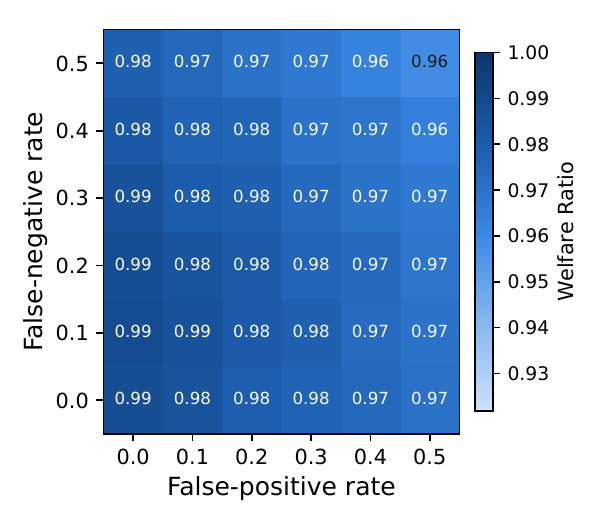}
        \caption{Welfare Ratio}
        \label{fig:err84_wel_E}
    \end{subfigure}
    \hfill
    \begin{subfigure}{0.32\textwidth}
        \centering
        \includegraphics[width=\linewidth]{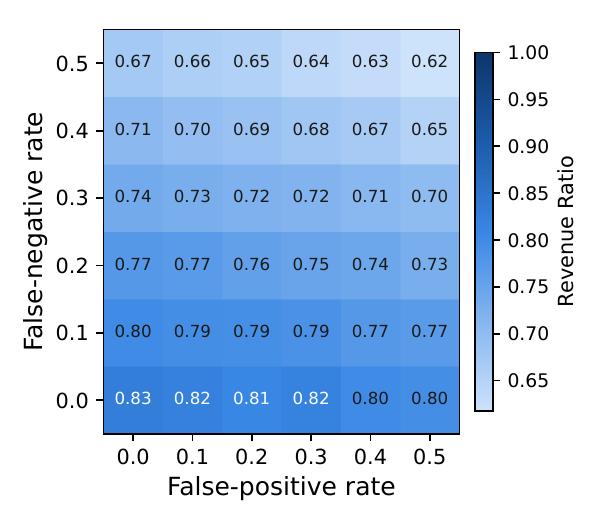}
        \caption{Revenue Ratio}
        \label{fig:err84_rev_E}
    \end{subfigure}
    \caption{C-PoP under a classifier with false-positive rate (horizontal axis) and false-negative rate (vertical axis), at about 84\% colluders, with $r = 10$ and $T = 110$. Each non-colluder is labeled a colluder with probability equal to the false-positive rate, and each colluder is labeled a non-colluder with probability equal to the false-negative rate, independently of values and firm size; colluders manipulate bids as in Algorithm~\ref{alg:cpop-manipulation}. Ratios are relative to truthful VCG without collusion, and each cell averages 1{,}000 simulations; standard errors are below 0.005 for the welfare and revenue ratios and at most 0.04 for the non-colluder welfare ratio (Tables~\ref{tab:err84-nc-U}--\ref{tab:err84-rev-E}). The top row uses Uniform(0,1) valuations and the bottom row uses Exponential ($\lambda = 2$) valuations. Non-colluder welfare ratios use a diverging scale centered at one (blue above, red below); welfare and revenue ratios use a sequential scale (darker is higher).}
    \label{fig:errgrid84}
\end{figure}

\begin{figure}[!htbp]
    \centering
    \begin{subfigure}{0.32\textwidth}
        \centering
        \includegraphics[width=\linewidth]{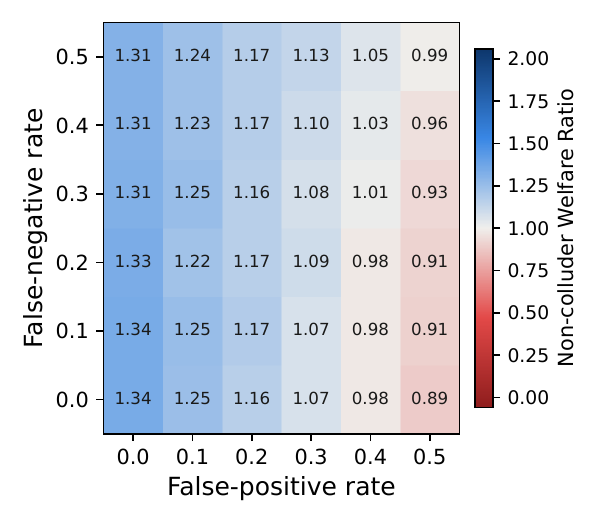}
        \caption{Non-colluder Welfare Ratio}
        \label{fig:err50_nc_U}
    \end{subfigure}
    \hfill
    \begin{subfigure}{0.32\textwidth}
        \centering
        \includegraphics[width=\linewidth]{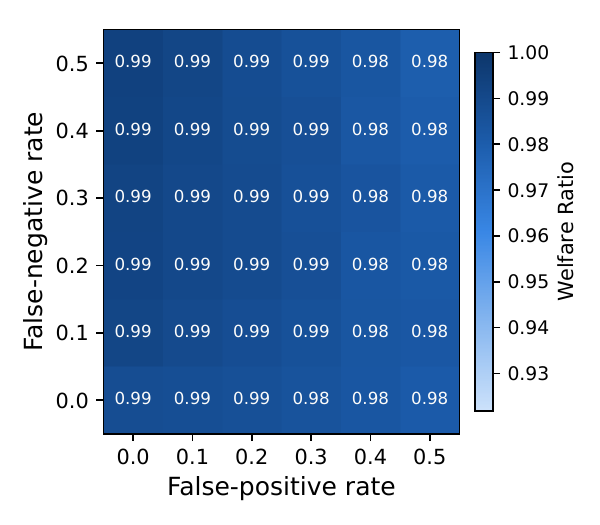}
        \caption{Welfare Ratio}
        \label{fig:err50_wel_U}
    \end{subfigure}
    \hfill
    \begin{subfigure}{0.32\textwidth}
        \centering
        \includegraphics[width=\linewidth]{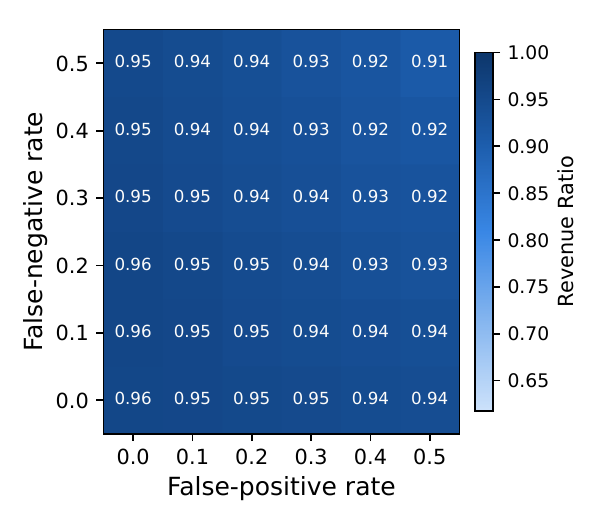}
        \caption{Revenue Ratio}
        \label{fig:err50_rev_U}
    \end{subfigure}

    \begin{subfigure}{0.32\textwidth}
        \centering
        \includegraphics[width=\linewidth]{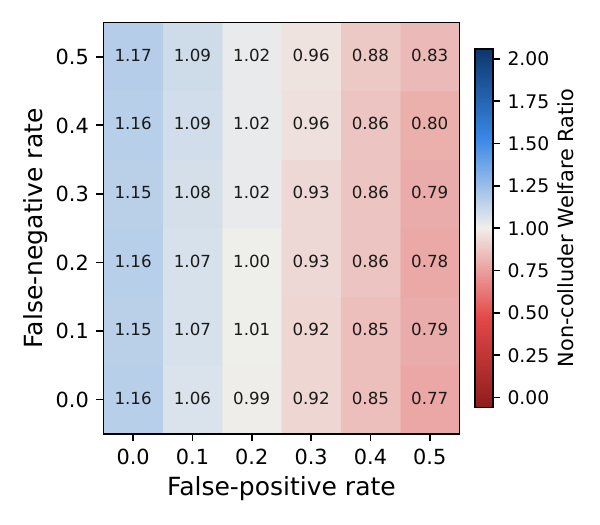}
        \caption{Non-colluder Welfare Ratio}
        \label{fig:err50_nc_E}
    \end{subfigure}
    \hfill
    \begin{subfigure}{0.32\textwidth}
        \centering
        \includegraphics[width=\linewidth]{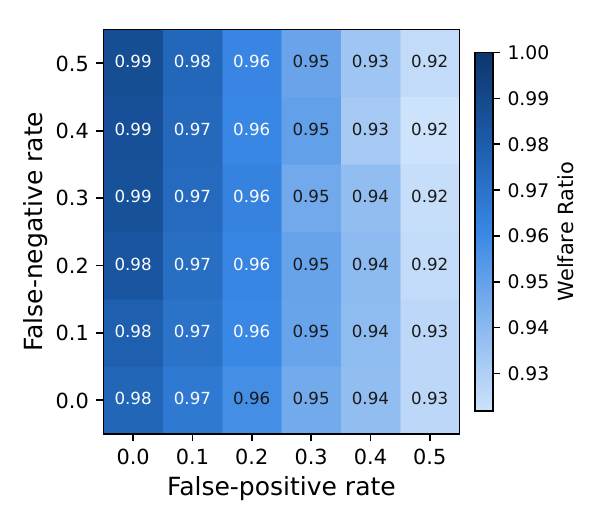}
        \caption{Welfare Ratio}
        \label{fig:err50_wel_E}
    \end{subfigure}
    \hfill
    \begin{subfigure}{0.32\textwidth}
        \centering
        \includegraphics[width=\linewidth]{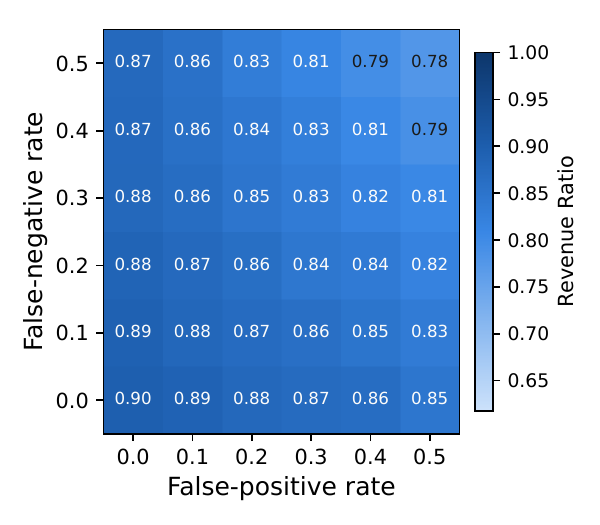}
        \caption{Revenue Ratio}
        \label{fig:err50_rev_E}
    \end{subfigure}
    \caption{As in Figure~\ref{fig:errgrid84}, at about 50\% colluders. Standard errors are below 0.005 for the welfare and revenue ratios and at most 0.02 for the non-colluder welfare ratio (Tables~\ref{tab:err50-nc-U}--\ref{tab:err50-rev-E}).}
    \label{fig:errgrid50}
\end{figure}

\subsection{Tables of Results}\label{app:tables}

Tables~\ref{tab:fig1-nc-U}--\ref{tab:fig1-rev-E} and \ref{tab:fig2-nc-U}--\ref{tab:fig2-rev-E} report the ratios plotted in Figures~\ref{fig:exp-unknown-2} and~\ref{fig:share}; Tables~\ref{tab:T20-nc-U}--\ref{tab:r15-rev-E} those in Figures~\ref{fig:app-T20} and~\ref{fig:app-r15}; and Tables~\ref{tab:err84-nc-U}--\ref{tab:err50-rev-E} those in Figures~\ref{fig:errgrid84} and~\ref{fig:errgrid50}. In every table, entries are means $\pm$ standard errors over 1{,}000 simulated auctions; $\pm 0.00$ denotes a standard error below $0.005$.

\par\medskip\noindent\begin{minipage}{\linewidth}
\centering
\captionof{table}{Non-colluder welfare ratio relative to truthful VCG without collusion (mean$\pm$SE) versus the number of bidders --- $r=10$, Uniform distribution}\label{tab:fig1-nc-U}
\small
\resizebox{\ifdim\width>\linewidth\linewidth\else\width\fi}{!}{%
\begin{tabular}{lcccccccccc}
\toprule
Mechanism & $T=20$ & $T=30$ & $T=40$ & $T=50$ & $T=60$ & $T=70$ & $T=80$ & $T=90$ & $T=100$ & $T=110$ \\
\midrule
VCG w/ Collusion & 1.33$\pm$0.01 & 1.72$\pm$0.02 & 1.99$\pm$0.03 & 2.18$\pm$0.04 & 2.28$\pm$0.04 & 2.47$\pm$0.04 & 2.47$\pm$0.05 & 2.60$\pm$0.05 & 2.51$\pm$0.05 & 2.55$\pm$0.05 \\
Posted Price & 0.77$\pm$0.01 & 0.63$\pm$0.01 & 0.53$\pm$0.02 & 0.46$\pm$0.02 & 0.40$\pm$0.01 & 0.36$\pm$0.02 & 0.36$\pm$0.02 & 0.33$\pm$0.01 & 0.31$\pm$0.01 & 0.28$\pm$0.01 \\
C-PoP (perfect) & 1.18$\pm$0.01 & 1.31$\pm$0.01 & 1.36$\pm$0.02 & 1.43$\pm$0.02 & 1.48$\pm$0.02 & 1.54$\pm$0.02 & 1.52$\pm$0.02 & 1.60$\pm$0.02 & 1.57$\pm$0.02 & 1.59$\pm$0.02 \\
C-PoP (firm size) & 1.24$\pm$0.01 & 1.47$\pm$0.02 & 1.57$\pm$0.02 & 1.67$\pm$0.03 & 1.73$\pm$0.03 & 1.82$\pm$0.03 & 1.79$\pm$0.03 & 1.90$\pm$0.03 & 1.83$\pm$0.03 & 1.87$\pm$0.03 \\
C-PoP (random) & 0.76$\pm$0.01 & 0.69$\pm$0.02 & 0.66$\pm$0.02 & 0.62$\pm$0.02 & 0.62$\pm$0.02 & 0.58$\pm$0.02 & 0.59$\pm$0.02 & 0.56$\pm$0.02 & 0.53$\pm$0.02 & 0.56$\pm$0.02 \\
C-PoP (50/50) & 0.99$\pm$0.01 & 1.18$\pm$0.02 & 1.24$\pm$0.02 & 1.28$\pm$0.03 & 1.30$\pm$0.03 & 1.41$\pm$0.03 & 1.34$\pm$0.03 & 1.44$\pm$0.03 & 1.36$\pm$0.03 & 1.36$\pm$0.03 \\
\bottomrule
\end{tabular}}
\end{minipage}\par\medskip
\par\medskip\noindent\begin{minipage}{\linewidth}
\centering
\captionof{table}{Welfare ratio relative to truthful VCG without collusion (mean$\pm$SE) versus the number of bidders --- $r=10$, Uniform distribution}\label{tab:fig1-wel-U}
\small
\resizebox{\ifdim\width>\linewidth\linewidth\else\width\fi}{!}{%
\begin{tabular}{lcccccccccc}
\toprule
Mechanism & $T=20$ & $T=30$ & $T=40$ & $T=50$ & $T=60$ & $T=70$ & $T=80$ & $T=90$ & $T=100$ & $T=110$ \\
\midrule
VCG w/ Collusion & 0.96$\pm$0.00 & 0.93$\pm$0.00 & 0.92$\pm$0.00 & 0.92$\pm$0.00 & 0.93$\pm$0.00 & 0.93$\pm$0.00 & 0.94$\pm$0.00 & 0.95$\pm$0.00 & 0.96$\pm$0.00 & 0.96$\pm$0.00 \\
Posted Price & 0.96$\pm$0.00 & 0.96$\pm$0.00 & 0.97$\pm$0.00 & 0.97$\pm$0.00 & 0.98$\pm$0.00 & 0.98$\pm$0.00 & 0.98$\pm$0.00 & 0.98$\pm$0.00 & 0.98$\pm$0.00 & 0.98$\pm$0.00 \\
C-PoP (perfect) & 0.99$\pm$0.00 & 0.99$\pm$0.00 & 0.99$\pm$0.00 & 0.99$\pm$0.00 & 0.99$\pm$0.00 & 0.99$\pm$0.00 & 0.99$\pm$0.00 & 0.99$\pm$0.00 & 0.99$\pm$0.00 & 0.99$\pm$0.00 \\
C-PoP (firm size) & 0.98$\pm$0.00 & 0.98$\pm$0.00 & 0.98$\pm$0.00 & 0.98$\pm$0.00 & 0.98$\pm$0.00 & 0.98$\pm$0.00 & 0.99$\pm$0.00 & 0.99$\pm$0.00 & 0.99$\pm$0.00 & 0.99$\pm$0.00 \\
C-PoP (random) & 0.97$\pm$0.00 & 0.97$\pm$0.00 & 0.98$\pm$0.00 & 0.98$\pm$0.00 & 0.98$\pm$0.00 & 0.99$\pm$0.00 & 0.99$\pm$0.00 & 0.99$\pm$0.00 & 0.99$\pm$0.00 & 0.99$\pm$0.00 \\
C-PoP (50/50) & 0.97$\pm$0.00 & 0.97$\pm$0.00 & 0.97$\pm$0.00 & 0.97$\pm$0.00 & 0.97$\pm$0.00 & 0.98$\pm$0.00 & 0.98$\pm$0.00 & 0.98$\pm$0.00 & 0.98$\pm$0.00 & 0.98$\pm$0.00 \\
\bottomrule
\end{tabular}}
\end{minipage}\par\medskip
\par\medskip\noindent\begin{minipage}{\linewidth}
\centering
\captionof{table}{Revenue ratio relative to truthful VCG without collusion (mean$\pm$SE) versus the number of bidders --- $r=10$, Uniform distribution}\label{tab:fig1-rev-U}
\small
\resizebox{\ifdim\width>\linewidth\linewidth\else\width\fi}{!}{%
\begin{tabular}{lcccccccccc}
\toprule
Mechanism & $T=20$ & $T=30$ & $T=40$ & $T=50$ & $T=60$ & $T=70$ & $T=80$ & $T=90$ & $T=100$ & $T=110$ \\
\midrule
VCG w/ Collusion & 0.04$\pm$0.00 & 0.12$\pm$0.01 & 0.22$\pm$0.01 & 0.33$\pm$0.01 & 0.44$\pm$0.01 & 0.52$\pm$0.01 & 0.59$\pm$0.01 & 0.64$\pm$0.01 & 0.68$\pm$0.00 & 0.71$\pm$0.00 \\
Posted Price & 0.90$\pm$0.00 & 0.89$\pm$0.00 & 0.90$\pm$0.00 & 0.91$\pm$0.00 & 0.92$\pm$0.00 & 0.93$\pm$0.00 & 0.93$\pm$0.00 & 0.94$\pm$0.00 & 0.94$\pm$0.00 & 0.95$\pm$0.00 \\
C-PoP (perfect) & 0.70$\pm$0.01 & 0.80$\pm$0.00 & 0.84$\pm$0.00 & 0.87$\pm$0.00 & 0.89$\pm$0.00 & 0.90$\pm$0.00 & 0.91$\pm$0.00 & 0.92$\pm$0.00 & 0.93$\pm$0.00 & 0.93$\pm$0.00 \\
C-PoP (firm size) & 0.51$\pm$0.01 & 0.60$\pm$0.00 & 0.69$\pm$0.00 & 0.74$\pm$0.00 & 0.78$\pm$0.00 & 0.81$\pm$0.00 & 0.83$\pm$0.00 & 0.85$\pm$0.00 & 0.86$\pm$0.00 & 0.88$\pm$0.00 \\
C-PoP (random) & 0.61$\pm$0.01 & 0.72$\pm$0.01 & 0.77$\pm$0.00 & 0.80$\pm$0.00 & 0.83$\pm$0.00 & 0.84$\pm$0.00 & 0.86$\pm$0.00 & 0.87$\pm$0.00 & 0.88$\pm$0.00 & 0.89$\pm$0.00 \\
C-PoP (50/50) & 0.49$\pm$0.00 & 0.53$\pm$0.00 & 0.60$\pm$0.00 & 0.66$\pm$0.00 & 0.70$\pm$0.00 & 0.74$\pm$0.00 & 0.77$\pm$0.00 & 0.79$\pm$0.00 & 0.81$\pm$0.00 & 0.83$\pm$0.00 \\
\bottomrule
\end{tabular}}
\end{minipage}\par\medskip
\par\medskip\noindent\begin{minipage}{\linewidth}
\centering
\captionof{table}{Non-colluder welfare ratio relative to truthful VCG without collusion (mean$\pm$SE) versus the number of bidders --- $r=10$, Exponential distribution}\label{tab:fig1-nc-E}
\small
\resizebox{\ifdim\width>\linewidth\linewidth\else\width\fi}{!}{%
\begin{tabular}{lcccccccccc}
\toprule
Mechanism & $T=20$ & $T=30$ & $T=40$ & $T=50$ & $T=60$ & $T=70$ & $T=80$ & $T=90$ & $T=100$ & $T=110$ \\
\midrule
VCG w/ Collusion & 1.18$\pm$0.01 & 1.34$\pm$0.01 & 1.44$\pm$0.01 & 1.53$\pm$0.02 & 1.56$\pm$0.02 & 1.60$\pm$0.02 & 1.65$\pm$0.02 & 1.64$\pm$0.02 & 1.64$\pm$0.02 & 1.70$\pm$0.02 \\
Posted Price & 0.81$\pm$0.01 & 0.69$\pm$0.01 & 0.64$\pm$0.02 & 0.61$\pm$0.02 & 0.55$\pm$0.02 & 0.55$\pm$0.02 & 0.52$\pm$0.02 & 0.49$\pm$0.02 & 0.49$\pm$0.02 & 0.46$\pm$0.02 \\
C-PoP (perfect) & 1.10$\pm$0.01 & 1.16$\pm$0.01 & 1.19$\pm$0.01 & 1.23$\pm$0.01 & 1.25$\pm$0.01 & 1.26$\pm$0.01 & 1.28$\pm$0.01 & 1.28$\pm$0.01 & 1.29$\pm$0.01 & 1.32$\pm$0.01 \\
C-PoP (firm size) & 1.13$\pm$0.01 & 1.24$\pm$0.01 & 1.28$\pm$0.01 & 1.35$\pm$0.01 & 1.37$\pm$0.02 & 1.38$\pm$0.02 & 1.41$\pm$0.02 & 1.42$\pm$0.02 & 1.43$\pm$0.02 & 1.46$\pm$0.02 \\
C-PoP (random) & 0.71$\pm$0.01 & 0.62$\pm$0.01 & 0.55$\pm$0.02 & 0.53$\pm$0.02 & 0.52$\pm$0.02 & 0.49$\pm$0.01 & 0.45$\pm$0.01 & 0.44$\pm$0.01 & 0.45$\pm$0.01 & 0.42$\pm$0.01 \\
C-PoP (50/50) & 0.89$\pm$0.01 & 0.93$\pm$0.02 & 0.98$\pm$0.02 & 1.03$\pm$0.02 & 1.00$\pm$0.02 & 0.99$\pm$0.02 & 1.03$\pm$0.02 & 1.01$\pm$0.02 & 0.99$\pm$0.02 & 1.04$\pm$0.02 \\
\bottomrule
\end{tabular}}
\end{minipage}\par\medskip
\par\medskip\noindent\begin{minipage}{\linewidth}
\centering
\captionof{table}{Welfare ratio relative to truthful VCG without collusion (mean$\pm$SE) versus the number of bidders --- $r=10$, Exponential distribution}\label{tab:fig1-wel-E}
\small
\resizebox{\ifdim\width>\linewidth\linewidth\else\width\fi}{!}{%
\begin{tabular}{lcccccccccc}
\toprule
Mechanism & $T=20$ & $T=30$ & $T=40$ & $T=50$ & $T=60$ & $T=70$ & $T=80$ & $T=90$ & $T=100$ & $T=110$ \\
\midrule
VCG w/ Collusion & 0.97$\pm$0.00 & 0.95$\pm$0.00 & 0.95$\pm$0.00 & 0.95$\pm$0.00 & 0.95$\pm$0.00 & 0.95$\pm$0.00 & 0.95$\pm$0.00 & 0.96$\pm$0.00 & 0.96$\pm$0.00 & 0.96$\pm$0.00 \\
Posted Price & 0.95$\pm$0.00 & 0.94$\pm$0.00 & 0.93$\pm$0.00 & 0.93$\pm$0.00 & 0.94$\pm$0.00 & 0.94$\pm$0.00 & 0.94$\pm$0.00 & 0.94$\pm$0.00 & 0.94$\pm$0.00 & 0.94$\pm$0.00 \\
C-PoP (perfect) & 0.99$\pm$0.00 & 0.99$\pm$0.00 & 0.99$\pm$0.00 & 0.99$\pm$0.00 & 0.99$\pm$0.00 & 0.99$\pm$0.00 & 0.99$\pm$0.00 & 0.99$\pm$0.00 & 0.99$\pm$0.00 & 0.99$\pm$0.00 \\
C-PoP (firm size) & 0.99$\pm$0.00 & 0.98$\pm$0.00 & 0.98$\pm$0.00 & 0.98$\pm$0.00 & 0.98$\pm$0.00 & 0.98$\pm$0.00 & 0.98$\pm$0.00 & 0.98$\pm$0.00 & 0.98$\pm$0.00 & 0.98$\pm$0.00 \\
C-PoP (random) & 0.96$\pm$0.00 & 0.95$\pm$0.00 & 0.95$\pm$0.00 & 0.95$\pm$0.00 & 0.95$\pm$0.00 & 0.95$\pm$0.00 & 0.95$\pm$0.00 & 0.95$\pm$0.00 & 0.95$\pm$0.00 & 0.95$\pm$0.00 \\
C-PoP (50/50) & 0.97$\pm$0.00 & 0.96$\pm$0.00 & 0.96$\pm$0.00 & 0.96$\pm$0.00 & 0.96$\pm$0.00 & 0.96$\pm$0.00 & 0.96$\pm$0.00 & 0.96$\pm$0.00 & 0.96$\pm$0.00 & 0.96$\pm$0.00 \\
\bottomrule
\end{tabular}}
\end{minipage}\par\medskip
\par\medskip\noindent\begin{minipage}{\linewidth}
\centering
\captionof{table}{Revenue ratio relative to truthful VCG without collusion (mean$\pm$SE) versus the number of bidders --- $r=10$, Exponential distribution}\label{tab:fig1-rev-E}
\small
\resizebox{\ifdim\width>\linewidth\linewidth\else\width\fi}{!}{%
\begin{tabular}{lcccccccccc}
\toprule
Mechanism & $T=20$ & $T=30$ & $T=40$ & $T=50$ & $T=60$ & $T=70$ & $T=80$ & $T=90$ & $T=100$ & $T=110$ \\
\midrule
VCG w/ Collusion & 0.06$\pm$0.00 & 0.15$\pm$0.01 & 0.26$\pm$0.01 & 0.34$\pm$0.01 & 0.41$\pm$0.01 & 0.44$\pm$0.01 & 0.49$\pm$0.01 & 0.51$\pm$0.01 & 0.53$\pm$0.00 & 0.56$\pm$0.00 \\
Posted Price & 0.96$\pm$0.01 & 0.90$\pm$0.01 & 0.90$\pm$0.00 & 0.90$\pm$0.00 & 0.89$\pm$0.00 & 0.90$\pm$0.00 & 0.89$\pm$0.00 & 0.89$\pm$0.00 & 0.89$\pm$0.00 & 0.90$\pm$0.00 \\
C-PoP (perfect) & 0.64$\pm$0.01 & 0.71$\pm$0.01 & 0.76$\pm$0.00 & 0.79$\pm$0.00 & 0.79$\pm$0.00 & 0.81$\pm$0.00 & 0.81$\pm$0.00 & 0.81$\pm$0.00 & 0.82$\pm$0.00 & 0.83$\pm$0.00 \\
C-PoP (firm size) & 0.46$\pm$0.01 & 0.49$\pm$0.00 & 0.56$\pm$0.00 & 0.61$\pm$0.00 & 0.64$\pm$0.00 & 0.66$\pm$0.00 & 0.68$\pm$0.00 & 0.69$\pm$0.00 & 0.70$\pm$0.00 & 0.72$\pm$0.00 \\
C-PoP (random) & 0.55$\pm$0.01 & 0.60$\pm$0.01 & 0.65$\pm$0.00 & 0.68$\pm$0.00 & 0.68$\pm$0.00 & 0.70$\pm$0.00 & 0.70$\pm$0.00 & 0.70$\pm$0.00 & 0.70$\pm$0.00 & 0.71$\pm$0.00 \\
C-PoP (50/50) & 0.43$\pm$0.01 & 0.42$\pm$0.00 & 0.47$\pm$0.00 & 0.50$\pm$0.00 & 0.53$\pm$0.00 & 0.56$\pm$0.00 & 0.58$\pm$0.00 & 0.59$\pm$0.00 & 0.60$\pm$0.00 & 0.62$\pm$0.00 \\
\bottomrule
\end{tabular}}
\end{minipage}\par\medskip

\par\medskip\noindent\begin{minipage}{\linewidth}
\centering
\captionof{table}{Non-colluder welfare ratio relative to truthful VCG without collusion (mean$\pm$SE) versus the percentage of colluders --- $r=10$, $T=110$, Uniform distribution}\label{tab:fig2-nc-U}
\small
\resizebox{\ifdim\width>\linewidth\linewidth\else\width\fi}{!}{%
\begin{tabular}{lcccccccc}
\toprule
 & \multicolumn{8}{c}{Percentage of colluders} \\
\cmidrule(lr){2-9}
Mechanism & 10\% & 15\% & 20\% & 26\% & 31\% & 37\% & 42\% & 47\% \\
\midrule
VCG w/ Collusion & 1.01$\pm$0.00 & 1.02$\pm$0.00 & 1.03$\pm$0.00 & 1.05$\pm$0.00 & 1.08$\pm$0.00 & 1.12$\pm$0.00 & 1.16$\pm$0.01 & 1.21$\pm$0.01 \\
Posted Price & 0.94$\pm$0.00 & 0.90$\pm$0.00 & 0.87$\pm$0.01 & 0.82$\pm$0.01 & 0.76$\pm$0.01 & 0.71$\pm$0.01 & 0.64$\pm$0.01 & 0.57$\pm$0.01 \\
C-PoP (perfect) & 1.10$\pm$0.00 & 1.15$\pm$0.00 & 1.18$\pm$0.01 & 1.21$\pm$0.01 & 1.24$\pm$0.01 & 1.28$\pm$0.01 & 1.30$\pm$0.01 & 1.32$\pm$0.01 \\
C-PoP (firm size) & 1.03$\pm$0.00 & 1.05$\pm$0.00 & 1.08$\pm$0.00 & 1.11$\pm$0.00 & 1.13$\pm$0.00 & 1.18$\pm$0.01 & 1.22$\pm$0.01 & 1.26$\pm$0.01 \\
C-PoP (random) & 1.00$\pm$0.00 & 1.00$\pm$0.00 & 1.00$\pm$0.00 & 1.00$\pm$0.00 & 1.00$\pm$0.00 & 1.02$\pm$0.01 & 1.00$\pm$0.01 & 0.98$\pm$0.01 \\
C-PoP (50/50) & 0.96$\pm$0.00 & 0.94$\pm$0.00 & 0.92$\pm$0.00 & 0.92$\pm$0.00 & 0.92$\pm$0.01 & 0.93$\pm$0.01 & 0.94$\pm$0.01 & 0.96$\pm$0.01 \\
\bottomrule
\end{tabular}}
\end{minipage}\par\medskip
\par\medskip\noindent\begin{minipage}{\linewidth}
\centering
\captionof{table}{Welfare ratio relative to truthful VCG without collusion (mean$\pm$SE) versus the percentage of colluders --- $r=10$, $T=110$, Uniform distribution}\label{tab:fig2-wel-U}
\small
\resizebox{\ifdim\width>\linewidth\linewidth\else\width\fi}{!}{%
\begin{tabular}{lcccccccc}
\toprule
 & \multicolumn{8}{c}{Percentage of colluders} \\
\cmidrule(lr){2-9}
Mechanism & 10\% & 15\% & 20\% & 26\% & 31\% & 37\% & 42\% & 47\% \\
\midrule
VCG w/ Collusion & 1.00$\pm$0.00 & 1.00$\pm$0.00 & 1.00$\pm$0.00 & 1.00$\pm$0.00 & 1.00$\pm$0.00 & 1.00$\pm$0.00 & 1.00$\pm$0.00 & 1.00$\pm$0.00 \\
Posted Price & 0.97$\pm$0.00 & 0.97$\pm$0.00 & 0.97$\pm$0.00 & 0.97$\pm$0.00 & 0.97$\pm$0.00 & 0.97$\pm$0.00 & 0.97$\pm$0.00 & 0.97$\pm$0.00 \\
C-PoP (perfect) & 0.99$\pm$0.00 & 0.99$\pm$0.00 & 0.99$\pm$0.00 & 0.99$\pm$0.00 & 0.99$\pm$0.00 & 0.99$\pm$0.00 & 0.99$\pm$0.00 & 0.99$\pm$0.00 \\
C-PoP (firm size) & 1.00$\pm$0.00 & 1.00$\pm$0.00 & 1.00$\pm$0.00 & 1.00$\pm$0.00 & 1.00$\pm$0.00 & 1.00$\pm$0.00 & 1.00$\pm$0.00 & 0.99$\pm$0.00 \\
C-PoP (random) & 0.99$\pm$0.00 & 0.99$\pm$0.00 & 0.99$\pm$0.00 & 0.99$\pm$0.00 & 0.99$\pm$0.00 & 0.98$\pm$0.00 & 0.98$\pm$0.00 & 0.98$\pm$0.00 \\
C-PoP (50/50) & 0.98$\pm$0.00 & 0.98$\pm$0.00 & 0.98$\pm$0.00 & 0.98$\pm$0.00 & 0.98$\pm$0.00 & 0.98$\pm$0.00 & 0.98$\pm$0.00 & 0.98$\pm$0.00 \\
\bottomrule
\end{tabular}}
\end{minipage}\par\medskip
\par\medskip\noindent\begin{minipage}{\linewidth}
\centering
\captionof{table}{Revenue ratio relative to truthful VCG without collusion (mean$\pm$SE) versus the percentage of colluders --- $r=10$, $T=110$, Uniform distribution}\label{tab:fig2-rev-U}
\small
\resizebox{\ifdim\width>\linewidth\linewidth\else\width\fi}{!}{%
\begin{tabular}{lcccccccc}
\toprule
 & \multicolumn{8}{c}{Percentage of colluders} \\
\cmidrule(lr){2-9}
Mechanism & 10\% & 15\% & 20\% & 26\% & 31\% & 37\% & 42\% & 47\% \\
\midrule
VCG w/ Collusion & 1.00$\pm$0.00 & 1.00$\pm$0.00 & 0.99$\pm$0.00 & 0.99$\pm$0.00 & 0.98$\pm$0.00 & 0.98$\pm$0.00 & 0.97$\pm$0.00 & 0.96$\pm$0.00 \\
Posted Price & 0.95$\pm$0.00 & 0.95$\pm$0.00 & 0.95$\pm$0.00 & 0.95$\pm$0.00 & 0.95$\pm$0.00 & 0.95$\pm$0.00 & 0.94$\pm$0.00 & 0.94$\pm$0.00 \\
C-PoP (perfect) & 0.99$\pm$0.00 & 0.98$\pm$0.00 & 0.98$\pm$0.00 & 0.97$\pm$0.00 & 0.97$\pm$0.00 & 0.97$\pm$0.00 & 0.96$\pm$0.00 & 0.96$\pm$0.00 \\
C-PoP (firm size) & 1.00$\pm$0.00 & 0.99$\pm$0.00 & 0.99$\pm$0.00 & 0.98$\pm$0.00 & 0.98$\pm$0.00 & 0.97$\pm$0.00 & 0.96$\pm$0.00 & 0.95$\pm$0.00 \\
C-PoP (random) & 0.99$\pm$0.00 & 0.98$\pm$0.00 & 0.97$\pm$0.00 & 0.96$\pm$0.00 & 0.95$\pm$0.00 & 0.94$\pm$0.00 & 0.93$\pm$0.00 & 0.91$\pm$0.00 \\
C-PoP (50/50) & 0.96$\pm$0.00 & 0.95$\pm$0.00 & 0.94$\pm$0.00 & 0.94$\pm$0.00 & 0.93$\pm$0.00 & 0.93$\pm$0.00 & 0.92$\pm$0.00 & 0.91$\pm$0.00 \\
\bottomrule
\end{tabular}}
\end{minipage}\par\medskip
\par\medskip\noindent\begin{minipage}{\linewidth}
\centering
\captionof{table}{Non-colluder welfare ratio relative to truthful VCG without collusion (mean$\pm$SE) versus the percentage of colluders --- $r=10$, $T=110$, Exponential distribution}\label{tab:fig2-nc-E}
\small
\resizebox{\ifdim\width>\linewidth\linewidth\else\width\fi}{!}{%
\begin{tabular}{lcccccccc}
\toprule
 & \multicolumn{8}{c}{Percentage of colluders} \\
\cmidrule(lr){2-9}
Mechanism & 10\% & 15\% & 21\% & 26\% & 31\% & 37\% & 42\% & 47\% \\
\midrule
VCG w/ Collusion & 1.01$\pm$0.00 & 1.01$\pm$0.00 & 1.02$\pm$0.00 & 1.04$\pm$0.00 & 1.05$\pm$0.00 & 1.07$\pm$0.00 & 1.09$\pm$0.00 & 1.12$\pm$0.00 \\
Posted Price & 0.89$\pm$0.00 & 0.88$\pm$0.00 & 0.86$\pm$0.00 & 0.84$\pm$0.01 & 0.81$\pm$0.01 & 0.81$\pm$0.01 & 0.75$\pm$0.01 & 0.73$\pm$0.01 \\
C-PoP (perfect) & 1.04$\pm$0.00 & 1.05$\pm$0.00 & 1.07$\pm$0.00 & 1.08$\pm$0.00 & 1.09$\pm$0.00 & 1.11$\pm$0.00 & 1.13$\pm$0.01 & 1.14$\pm$0.01 \\
C-PoP (firm size) & 1.02$\pm$0.00 & 1.04$\pm$0.00 & 1.05$\pm$0.00 & 1.06$\pm$0.00 & 1.07$\pm$0.00 & 1.09$\pm$0.00 & 1.11$\pm$0.00 & 1.13$\pm$0.00 \\
C-PoP (random) & 0.98$\pm$0.00 & 0.97$\pm$0.00 & 0.95$\pm$0.00 & 0.94$\pm$0.00 & 0.92$\pm$0.00 & 0.90$\pm$0.01 & 0.86$\pm$0.01 & 0.83$\pm$0.01 \\
C-PoP (50/50) & 0.89$\pm$0.00 & 0.87$\pm$0.00 & 0.85$\pm$0.00 & 0.84$\pm$0.01 & 0.82$\pm$0.01 & 0.81$\pm$0.01 & 0.81$\pm$0.01 & 0.81$\pm$0.01 \\
\bottomrule
\end{tabular}}
\end{minipage}\par\medskip
\par\medskip\noindent\begin{minipage}{\linewidth}
\centering
\captionof{table}{Welfare ratio relative to truthful VCG without collusion (mean$\pm$SE) versus the percentage of colluders --- $r=10$, $T=110$, Exponential distribution}\label{tab:fig2-wel-E}
\small
\resizebox{\ifdim\width>\linewidth\linewidth\else\width\fi}{!}{%
\begin{tabular}{lcccccccc}
\toprule
 & \multicolumn{8}{c}{Percentage of colluders} \\
\cmidrule(lr){2-9}
Mechanism & 10\% & 15\% & 21\% & 26\% & 31\% & 37\% & 42\% & 47\% \\
\midrule
VCG w/ Collusion & 1.00$\pm$0.00 & 1.00$\pm$0.00 & 1.00$\pm$0.00 & 1.00$\pm$0.00 & 1.00$\pm$0.00 & 1.00$\pm$0.00 & 0.99$\pm$0.00 & 0.99$\pm$0.00 \\
Posted Price & 0.91$\pm$0.00 & 0.91$\pm$0.00 & 0.91$\pm$0.00 & 0.91$\pm$0.00 & 0.91$\pm$0.00 & 0.92$\pm$0.00 & 0.91$\pm$0.00 & 0.91$\pm$0.00 \\
C-PoP (perfect) & 0.98$\pm$0.00 & 0.97$\pm$0.00 & 0.97$\pm$0.00 & 0.97$\pm$0.00 & 0.97$\pm$0.00 & 0.97$\pm$0.00 & 0.98$\pm$0.00 & 0.98$\pm$0.00 \\
C-PoP (firm size) & 0.99$\pm$0.00 & 0.99$\pm$0.00 & 0.99$\pm$0.00 & 0.99$\pm$0.00 & 0.99$\pm$0.00 & 0.99$\pm$0.00 & 0.99$\pm$0.00 & 0.99$\pm$0.00 \\
C-PoP (random) & 0.98$\pm$0.00 & 0.97$\pm$0.00 & 0.96$\pm$0.00 & 0.95$\pm$0.00 & 0.94$\pm$0.00 & 0.93$\pm$0.00 & 0.93$\pm$0.00 & 0.92$\pm$0.00 \\
C-PoP (50/50) & 0.92$\pm$0.00 & 0.91$\pm$0.00 & 0.91$\pm$0.00 & 0.91$\pm$0.00 & 0.91$\pm$0.00 & 0.91$\pm$0.00 & 0.91$\pm$0.00 & 0.92$\pm$0.00 \\
\bottomrule
\end{tabular}}
\end{minipage}\par\medskip
\par\medskip\noindent\begin{minipage}{\linewidth}
\centering
\captionof{table}{Revenue ratio relative to truthful VCG without collusion (mean$\pm$SE) versus the percentage of colluders --- $r=10$, $T=110$, Exponential distribution}\label{tab:fig2-rev-E}
\small
\resizebox{\ifdim\width>\linewidth\linewidth\else\width\fi}{!}{%
\begin{tabular}{lcccccccc}
\toprule
 & \multicolumn{8}{c}{Percentage of colluders} \\
\cmidrule(lr){2-9}
Mechanism & 10\% & 15\% & 21\% & 26\% & 31\% & 37\% & 42\% & 47\% \\
\midrule
VCG w/ Collusion & 0.99$\pm$0.00 & 0.98$\pm$0.00 & 0.97$\pm$0.00 & 0.96$\pm$0.00 & 0.94$\pm$0.00 & 0.92$\pm$0.00 & 0.91$\pm$0.00 & 0.89$\pm$0.00 \\
Posted Price & 0.89$\pm$0.00 & 0.90$\pm$0.00 & 0.90$\pm$0.00 & 0.90$\pm$0.00 & 0.89$\pm$0.00 & 0.90$\pm$0.00 & 0.90$\pm$0.00 & 0.90$\pm$0.00 \\
C-PoP (perfect) & 0.97$\pm$0.00 & 0.95$\pm$0.00 & 0.94$\pm$0.00 & 0.94$\pm$0.00 & 0.93$\pm$0.00 & 0.92$\pm$0.00 & 0.91$\pm$0.00 & 0.90$\pm$0.00 \\
C-PoP (firm size) & 0.98$\pm$0.00 & 0.97$\pm$0.00 & 0.96$\pm$0.00 & 0.95$\pm$0.00 & 0.94$\pm$0.00 & 0.91$\pm$0.00 & 0.90$\pm$0.00 & 0.88$\pm$0.00 \\
C-PoP (random) & 0.96$\pm$0.00 & 0.94$\pm$0.00 & 0.91$\pm$0.00 & 0.89$\pm$0.00 & 0.86$\pm$0.00 & 0.83$\pm$0.00 & 0.81$\pm$0.00 & 0.79$\pm$0.00 \\
C-PoP (50/50) & 0.88$\pm$0.00 & 0.88$\pm$0.00 & 0.86$\pm$0.00 & 0.85$\pm$0.00 & 0.83$\pm$0.00 & 0.82$\pm$0.00 & 0.80$\pm$0.00 & 0.79$\pm$0.00 \\
\bottomrule
\end{tabular}}
\end{minipage}\par\medskip

\par\medskip\noindent\begin{minipage}{\linewidth}
\centering
\captionof{table}{Non-colluder welfare ratio relative to truthful VCG without collusion (mean$\pm$SE) versus the percentage of colluders --- $r=10$, $T=20$, Uniform distribution}\label{tab:T20-nc-U}
\small
\resizebox{\ifdim\width>\linewidth\linewidth\else\width\fi}{!}{%
\begin{tabular}{lcccccccc}
\toprule
 & \multicolumn{8}{c}{Percentage of colluders} \\
\cmidrule(lr){2-9}
Mechanism & 10\% & 15\% & 21\% & 26\% & 31\% & 37\% & 42\% & 48\% \\
\midrule
VCG w/ Collusion & 1.01$\pm$0.00 & 1.01$\pm$0.00 & 1.02$\pm$0.00 & 1.04$\pm$0.00 & 1.06$\pm$0.00 & 1.08$\pm$0.00 & 1.10$\pm$0.00 & 1.13$\pm$0.00 \\
Posted Price & 0.93$\pm$0.00 & 0.93$\pm$0.00 & 0.92$\pm$0.00 & 0.91$\pm$0.00 & 0.90$\pm$0.00 & 0.89$\pm$0.00 & 0.88$\pm$0.01 & 0.86$\pm$0.01 \\
C-PoP (perfect) & 1.02$\pm$0.00 & 1.03$\pm$0.00 & 1.04$\pm$0.00 & 1.05$\pm$0.00 & 1.06$\pm$0.00 & 1.06$\pm$0.00 & 1.08$\pm$0.00 & 1.09$\pm$0.00 \\
C-PoP (firm size) & 1.01$\pm$0.00 & 1.02$\pm$0.00 & 1.03$\pm$0.00 & 1.05$\pm$0.00 & 1.06$\pm$0.00 & 1.08$\pm$0.00 & 1.10$\pm$0.00 & 1.13$\pm$0.00 \\
C-PoP (random) & 0.98$\pm$0.00 & 0.98$\pm$0.00 & 0.98$\pm$0.00 & 0.98$\pm$0.00 & 0.98$\pm$0.00 & 0.98$\pm$0.00 & 0.98$\pm$0.00 & 0.96$\pm$0.01 \\
C-PoP (50/50) & 0.94$\pm$0.00 & 0.94$\pm$0.00 & 0.94$\pm$0.00 & 0.94$\pm$0.00 & 0.94$\pm$0.00 & 0.95$\pm$0.00 & 0.96$\pm$0.00 & 0.97$\pm$0.01 \\
\bottomrule
\end{tabular}}
\end{minipage}\par\medskip
\par\medskip\noindent\begin{minipage}{\linewidth}
\centering
\captionof{table}{Welfare ratio relative to truthful VCG without collusion (mean$\pm$SE) versus the percentage of colluders --- $r=10$, $T=20$, Uniform distribution}\label{tab:T20-wel-U}
\small
\resizebox{\ifdim\width>\linewidth\linewidth\else\width\fi}{!}{%
\begin{tabular}{lcccccccc}
\toprule
 & \multicolumn{8}{c}{Percentage of colluders} \\
\cmidrule(lr){2-9}
Mechanism & 10\% & 15\% & 21\% & 26\% & 31\% & 37\% & 42\% & 48\% \\
\midrule
VCG w/ Collusion & 1.00$\pm$0.00 & 1.00$\pm$0.00 & 1.00$\pm$0.00 & 0.99$\pm$0.00 & 0.99$\pm$0.00 & 0.99$\pm$0.00 & 0.98$\pm$0.00 & 0.97$\pm$0.00 \\
Posted Price & 0.93$\pm$0.00 & 0.94$\pm$0.00 & 0.94$\pm$0.00 & 0.94$\pm$0.00 & 0.94$\pm$0.00 & 0.94$\pm$0.00 & 0.94$\pm$0.00 & 0.94$\pm$0.00 \\
C-PoP (perfect) & 0.98$\pm$0.00 & 0.98$\pm$0.00 & 0.98$\pm$0.00 & 0.98$\pm$0.00 & 0.98$\pm$0.00 & 0.98$\pm$0.00 & 0.98$\pm$0.00 & 0.98$\pm$0.00 \\
C-PoP (firm size) & 0.99$\pm$0.00 & 0.99$\pm$0.00 & 0.99$\pm$0.00 & 0.99$\pm$0.00 & 0.99$\pm$0.00 & 0.99$\pm$0.00 & 0.98$\pm$0.00 & 0.98$\pm$0.00 \\
C-PoP (random) & 0.98$\pm$0.00 & 0.98$\pm$0.00 & 0.97$\pm$0.00 & 0.96$\pm$0.00 & 0.96$\pm$0.00 & 0.95$\pm$0.00 & 0.95$\pm$0.00 & 0.94$\pm$0.00 \\
C-PoP (50/50) & 0.94$\pm$0.00 & 0.94$\pm$0.00 & 0.94$\pm$0.00 & 0.94$\pm$0.00 & 0.94$\pm$0.00 & 0.94$\pm$0.00 & 0.94$\pm$0.00 & 0.94$\pm$0.00 \\
\bottomrule
\end{tabular}}
\end{minipage}\par\medskip
\par\medskip\noindent\begin{minipage}{\linewidth}
\centering
\captionof{table}{Revenue ratio relative to truthful VCG without collusion (mean$\pm$SE) versus the percentage of colluders --- $r=10$, $T=20$, Uniform distribution}\label{tab:T20-rev-U}
\small
\resizebox{\ifdim\width>\linewidth\linewidth\else\width\fi}{!}{%
\begin{tabular}{lcccccccc}
\toprule
 & \multicolumn{8}{c}{Percentage of colluders} \\
\cmidrule(lr){2-9}
Mechanism & 10\% & 15\% & 21\% & 26\% & 31\% & 37\% & 42\% & 48\% \\
\midrule
VCG w/ Collusion & 0.98$\pm$0.00 & 0.95$\pm$0.00 & 0.92$\pm$0.00 & 0.87$\pm$0.01 & 0.83$\pm$0.01 & 0.76$\pm$0.01 & 0.69$\pm$0.01 & 0.59$\pm$0.01 \\
Posted Price & 0.91$\pm$0.01 & 0.90$\pm$0.01 & 0.91$\pm$0.01 & 0.91$\pm$0.01 & 0.89$\pm$0.01 & 0.90$\pm$0.01 & 0.90$\pm$0.01 & 0.90$\pm$0.00 \\
C-PoP (perfect) & 0.96$\pm$0.00 & 0.94$\pm$0.00 & 0.93$\pm$0.00 & 0.92$\pm$0.00 & 0.90$\pm$0.00 & 0.89$\pm$0.00 & 0.86$\pm$0.00 & 0.85$\pm$0.00 \\
C-PoP (firm size) & 0.97$\pm$0.00 & 0.94$\pm$0.00 & 0.91$\pm$0.00 & 0.87$\pm$0.01 & 0.83$\pm$0.01 & 0.78$\pm$0.01 & 0.72$\pm$0.01 & 0.66$\pm$0.01 \\
C-PoP (random) & 0.94$\pm$0.00 & 0.91$\pm$0.00 & 0.86$\pm$0.00 & 0.80$\pm$0.01 & 0.77$\pm$0.01 & 0.72$\pm$0.01 & 0.68$\pm$0.01 & 0.64$\pm$0.01 \\
C-PoP (50/50) & 0.84$\pm$0.00 & 0.81$\pm$0.00 & 0.79$\pm$0.01 & 0.75$\pm$0.01 & 0.72$\pm$0.01 & 0.70$\pm$0.01 & 0.66$\pm$0.01 & 0.62$\pm$0.01 \\
\bottomrule
\end{tabular}}
\end{minipage}\par\medskip
\par\medskip\noindent\begin{minipage}{\linewidth}
\centering
\captionof{table}{Non-colluder welfare ratio relative to truthful VCG without collusion (mean$\pm$SE) versus the percentage of colluders --- $r=10$, $T=20$, Exponential distribution}\label{tab:T20-nc-E}
\small
\resizebox{\ifdim\width>\linewidth\linewidth\else\width\fi}{!}{%
\begin{tabular}{lcccccccc}
\toprule
 & \multicolumn{8}{c}{Percentage of colluders} \\
\cmidrule(lr){2-9}
Mechanism & 10\% & 15\% & 20\% & 26\% & 32\% & 37\% & 42\% & 48\% \\
\midrule
VCG w/ Collusion & 1.00$\pm$0.00 & 1.01$\pm$0.00 & 1.01$\pm$0.00 & 1.02$\pm$0.00 & 1.03$\pm$0.00 & 1.04$\pm$0.00 & 1.05$\pm$0.00 & 1.06$\pm$0.00 \\
Posted Price & 0.92$\pm$0.00 & 0.92$\pm$0.00 & 0.92$\pm$0.00 & 0.91$\pm$0.00 & 0.90$\pm$0.00 & 0.91$\pm$0.00 & 0.89$\pm$0.01 & 0.89$\pm$0.01 \\
C-PoP (perfect) & 1.00$\pm$0.00 & 1.01$\pm$0.00 & 1.01$\pm$0.00 & 1.02$\pm$0.00 & 1.02$\pm$0.00 & 1.03$\pm$0.00 & 1.03$\pm$0.00 & 1.04$\pm$0.00 \\
C-PoP (firm size) & 1.00$\pm$0.00 & 1.01$\pm$0.00 & 1.01$\pm$0.00 & 1.02$\pm$0.00 & 1.03$\pm$0.00 & 1.04$\pm$0.00 & 1.05$\pm$0.00 & 1.06$\pm$0.00 \\
C-PoP (random) & 0.98$\pm$0.00 & 0.97$\pm$0.00 & 0.97$\pm$0.00 & 0.95$\pm$0.00 & 0.95$\pm$0.00 & 0.94$\pm$0.00 & 0.94$\pm$0.00 & 0.90$\pm$0.01 \\
C-PoP (50/50) & 0.92$\pm$0.00 & 0.91$\pm$0.00 & 0.90$\pm$0.00 & 0.90$\pm$0.00 & 0.89$\pm$0.00 & 0.90$\pm$0.01 & 0.90$\pm$0.01 & 0.90$\pm$0.01 \\
\bottomrule
\end{tabular}}
\end{minipage}\par\medskip
\par\medskip\noindent\begin{minipage}{\linewidth}
\centering
\captionof{table}{Welfare ratio relative to truthful VCG without collusion (mean$\pm$SE) versus the percentage of colluders --- $r=10$, $T=20$, Exponential distribution}\label{tab:T20-wel-E}
\small
\resizebox{\ifdim\width>\linewidth\linewidth\else\width\fi}{!}{%
\begin{tabular}{lcccccccc}
\toprule
 & \multicolumn{8}{c}{Percentage of colluders} \\
\cmidrule(lr){2-9}
Mechanism & 10\% & 15\% & 20\% & 26\% & 32\% & 37\% & 42\% & 48\% \\
\midrule
VCG w/ Collusion & 1.00$\pm$0.00 & 1.00$\pm$0.00 & 1.00$\pm$0.00 & 1.00$\pm$0.00 & 0.99$\pm$0.00 & 0.99$\pm$0.00 & 0.99$\pm$0.00 & 0.99$\pm$0.00 \\
Posted Price & 0.92$\pm$0.00 & 0.93$\pm$0.00 & 0.93$\pm$0.00 & 0.93$\pm$0.00 & 0.93$\pm$0.00 & 0.93$\pm$0.00 & 0.93$\pm$0.00 & 0.93$\pm$0.00 \\
C-PoP (perfect) & 0.98$\pm$0.00 & 0.98$\pm$0.00 & 0.98$\pm$0.00 & 0.98$\pm$0.00 & 0.98$\pm$0.00 & 0.98$\pm$0.00 & 0.98$\pm$0.00 & 0.98$\pm$0.00 \\
C-PoP (firm size) & 0.99$\pm$0.00 & 0.99$\pm$0.00 & 0.99$\pm$0.00 & 0.99$\pm$0.00 & 0.99$\pm$0.00 & 0.99$\pm$0.00 & 0.99$\pm$0.00 & 0.98$\pm$0.00 \\
C-PoP (random) & 0.98$\pm$0.00 & 0.97$\pm$0.00 & 0.96$\pm$0.00 & 0.95$\pm$0.00 & 0.95$\pm$0.00 & 0.94$\pm$0.00 & 0.94$\pm$0.00 & 0.93$\pm$0.00 \\
C-PoP (50/50) & 0.92$\pm$0.00 & 0.92$\pm$0.00 & 0.92$\pm$0.00 & 0.92$\pm$0.00 & 0.91$\pm$0.00 & 0.92$\pm$0.00 & 0.93$\pm$0.00 & 0.93$\pm$0.00 \\
\bottomrule
\end{tabular}}
\end{minipage}\par\medskip
\par\medskip\noindent\begin{minipage}{\linewidth}
\centering
\captionof{table}{Revenue ratio relative to truthful VCG without collusion (mean$\pm$SE) versus the percentage of colluders --- $r=10$, $T=20$, Exponential distribution}\label{tab:T20-rev-E}
\small
\resizebox{\ifdim\width>\linewidth\linewidth\else\width\fi}{!}{%
\begin{tabular}{lcccccccc}
\toprule
 & \multicolumn{8}{c}{Percentage of colluders} \\
\cmidrule(lr){2-9}
Mechanism & 10\% & 15\% & 20\% & 26\% & 32\% & 37\% & 42\% & 48\% \\
\midrule
VCG w/ Collusion & 0.98$\pm$0.00 & 0.95$\pm$0.00 & 0.91$\pm$0.01 & 0.86$\pm$0.01 & 0.80$\pm$0.01 & 0.75$\pm$0.01 & 0.68$\pm$0.01 & 0.60$\pm$0.01 \\
Posted Price & 0.97$\pm$0.01 & 0.97$\pm$0.01 & 0.97$\pm$0.01 & 0.96$\pm$0.01 & 0.96$\pm$0.01 & 0.97$\pm$0.01 & 0.96$\pm$0.01 & 0.97$\pm$0.01 \\
C-PoP (perfect) & 0.97$\pm$0.00 & 0.96$\pm$0.00 & 0.94$\pm$0.01 & 0.91$\pm$0.01 & 0.89$\pm$0.01 & 0.88$\pm$0.01 & 0.86$\pm$0.01 & 0.84$\pm$0.01 \\
C-PoP (firm size) & 0.97$\pm$0.00 & 0.94$\pm$0.00 & 0.91$\pm$0.01 & 0.86$\pm$0.01 & 0.81$\pm$0.01 & 0.76$\pm$0.01 & 0.70$\pm$0.01 & 0.64$\pm$0.01 \\
C-PoP (random) & 0.95$\pm$0.00 & 0.91$\pm$0.01 & 0.86$\pm$0.01 & 0.80$\pm$0.01 & 0.74$\pm$0.01 & 0.70$\pm$0.01 & 0.65$\pm$0.01 & 0.62$\pm$0.01 \\
C-PoP (50/50) & 0.80$\pm$0.01 & 0.79$\pm$0.01 & 0.76$\pm$0.01 & 0.72$\pm$0.01 & 0.68$\pm$0.01 & 0.66$\pm$0.01 & 0.62$\pm$0.01 & 0.59$\pm$0.01 \\
\bottomrule
\end{tabular}}
\end{minipage}\par\medskip

\par\medskip\noindent\begin{minipage}{\linewidth}
\centering
\captionof{table}{Non-colluder welfare ratio relative to truthful VCG without collusion (mean$\pm$SE) versus the number of bidders --- $r=15$, Uniform distribution}\label{tab:r15-nc-U}
\small
\resizebox{\ifdim\width>\linewidth\linewidth\else\width\fi}{!}{%
\begin{tabular}{lcccccccccc}
\toprule
Mechanism & $T=20$ & $T=30$ & $T=40$ & $T=50$ & $T=60$ & $T=70$ & $T=80$ & $T=90$ & $T=100$ & $T=110$ \\
\midrule
VCG w/ Collusion & 1.08$\pm$0.00 & 1.33$\pm$0.01 & 1.61$\pm$0.02 & 1.86$\pm$0.02 & 2.05$\pm$0.03 & 2.23$\pm$0.03 & 2.29$\pm$0.03 & 2.42$\pm$0.04 & 2.48$\pm$0.04 & 2.54$\pm$0.04 \\
Posted Price & 0.92$\pm$0.01 & 0.77$\pm$0.01 & 0.67$\pm$0.01 & 0.59$\pm$0.01 & 0.52$\pm$0.01 & 0.50$\pm$0.01 & 0.43$\pm$0.01 & 0.42$\pm$0.01 & 0.39$\pm$0.01 & 0.36$\pm$0.01 \\
C-PoP (perfect) & 1.05$\pm$0.00 & 1.16$\pm$0.01 & 1.23$\pm$0.01 & 1.28$\pm$0.01 & 1.31$\pm$0.01 & 1.35$\pm$0.01 & 1.37$\pm$0.01 & 1.44$\pm$0.02 & 1.42$\pm$0.02 & 1.47$\pm$0.02 \\
C-PoP (firm size) & 1.06$\pm$0.00 & 1.25$\pm$0.01 & 1.40$\pm$0.01 & 1.51$\pm$0.02 & 1.57$\pm$0.02 & 1.64$\pm$0.02 & 1.66$\pm$0.02 & 1.74$\pm$0.02 & 1.74$\pm$0.02 & 1.81$\pm$0.02 \\
C-PoP (random) & 0.86$\pm$0.01 & 0.76$\pm$0.01 & 0.69$\pm$0.01 & 0.67$\pm$0.01 & 0.62$\pm$0.01 & 0.61$\pm$0.01 & 0.59$\pm$0.01 & 0.59$\pm$0.01 & 0.56$\pm$0.01 & 0.56$\pm$0.01 \\
C-PoP (50/50) & 0.95$\pm$0.01 & 1.05$\pm$0.01 & 1.13$\pm$0.01 & 1.23$\pm$0.02 & 1.25$\pm$0.02 & 1.30$\pm$0.02 & 1.31$\pm$0.02 & 1.36$\pm$0.02 & 1.34$\pm$0.02 & 1.34$\pm$0.02 \\
\bottomrule
\end{tabular}}
\end{minipage}\par\medskip
\par\medskip\noindent\begin{minipage}{\linewidth}
\centering
\captionof{table}{Welfare ratio relative to truthful VCG without collusion (mean$\pm$SE) versus the number of bidders --- $r=15$, Uniform distribution}\label{tab:r15-wel-U}
\small
\resizebox{\ifdim\width>\linewidth\linewidth\else\width\fi}{!}{%
\begin{tabular}{lcccccccccc}
\toprule
Mechanism & $T=20$ & $T=30$ & $T=40$ & $T=50$ & $T=60$ & $T=70$ & $T=80$ & $T=90$ & $T=100$ & $T=110$ \\
\midrule
VCG w/ Collusion & 0.99$\pm$0.00 & 0.95$\pm$0.00 & 0.92$\pm$0.00 & 0.91$\pm$0.00 & 0.90$\pm$0.00 & 0.90$\pm$0.00 & 0.91$\pm$0.00 & 0.92$\pm$0.00 & 0.93$\pm$0.00 & 0.94$\pm$0.00 \\
Posted Price & 0.97$\pm$0.00 & 0.96$\pm$0.00 & 0.97$\pm$0.00 & 0.97$\pm$0.00 & 0.97$\pm$0.00 & 0.97$\pm$0.00 & 0.98$\pm$0.00 & 0.98$\pm$0.00 & 0.98$\pm$0.00 & 0.98$\pm$0.00 \\
C-PoP (perfect) & 0.99$\pm$0.00 & 0.99$\pm$0.00 & 0.99$\pm$0.00 & 0.99$\pm$0.00 & 0.99$\pm$0.00 & 0.99$\pm$0.00 & 0.99$\pm$0.00 & 0.99$\pm$0.00 & 0.99$\pm$0.00 & 0.99$\pm$0.00 \\
C-PoP (firm size) & 0.99$\pm$0.00 & 0.98$\pm$0.00 & 0.98$\pm$0.00 & 0.98$\pm$0.00 & 0.98$\pm$0.00 & 0.98$\pm$0.00 & 0.98$\pm$0.00 & 0.98$\pm$0.00 & 0.99$\pm$0.00 & 0.99$\pm$0.00 \\
C-PoP (random) & 0.98$\pm$0.00 & 0.97$\pm$0.00 & 0.97$\pm$0.00 & 0.98$\pm$0.00 & 0.98$\pm$0.00 & 0.98$\pm$0.00 & 0.98$\pm$0.00 & 0.99$\pm$0.00 & 0.99$\pm$0.00 & 0.99$\pm$0.00 \\
C-PoP (50/50) & 0.98$\pm$0.00 & 0.97$\pm$0.00 & 0.97$\pm$0.00 & 0.97$\pm$0.00 & 0.97$\pm$0.00 & 0.97$\pm$0.00 & 0.97$\pm$0.00 & 0.98$\pm$0.00 & 0.98$\pm$0.00 & 0.98$\pm$0.00 \\
\bottomrule
\end{tabular}}
\end{minipage}\par\medskip
\par\medskip\noindent\begin{minipage}{\linewidth}
\centering
\captionof{table}{Revenue ratio relative to truthful VCG without collusion (mean$\pm$SE) versus the number of bidders --- $r=15$, Uniform distribution}\label{tab:r15-rev-U}
\small
\resizebox{\ifdim\width>\linewidth\linewidth\else\width\fi}{!}{%
\begin{tabular}{lcccccccccc}
\toprule
Mechanism & $T=20$ & $T=30$ & $T=40$ & $T=50$ & $T=60$ & $T=70$ & $T=80$ & $T=90$ & $T=100$ & $T=110$ \\
\midrule
VCG w/ Collusion & 0.00$\pm$0.00 & 0.03$\pm$0.00 & 0.07$\pm$0.00 & 0.14$\pm$0.01 & 0.21$\pm$0.01 & 0.30$\pm$0.01 & 0.39$\pm$0.01 & 0.46$\pm$0.01 & 0.52$\pm$0.01 & 0.57$\pm$0.01 \\
Posted Price & 1.01$\pm$0.01 & 0.91$\pm$0.00 & 0.91$\pm$0.00 & 0.91$\pm$0.00 & 0.92$\pm$0.00 & 0.92$\pm$0.00 & 0.93$\pm$0.00 & 0.94$\pm$0.00 & 0.94$\pm$0.00 & 0.94$\pm$0.00 \\
C-PoP (perfect) & 0.59$\pm$0.01 & 0.74$\pm$0.00 & 0.81$\pm$0.00 & 0.84$\pm$0.00 & 0.87$\pm$0.00 & 0.89$\pm$0.00 & 0.89$\pm$0.00 & 0.91$\pm$0.00 & 0.91$\pm$0.00 & 0.92$\pm$0.00 \\
C-PoP (firm size) & 0.50$\pm$0.01 & 0.48$\pm$0.00 & 0.56$\pm$0.00 & 0.64$\pm$0.00 & 0.70$\pm$0.00 & 0.74$\pm$0.00 & 0.77$\pm$0.00 & 0.80$\pm$0.00 & 0.82$\pm$0.00 & 0.83$\pm$0.00 \\
C-PoP (random) & 0.54$\pm$0.01 & 0.67$\pm$0.01 & 0.75$\pm$0.00 & 0.77$\pm$0.00 & 0.80$\pm$0.00 & 0.82$\pm$0.00 & 0.83$\pm$0.00 & 0.84$\pm$0.00 & 0.86$\pm$0.00 & 0.86$\pm$0.00 \\
C-PoP (50/50) & 0.49$\pm$0.01 & 0.44$\pm$0.00 & 0.50$\pm$0.00 & 0.55$\pm$0.00 & 0.61$\pm$0.00 & 0.65$\pm$0.00 & 0.69$\pm$0.00 & 0.72$\pm$0.00 & 0.75$\pm$0.00 & 0.77$\pm$0.00 \\
\bottomrule
\end{tabular}}
\end{minipage}\par\medskip
\par\medskip\noindent\begin{minipage}{\linewidth}
\centering
\captionof{table}{Non-colluder welfare ratio relative to truthful VCG without collusion (mean$\pm$SE) versus the number of bidders --- $r=15$, Exponential distribution}\label{tab:r15-nc-E}
\small
\resizebox{\ifdim\width>\linewidth\linewidth\else\width\fi}{!}{%
\begin{tabular}{lcccccccccc}
\toprule
Mechanism & $T=20$ & $T=30$ & $T=40$ & $T=50$ & $T=60$ & $T=70$ & $T=80$ & $T=90$ & $T=100$ & $T=110$ \\
\midrule
VCG w/ Collusion & 1.04$\pm$0.00 & 1.17$\pm$0.01 & 1.30$\pm$0.01 & 1.41$\pm$0.01 & 1.45$\pm$0.01 & 1.54$\pm$0.01 & 1.57$\pm$0.02 & 1.60$\pm$0.02 & 1.62$\pm$0.02 & 1.64$\pm$0.02 \\
Posted Price & 0.94$\pm$0.01 & 0.83$\pm$0.01 & 0.75$\pm$0.01 & 0.70$\pm$0.01 & 0.65$\pm$0.01 & 0.62$\pm$0.01 & 0.62$\pm$0.01 & 0.58$\pm$0.01 & 0.53$\pm$0.01 & 0.55$\pm$0.01 \\
C-PoP (perfect) & 1.03$\pm$0.00 & 1.08$\pm$0.00 & 1.11$\pm$0.01 & 1.15$\pm$0.01 & 1.15$\pm$0.01 & 1.19$\pm$0.01 & 1.21$\pm$0.01 & 1.22$\pm$0.01 & 1.23$\pm$0.01 & 1.22$\pm$0.01 \\
C-PoP (firm size) & 1.03$\pm$0.00 & 1.14$\pm$0.01 & 1.20$\pm$0.01 & 1.27$\pm$0.01 & 1.29$\pm$0.01 & 1.34$\pm$0.01 & 1.37$\pm$0.01 & 1.38$\pm$0.01 & 1.39$\pm$0.01 & 1.40$\pm$0.01 \\
C-PoP (random) & 0.88$\pm$0.01 & 0.70$\pm$0.01 & 0.62$\pm$0.01 & 0.57$\pm$0.01 & 0.53$\pm$0.01 & 0.51$\pm$0.01 & 0.49$\pm$0.01 & 0.50$\pm$0.01 & 0.48$\pm$0.01 & 0.45$\pm$0.01 \\
C-PoP (50/50) & 0.92$\pm$0.01 & 0.94$\pm$0.01 & 0.97$\pm$0.01 & 1.01$\pm$0.01 & 0.98$\pm$0.01 & 1.04$\pm$0.01 & 1.02$\pm$0.01 & 1.03$\pm$0.01 & 1.01$\pm$0.01 & 0.98$\pm$0.01 \\
\bottomrule
\end{tabular}}
\end{minipage}\par\medskip
\par\medskip\noindent\begin{minipage}{\linewidth}
\centering
\captionof{table}{Welfare ratio relative to truthful VCG without collusion (mean$\pm$SE) versus the number of bidders --- $r=15$, Exponential distribution}\label{tab:r15-wel-E}
\small
\resizebox{\ifdim\width>\linewidth\linewidth\else\width\fi}{!}{%
\begin{tabular}{lcccccccccc}
\toprule
Mechanism & $T=20$ & $T=30$ & $T=40$ & $T=50$ & $T=60$ & $T=70$ & $T=80$ & $T=90$ & $T=100$ & $T=110$ \\
\midrule
VCG w/ Collusion & 0.99$\pm$0.00 & 0.97$\pm$0.00 & 0.95$\pm$0.00 & 0.94$\pm$0.00 & 0.94$\pm$0.00 & 0.94$\pm$0.00 & 0.95$\pm$0.00 & 0.95$\pm$0.00 & 0.95$\pm$0.00 & 0.95$\pm$0.00 \\
Posted Price & 0.97$\pm$0.00 & 0.96$\pm$0.00 & 0.94$\pm$0.00 & 0.95$\pm$0.00 & 0.94$\pm$0.00 & 0.94$\pm$0.00 & 0.95$\pm$0.00 & 0.95$\pm$0.00 & 0.95$\pm$0.00 & 0.95$\pm$0.00 \\
C-PoP (perfect) & 1.00$\pm$0.00 & 0.99$\pm$0.00 & 0.99$\pm$0.00 & 0.99$\pm$0.00 & 0.99$\pm$0.00 & 0.99$\pm$0.00 & 0.99$\pm$0.00 & 0.99$\pm$0.00 & 0.99$\pm$0.00 & 0.99$\pm$0.00 \\
C-PoP (firm size) & 1.00$\pm$0.00 & 0.99$\pm$0.00 & 0.98$\pm$0.00 & 0.98$\pm$0.00 & 0.98$\pm$0.00 & 0.98$\pm$0.00 & 0.98$\pm$0.00 & 0.98$\pm$0.00 & 0.98$\pm$0.00 & 0.98$\pm$0.00 \\
C-PoP (random) & 0.98$\pm$0.00 & 0.96$\pm$0.00 & 0.95$\pm$0.00 & 0.95$\pm$0.00 & 0.95$\pm$0.00 & 0.95$\pm$0.00 & 0.95$\pm$0.00 & 0.95$\pm$0.00 & 0.95$\pm$0.00 & 0.96$\pm$0.00 \\
C-PoP (50/50) & 0.98$\pm$0.00 & 0.97$\pm$0.00 & 0.96$\pm$0.00 & 0.96$\pm$0.00 & 0.96$\pm$0.00 & 0.96$\pm$0.00 & 0.96$\pm$0.00 & 0.96$\pm$0.00 & 0.96$\pm$0.00 & 0.96$\pm$0.00 \\
\bottomrule
\end{tabular}}
\end{minipage}\par\medskip
\par\medskip\noindent\begin{minipage}{\linewidth}
\centering
\captionof{table}{Revenue ratio relative to truthful VCG without collusion (mean$\pm$SE) versus the number of bidders --- $r=15$, Exponential distribution}\label{tab:r15-rev-E}
\small
\resizebox{\ifdim\width>\linewidth\linewidth\else\width\fi}{!}{%
\begin{tabular}{lcccccccccc}
\toprule
Mechanism & $T=20$ & $T=30$ & $T=40$ & $T=50$ & $T=60$ & $T=70$ & $T=80$ & $T=90$ & $T=100$ & $T=110$ \\
\midrule
VCG w/ Collusion & 0.01$\pm$0.00 & 0.05$\pm$0.00 & 0.11$\pm$0.00 & 0.18$\pm$0.01 & 0.26$\pm$0.01 & 0.32$\pm$0.01 & 0.37$\pm$0.01 & 0.42$\pm$0.01 & 0.45$\pm$0.00 & 0.47$\pm$0.00 \\
Posted Price & 1.15$\pm$0.01 & 0.96$\pm$0.01 & 0.93$\pm$0.00 & 0.92$\pm$0.00 & 0.91$\pm$0.00 & 0.91$\pm$0.00 & 0.91$\pm$0.00 & 0.91$\pm$0.00 & 0.91$\pm$0.00 & 0.91$\pm$0.00 \\
C-PoP (perfect) & 0.54$\pm$0.01 & 0.69$\pm$0.01 & 0.75$\pm$0.00 & 0.77$\pm$0.00 & 0.79$\pm$0.00 & 0.79$\pm$0.00 & 0.81$\pm$0.00 & 0.82$\pm$0.00 & 0.82$\pm$0.00 & 0.83$\pm$0.00 \\
C-PoP (firm size) & 0.48$\pm$0.01 & 0.40$\pm$0.00 & 0.47$\pm$0.00 & 0.52$\pm$0.00 & 0.57$\pm$0.00 & 0.60$\pm$0.00 & 0.62$\pm$0.00 & 0.65$\pm$0.00 & 0.66$\pm$0.00 & 0.68$\pm$0.00 \\
C-PoP (random) & 0.49$\pm$0.01 & 0.62$\pm$0.01 & 0.65$\pm$0.00 & 0.66$\pm$0.00 & 0.68$\pm$0.00 & 0.69$\pm$0.00 & 0.69$\pm$0.00 & 0.69$\pm$0.00 & 0.70$\pm$0.00 & 0.71$\pm$0.00 \\
C-PoP (50/50) & 0.45$\pm$0.01 & 0.37$\pm$0.00 & 0.40$\pm$0.00 & 0.43$\pm$0.00 & 0.47$\pm$0.00 & 0.49$\pm$0.00 & 0.52$\pm$0.00 & 0.54$\pm$0.00 & 0.56$\pm$0.00 & 0.58$\pm$0.00 \\
\bottomrule
\end{tabular}}
\end{minipage}\par\medskip

\par\medskip\noindent\begin{minipage}{\linewidth}
\centering
\captionof{table}{Non-colluder welfare ratio of C-PoP by false-positive (columns) and false-negative (rows) rate (mean$\pm$SE) --- $r=10$, $T=110$, 84\% colluders, Uniform distribution}\label{tab:err84-nc-U}
\small
\resizebox{\ifdim\width>\linewidth\linewidth\else\width\fi}{!}{%
\begin{tabular}{lcccccc}
\toprule
 & FP $=0.0$ & FP $=0.1$ & FP $=0.2$ & FP $=0.3$ & FP $=0.4$ & FP $=0.5$ \\
\midrule
FN $=0.0$ & 1.60$\pm$0.02 & 1.48$\pm$0.02 & 1.36$\pm$0.02 & 1.23$\pm$0.02 & 1.08$\pm$0.02 & 0.99$\pm$0.02 \\
FN $=0.1$ & 1.66$\pm$0.03 & 1.54$\pm$0.03 & 1.44$\pm$0.03 & 1.29$\pm$0.02 & 1.14$\pm$0.02 & 0.99$\pm$0.02 \\
FN $=0.2$ & 1.71$\pm$0.03 & 1.66$\pm$0.03 & 1.48$\pm$0.03 & 1.38$\pm$0.03 & 1.22$\pm$0.02 & 1.07$\pm$0.02 \\
FN $=0.3$ & 1.88$\pm$0.03 & 1.76$\pm$0.03 & 1.61$\pm$0.03 & 1.44$\pm$0.03 & 1.29$\pm$0.02 & 1.14$\pm$0.02 \\
FN $=0.4$ & 1.92$\pm$0.03 & 1.89$\pm$0.03 & 1.69$\pm$0.03 & 1.58$\pm$0.03 & 1.36$\pm$0.03 & 1.25$\pm$0.03 \\
FN $=0.5$ & 2.06$\pm$0.04 & 1.94$\pm$0.04 & 1.76$\pm$0.03 & 1.68$\pm$0.03 & 1.52$\pm$0.03 & 1.36$\pm$0.03 \\
\bottomrule
\end{tabular}}
\end{minipage}\par\medskip
\par\medskip\noindent\begin{minipage}{\linewidth}
\centering
\captionof{table}{Welfare ratio of C-PoP by false-positive (columns) and false-negative (rows) rate (mean$\pm$SE) --- $r=10$, $T=110$, 84\% colluders, Uniform distribution}\label{tab:err84-wel-U}
\small
\resizebox{\ifdim\width>\linewidth\linewidth\else\width\fi}{!}{%
\begin{tabular}{lcccccc}
\toprule
 & FP $=0.0$ & FP $=0.1$ & FP $=0.2$ & FP $=0.3$ & FP $=0.4$ & FP $=0.5$ \\
\midrule
FN $=0.0$ & 0.99$\pm$0.00 & 0.99$\pm$0.00 & 0.99$\pm$0.00 & 0.99$\pm$0.00 & 0.99$\pm$0.00 & 0.99$\pm$0.00 \\
FN $=0.1$ & 0.99$\pm$0.00 & 0.99$\pm$0.00 & 0.99$\pm$0.00 & 0.99$\pm$0.00 & 0.99$\pm$0.00 & 0.99$\pm$0.00 \\
FN $=0.2$ & 0.99$\pm$0.00 & 0.99$\pm$0.00 & 0.99$\pm$0.00 & 0.99$\pm$0.00 & 0.99$\pm$0.00 & 0.99$\pm$0.00 \\
FN $=0.3$ & 0.99$\pm$0.00 & 0.99$\pm$0.00 & 0.99$\pm$0.00 & 0.99$\pm$0.00 & 0.99$\pm$0.00 & 0.99$\pm$0.00 \\
FN $=0.4$ & 0.99$\pm$0.00 & 0.99$\pm$0.00 & 0.99$\pm$0.00 & 0.99$\pm$0.00 & 0.99$\pm$0.00 & 0.99$\pm$0.00 \\
FN $=0.5$ & 0.99$\pm$0.00 & 0.98$\pm$0.00 & 0.99$\pm$0.00 & 0.98$\pm$0.00 & 0.98$\pm$0.00 & 0.98$\pm$0.00 \\
\bottomrule
\end{tabular}}
\end{minipage}\par\medskip
\par\medskip\noindent\begin{minipage}{\linewidth}
\centering
\captionof{table}{Revenue ratio of C-PoP by false-positive (columns) and false-negative (rows) rate (mean$\pm$SE) --- $r=10$, $T=110$, 84\% colluders, Uniform distribution}\label{tab:err84-rev-U}
\small
\resizebox{\ifdim\width>\linewidth\linewidth\else\width\fi}{!}{%
\begin{tabular}{lcccccc}
\toprule
 & FP $=0.0$ & FP $=0.1$ & FP $=0.2$ & FP $=0.3$ & FP $=0.4$ & FP $=0.5$ \\
\midrule
FN $=0.0$ & 0.93$\pm$0.00 & 0.93$\pm$0.00 & 0.93$\pm$0.00 & 0.93$\pm$0.00 & 0.93$\pm$0.00 & 0.92$\pm$0.00 \\
FN $=0.1$ & 0.92$\pm$0.00 & 0.92$\pm$0.00 & 0.92$\pm$0.00 & 0.91$\pm$0.00 & 0.91$\pm$0.00 & 0.91$\pm$0.00 \\
FN $=0.2$ & 0.91$\pm$0.00 & 0.90$\pm$0.00 & 0.90$\pm$0.00 & 0.90$\pm$0.00 & 0.90$\pm$0.00 & 0.89$\pm$0.00 \\
FN $=0.3$ & 0.89$\pm$0.00 & 0.89$\pm$0.00 & 0.89$\pm$0.00 & 0.88$\pm$0.00 & 0.88$\pm$0.00 & 0.87$\pm$0.00 \\
FN $=0.4$ & 0.87$\pm$0.00 & 0.87$\pm$0.00 & 0.87$\pm$0.00 & 0.86$\pm$0.00 & 0.86$\pm$0.00 & 0.85$\pm$0.00 \\
FN $=0.5$ & 0.85$\pm$0.00 & 0.85$\pm$0.00 & 0.84$\pm$0.00 & 0.83$\pm$0.00 & 0.83$\pm$0.00 & 0.82$\pm$0.00 \\
\bottomrule
\end{tabular}}
\end{minipage}\par\medskip
\par\medskip\noindent\begin{minipage}{\linewidth}
\centering
\captionof{table}{Non-colluder welfare ratio of C-PoP by false-positive (columns) and false-negative (rows) rate (mean$\pm$SE) --- $r=10$, $T=110$, 84\% colluders, Exponential distribution}\label{tab:err84-nc-E}
\small
\resizebox{\ifdim\width>\linewidth\linewidth\else\width\fi}{!}{%
\begin{tabular}{lcccccc}
\toprule
 & FP $=0.0$ & FP $=0.1$ & FP $=0.2$ & FP $=0.3$ & FP $=0.4$ & FP $=0.5$ \\
\midrule
FN $=0.0$ & 1.30$\pm$0.01 & 1.23$\pm$0.02 & 1.09$\pm$0.02 & 1.00$\pm$0.02 & 0.89$\pm$0.02 & 0.80$\pm$0.02 \\
FN $=0.1$ & 1.37$\pm$0.02 & 1.28$\pm$0.02 & 1.15$\pm$0.02 & 1.04$\pm$0.02 & 0.91$\pm$0.02 & 0.80$\pm$0.02 \\
FN $=0.2$ & 1.40$\pm$0.02 & 1.29$\pm$0.02 & 1.20$\pm$0.02 & 1.07$\pm$0.02 & 0.94$\pm$0.02 & 0.84$\pm$0.02 \\
FN $=0.3$ & 1.46$\pm$0.02 & 1.37$\pm$0.02 & 1.24$\pm$0.02 & 1.14$\pm$0.02 & 1.04$\pm$0.02 & 0.87$\pm$0.02 \\
FN $=0.4$ & 1.51$\pm$0.02 & 1.41$\pm$0.02 & 1.28$\pm$0.02 & 1.17$\pm$0.02 & 1.07$\pm$0.02 & 0.94$\pm$0.02 \\
FN $=0.5$ & 1.52$\pm$0.02 & 1.46$\pm$0.02 & 1.34$\pm$0.02 & 1.22$\pm$0.02 & 1.09$\pm$0.02 & 1.01$\pm$0.02 \\
\bottomrule
\end{tabular}}
\end{minipage}\par\medskip
\par\medskip\noindent\begin{minipage}{\linewidth}
\centering
\captionof{table}{Welfare ratio of C-PoP by false-positive (columns) and false-negative (rows) rate (mean$\pm$SE) --- $r=10$, $T=110$, 84\% colluders, Exponential distribution}\label{tab:err84-wel-E}
\small
\resizebox{\ifdim\width>\linewidth\linewidth\else\width\fi}{!}{%
\begin{tabular}{lcccccc}
\toprule
 & FP $=0.0$ & FP $=0.1$ & FP $=0.2$ & FP $=0.3$ & FP $=0.4$ & FP $=0.5$ \\
\midrule
FN $=0.0$ & 0.99$\pm$0.00 & 0.98$\pm$0.00 & 0.98$\pm$0.00 & 0.98$\pm$0.00 & 0.97$\pm$0.00 & 0.97$\pm$0.00 \\
FN $=0.1$ & 0.99$\pm$0.00 & 0.99$\pm$0.00 & 0.98$\pm$0.00 & 0.98$\pm$0.00 & 0.97$\pm$0.00 & 0.97$\pm$0.00 \\
FN $=0.2$ & 0.99$\pm$0.00 & 0.98$\pm$0.00 & 0.98$\pm$0.00 & 0.98$\pm$0.00 & 0.97$\pm$0.00 & 0.97$\pm$0.00 \\
FN $=0.3$ & 0.99$\pm$0.00 & 0.98$\pm$0.00 & 0.98$\pm$0.00 & 0.97$\pm$0.00 & 0.97$\pm$0.00 & 0.97$\pm$0.00 \\
FN $=0.4$ & 0.98$\pm$0.00 & 0.98$\pm$0.00 & 0.98$\pm$0.00 & 0.97$\pm$0.00 & 0.97$\pm$0.00 & 0.96$\pm$0.00 \\
FN $=0.5$ & 0.98$\pm$0.00 & 0.97$\pm$0.00 & 0.97$\pm$0.00 & 0.97$\pm$0.00 & 0.96$\pm$0.00 & 0.96$\pm$0.00 \\
\bottomrule
\end{tabular}}
\end{minipage}\par\medskip
\par\medskip\noindent\begin{minipage}{\linewidth}
\centering
\captionof{table}{Revenue ratio of C-PoP by false-positive (columns) and false-negative (rows) rate (mean$\pm$SE) --- $r=10$, $T=110$, 84\% colluders, Exponential distribution}\label{tab:err84-rev-E}
\small
\resizebox{\ifdim\width>\linewidth\linewidth\else\width\fi}{!}{%
\begin{tabular}{lcccccc}
\toprule
 & FP $=0.0$ & FP $=0.1$ & FP $=0.2$ & FP $=0.3$ & FP $=0.4$ & FP $=0.5$ \\
\midrule
FN $=0.0$ & 0.83$\pm$0.00 & 0.82$\pm$0.00 & 0.81$\pm$0.00 & 0.82$\pm$0.00 & 0.80$\pm$0.00 & 0.80$\pm$0.00 \\
FN $=0.1$ & 0.80$\pm$0.00 & 0.79$\pm$0.00 & 0.79$\pm$0.00 & 0.79$\pm$0.00 & 0.77$\pm$0.00 & 0.77$\pm$0.00 \\
FN $=0.2$ & 0.77$\pm$0.00 & 0.77$\pm$0.00 & 0.76$\pm$0.00 & 0.75$\pm$0.00 & 0.74$\pm$0.00 & 0.73$\pm$0.00 \\
FN $=0.3$ & 0.74$\pm$0.00 & 0.73$\pm$0.00 & 0.72$\pm$0.00 & 0.72$\pm$0.00 & 0.71$\pm$0.00 & 0.70$\pm$0.00 \\
FN $=0.4$ & 0.71$\pm$0.00 & 0.70$\pm$0.00 & 0.69$\pm$0.00 & 0.68$\pm$0.00 & 0.67$\pm$0.00 & 0.65$\pm$0.00 \\
FN $=0.5$ & 0.67$\pm$0.00 & 0.66$\pm$0.00 & 0.65$\pm$0.00 & 0.64$\pm$0.00 & 0.63$\pm$0.00 & 0.62$\pm$0.00 \\
\bottomrule
\end{tabular}}
\end{minipage}\par\medskip

\par\medskip\noindent\begin{minipage}{\linewidth}
\centering
\captionof{table}{Non-colluder welfare ratio of C-PoP by false-positive (columns) and false-negative (rows) rate (mean$\pm$SE) --- $r=10$, $T=110$, 50\% colluders, Uniform distribution}\label{tab:err50-nc-U}
\small
\resizebox{\ifdim\width>\linewidth\linewidth\else\width\fi}{!}{%
\begin{tabular}{lcccccc}
\toprule
 & FP $=0.0$ & FP $=0.1$ & FP $=0.2$ & FP $=0.3$ & FP $=0.4$ & FP $=0.5$ \\
\midrule
FN $=0.0$ & 1.34$\pm$0.01 & 1.25$\pm$0.01 & 1.16$\pm$0.01 & 1.07$\pm$0.01 & 0.98$\pm$0.01 & 0.89$\pm$0.01 \\
FN $=0.1$ & 1.34$\pm$0.01 & 1.25$\pm$0.01 & 1.17$\pm$0.01 & 1.07$\pm$0.01 & 0.98$\pm$0.01 & 0.91$\pm$0.01 \\
FN $=0.2$ & 1.33$\pm$0.01 & 1.22$\pm$0.01 & 1.17$\pm$0.01 & 1.09$\pm$0.01 & 0.98$\pm$0.01 & 0.91$\pm$0.01 \\
FN $=0.3$ & 1.31$\pm$0.01 & 1.25$\pm$0.01 & 1.16$\pm$0.01 & 1.08$\pm$0.01 & 1.01$\pm$0.01 & 0.93$\pm$0.01 \\
FN $=0.4$ & 1.31$\pm$0.01 & 1.23$\pm$0.01 & 1.17$\pm$0.01 & 1.10$\pm$0.01 & 1.03$\pm$0.01 & 0.96$\pm$0.01 \\
FN $=0.5$ & 1.31$\pm$0.01 & 1.24$\pm$0.01 & 1.17$\pm$0.01 & 1.13$\pm$0.01 & 1.05$\pm$0.01 & 0.99$\pm$0.01 \\
\bottomrule
\end{tabular}}
\end{minipage}\par\medskip
\par\medskip\noindent\begin{minipage}{\linewidth}
\centering
\captionof{table}{Welfare ratio of C-PoP by false-positive (columns) and false-negative (rows) rate (mean$\pm$SE) --- $r=10$, $T=110$, 50\% colluders, Uniform distribution}\label{tab:err50-wel-U}
\small
\resizebox{\ifdim\width>\linewidth\linewidth\else\width\fi}{!}{%
\begin{tabular}{lcccccc}
\toprule
 & FP $=0.0$ & FP $=0.1$ & FP $=0.2$ & FP $=0.3$ & FP $=0.4$ & FP $=0.5$ \\
\midrule
FN $=0.0$ & 0.99$\pm$0.00 & 0.99$\pm$0.00 & 0.99$\pm$0.00 & 0.98$\pm$0.00 & 0.98$\pm$0.00 & 0.98$\pm$0.00 \\
FN $=0.1$ & 0.99$\pm$0.00 & 0.99$\pm$0.00 & 0.99$\pm$0.00 & 0.99$\pm$0.00 & 0.98$\pm$0.00 & 0.98$\pm$0.00 \\
FN $=0.2$ & 0.99$\pm$0.00 & 0.99$\pm$0.00 & 0.99$\pm$0.00 & 0.99$\pm$0.00 & 0.98$\pm$0.00 & 0.98$\pm$0.00 \\
FN $=0.3$ & 0.99$\pm$0.00 & 0.99$\pm$0.00 & 0.99$\pm$0.00 & 0.99$\pm$0.00 & 0.98$\pm$0.00 & 0.98$\pm$0.00 \\
FN $=0.4$ & 0.99$\pm$0.00 & 0.99$\pm$0.00 & 0.99$\pm$0.00 & 0.99$\pm$0.00 & 0.98$\pm$0.00 & 0.98$\pm$0.00 \\
FN $=0.5$ & 0.99$\pm$0.00 & 0.99$\pm$0.00 & 0.99$\pm$0.00 & 0.99$\pm$0.00 & 0.98$\pm$0.00 & 0.98$\pm$0.00 \\
\bottomrule
\end{tabular}}
\end{minipage}\par\medskip
\par\medskip\noindent\begin{minipage}{\linewidth}
\centering
\captionof{table}{Revenue ratio of C-PoP by false-positive (columns) and false-negative (rows) rate (mean$\pm$SE) --- $r=10$, $T=110$, 50\% colluders, Uniform distribution}\label{tab:err50-rev-U}
\small
\resizebox{\ifdim\width>\linewidth\linewidth\else\width\fi}{!}{%
\begin{tabular}{lcccccc}
\toprule
 & FP $=0.0$ & FP $=0.1$ & FP $=0.2$ & FP $=0.3$ & FP $=0.4$ & FP $=0.5$ \\
\midrule
FN $=0.0$ & 0.96$\pm$0.00 & 0.95$\pm$0.00 & 0.95$\pm$0.00 & 0.95$\pm$0.00 & 0.94$\pm$0.00 & 0.94$\pm$0.00 \\
FN $=0.1$ & 0.96$\pm$0.00 & 0.95$\pm$0.00 & 0.95$\pm$0.00 & 0.94$\pm$0.00 & 0.94$\pm$0.00 & 0.94$\pm$0.00 \\
FN $=0.2$ & 0.96$\pm$0.00 & 0.95$\pm$0.00 & 0.95$\pm$0.00 & 0.94$\pm$0.00 & 0.93$\pm$0.00 & 0.93$\pm$0.00 \\
FN $=0.3$ & 0.95$\pm$0.00 & 0.95$\pm$0.00 & 0.94$\pm$0.00 & 0.94$\pm$0.00 & 0.93$\pm$0.00 & 0.92$\pm$0.00 \\
FN $=0.4$ & 0.95$\pm$0.00 & 0.94$\pm$0.00 & 0.94$\pm$0.00 & 0.93$\pm$0.00 & 0.92$\pm$0.00 & 0.92$\pm$0.00 \\
FN $=0.5$ & 0.95$\pm$0.00 & 0.94$\pm$0.00 & 0.94$\pm$0.00 & 0.93$\pm$0.00 & 0.92$\pm$0.00 & 0.91$\pm$0.00 \\
\bottomrule
\end{tabular}}
\end{minipage}\par\medskip
\par\medskip\noindent\begin{minipage}{\linewidth}
\centering
\captionof{table}{Non-colluder welfare ratio of C-PoP by false-positive (columns) and false-negative (rows) rate (mean$\pm$SE) --- $r=10$, $T=110$, 50\% colluders, Exponential distribution}\label{tab:err50-nc-E}
\small
\resizebox{\ifdim\width>\linewidth\linewidth\else\width\fi}{!}{%
\begin{tabular}{lcccccc}
\toprule
 & FP $=0.0$ & FP $=0.1$ & FP $=0.2$ & FP $=0.3$ & FP $=0.4$ & FP $=0.5$ \\
\midrule
FN $=0.0$ & 1.16$\pm$0.01 & 1.06$\pm$0.01 & 0.99$\pm$0.01 & 0.92$\pm$0.01 & 0.85$\pm$0.01 & 0.77$\pm$0.01 \\
FN $=0.1$ & 1.15$\pm$0.01 & 1.07$\pm$0.01 & 1.01$\pm$0.01 & 0.92$\pm$0.01 & 0.85$\pm$0.01 & 0.79$\pm$0.01 \\
FN $=0.2$ & 1.16$\pm$0.01 & 1.07$\pm$0.01 & 1.00$\pm$0.01 & 0.93$\pm$0.01 & 0.86$\pm$0.01 & 0.78$\pm$0.01 \\
FN $=0.3$ & 1.15$\pm$0.01 & 1.08$\pm$0.01 & 1.02$\pm$0.01 & 0.93$\pm$0.01 & 0.86$\pm$0.01 & 0.79$\pm$0.01 \\
FN $=0.4$ & 1.16$\pm$0.01 & 1.09$\pm$0.01 & 1.02$\pm$0.01 & 0.96$\pm$0.01 & 0.86$\pm$0.01 & 0.80$\pm$0.01 \\
FN $=0.5$ & 1.17$\pm$0.01 & 1.09$\pm$0.01 & 1.02$\pm$0.01 & 0.96$\pm$0.01 & 0.88$\pm$0.01 & 0.83$\pm$0.01 \\
\bottomrule
\end{tabular}}
\end{minipage}\par\medskip
\par\medskip\noindent\begin{minipage}{\linewidth}
\centering
\captionof{table}{Welfare ratio of C-PoP by false-positive (columns) and false-negative (rows) rate (mean$\pm$SE) --- $r=10$, $T=110$, 50\% colluders, Exponential distribution}\label{tab:err50-wel-E}
\small
\resizebox{\ifdim\width>\linewidth\linewidth\else\width\fi}{!}{%
\begin{tabular}{lcccccc}
\toprule
 & FP $=0.0$ & FP $=0.1$ & FP $=0.2$ & FP $=0.3$ & FP $=0.4$ & FP $=0.5$ \\
\midrule
FN $=0.0$ & 0.98$\pm$0.00 & 0.97$\pm$0.00 & 0.96$\pm$0.00 & 0.95$\pm$0.00 & 0.94$\pm$0.00 & 0.93$\pm$0.00 \\
FN $=0.1$ & 0.98$\pm$0.00 & 0.97$\pm$0.00 & 0.96$\pm$0.00 & 0.95$\pm$0.00 & 0.94$\pm$0.00 & 0.93$\pm$0.00 \\
FN $=0.2$ & 0.98$\pm$0.00 & 0.97$\pm$0.00 & 0.96$\pm$0.00 & 0.95$\pm$0.00 & 0.94$\pm$0.00 & 0.92$\pm$0.00 \\
FN $=0.3$ & 0.99$\pm$0.00 & 0.97$\pm$0.00 & 0.96$\pm$0.00 & 0.95$\pm$0.00 & 0.94$\pm$0.00 & 0.92$\pm$0.00 \\
FN $=0.4$ & 0.99$\pm$0.00 & 0.97$\pm$0.00 & 0.96$\pm$0.00 & 0.95$\pm$0.00 & 0.93$\pm$0.00 & 0.92$\pm$0.00 \\
FN $=0.5$ & 0.99$\pm$0.00 & 0.98$\pm$0.00 & 0.96$\pm$0.00 & 0.95$\pm$0.00 & 0.93$\pm$0.00 & 0.92$\pm$0.00 \\
\bottomrule
\end{tabular}}
\end{minipage}\par\medskip
\par\medskip\noindent\begin{minipage}{\linewidth}
\centering
\captionof{table}{Revenue ratio of C-PoP by false-positive (columns) and false-negative (rows) rate (mean$\pm$SE) --- $r=10$, $T=110$, 50\% colluders, Exponential distribution}\label{tab:err50-rev-E}
\small
\resizebox{\ifdim\width>\linewidth\linewidth\else\width\fi}{!}{%
\begin{tabular}{lcccccc}
\toprule
 & FP $=0.0$ & FP $=0.1$ & FP $=0.2$ & FP $=0.3$ & FP $=0.4$ & FP $=0.5$ \\
\midrule
FN $=0.0$ & 0.90$\pm$0.00 & 0.89$\pm$0.00 & 0.88$\pm$0.00 & 0.87$\pm$0.00 & 0.86$\pm$0.00 & 0.85$\pm$0.00 \\
FN $=0.1$ & 0.89$\pm$0.00 & 0.88$\pm$0.00 & 0.87$\pm$0.00 & 0.86$\pm$0.00 & 0.85$\pm$0.00 & 0.83$\pm$0.00 \\
FN $=0.2$ & 0.88$\pm$0.00 & 0.87$\pm$0.00 & 0.86$\pm$0.00 & 0.84$\pm$0.00 & 0.84$\pm$0.00 & 0.82$\pm$0.00 \\
FN $=0.3$ & 0.88$\pm$0.00 & 0.86$\pm$0.00 & 0.85$\pm$0.00 & 0.83$\pm$0.00 & 0.82$\pm$0.00 & 0.81$\pm$0.00 \\
FN $=0.4$ & 0.87$\pm$0.00 & 0.86$\pm$0.00 & 0.84$\pm$0.00 & 0.83$\pm$0.00 & 0.81$\pm$0.00 & 0.79$\pm$0.00 \\
FN $=0.5$ & 0.87$\pm$0.00 & 0.86$\pm$0.00 & 0.83$\pm$0.00 & 0.81$\pm$0.00 & 0.79$\pm$0.00 & 0.78$\pm$0.00 \\
\bottomrule
\end{tabular}}
\end{minipage}\par\medskip
\clearpage   

\end{document}